\documentclass[11pt, dvipsnames]{article}
\usepackage[dvipsnames, table]{xcolor}

\def\bSig\bm{\Sigma}

\usepackage{amsmath, amssymb}
\usepackage[colorlinks,citecolor=blue, urlcolor=blue, linkcolor = blue, breaklinks=true]{hyperref}
\usepackage{tikz,pgf}
\usepackage{multirow, multicol}
\usepackage{float}
\usepackage{amsfonts}
\usepackage{graphicx,centernot}
\usepackage{enumerate, subcaption, natbib, comment}
\usepackage[font=singlespacing]{caption}
\newcommand{\bigCI}{\mathrel{\text{\scalebox{1.07}{$\perp\mkern-10mu\perp$}}}}

\usepackage{amsthm}
\newtheorem{theorem}{Theorem}

\newtheorem{proposition}{Proposition}
\newtheorem{lemma}{Lemma}

\usepackage{algorithm}
\usepackage[noend]{algpseudocode}
\def\IF{\textbf{IF}}
\makeatletter
\def\BState{\State\hskip-\ALG@thistlm}
\usepackage{breakcites}
\makeatother
\usepackage{enumitem, bm}
\newcommand{\blind}{1}

\usepackage[margin=1.0in]{geometry}

\begin{document}

\def\spacingset#1{\renewcommand{\baselinestretch}%
{#1}\small\normalsize} \spacingset{1}


\if1\blind
{
  \title{Doubly Robust Nonparametric Instrumental Variable Estimators for Survival Outcomes}
  \author{Youjin Lee\thanks{Center for Causal Inference, University of Pennsylvania. Email: youjin.lee@pennmedicine.upenn.edu},~Edward H. Kennedy\thanks{Department of Statistics, Carnegie Mellon University. Email: edward@stat.cmu.edu}, and Nandita Mitra\thanks{Department of Biostatistics, Epidemiology and Informatics, University of Pennsylvania. Email: nanditam@pennmedicine.upenn.edu} }
    \date{}
  \maketitle
} \fi

\if0\blind
{
  \bigskip
  \bigskip
  \bigskip
  \begin{center}
    {\LARGE\bf Title}
\end{center}
  \medskip
} \fi



\begin{abstract}
Instrumental variable (IV) methods allow us the opportunity to address unmeasured confounding in causal inference. However, most IV methods are only applicable to discrete or continuous outcomes with very few IV methods for censored survival outcomes. In this work we propose nonparametric estimators for the local average treatment effect on survival probabilities under both nonignorable and ignorable censoring. We provide an efficient influence function-based estimator and a simple estimation procedure when the IV is either binary or continuous. The proposed estimators possess double-robustness properties and can easily incorporate nonparametric estimation using machine learning tools. In simulation studies, we demonstrate the flexibility and efficiency of our proposed estimators under various plausible scenarios. We apply our method to the Prostate, Lung, Colorectal, and Ovarian Cancer Screening Trial for estimating the causal effect of screening on survival probabilities and investigate the causal contrasts between the two interventions under different censoring assumptions.
\end{abstract}

%

\noindent%
\vfill


\maketitle

\section{Introduction}
\label{s:intro}

Over the past few decades, we have seen an enormous growth in applications of survival analysis in clinical and epidemiological research that examine the association between the time to an event of interest and exposure. More recently, there has been more focus on the causal interpretation of survival comparisons beyond associational findings to better understand underlying causal mechanisms and to connect statistical findings to public policy making. To this end, causal inference methods such as propensity score weighting for survival outcomes have been developed~\citep{cole2004adjusted, cheng2012estimating, austin2014use, andersen2017causal}. However, these approaches do not directly address causal effect estimation in the presence of unmeasured confounding. 

Instrumental variable (IV) approaches have been widely adopted in medical and epidemiologic research, as well as in economics and the social sciences, to account for unmeasured confounding~\citep{angrist1996identification, hernan2006instruments, baiocchi2014instrumental}. 
However, even if one obtains a valid and strong instrument, standard IV approaches, such as two-stage least squares (2SLS)
may only be valid in the context of linear models or other collapsible models. When outcomes are right censored survival times, it is not straightforward to apply 2SLS methods when using popular survival models such as the Cox proportional hazards model.  For instance, the IV identification assumptions conflict with the underlying proportionality assumption of a Cox model, and treatment effects are often not collapsed into a hazard ratio~\citep{li2015instrumental, wan2018general, dukes2019doubly}. 

Recently, IV approaches using additive hazard models, which have a similar form to linear models and have interpretable coefficients,  have been developed~\citep{tchetgen2015instrumental, li2015instrumental, dukes2019doubly, brueckner2019instrumental}. \cite{li2015instrumental} directly considered a linear structural equation model for hazards and estimated a consistent estimator for the associated coefficients.
\cite{tchetgen2015instrumental}, extending \cite{terza2008two}'s approach, proposed a so-called control function approach that includes the residuals obtained from the first-stage regression in the second stage model as a proxy of unmeasured confounders.
Recently,~\cite{brueckner2019instrumental} further considered residual inclusion in a semi-parametric additive hazard model where the treatment effect could be time-dependent, and
\cite{dukes2019doubly} have proposed doubly robust estimators for the hazard difference that have flexible properties. 

Cox proportional models have also been considered in the IV setting.  \cite{mackenzie2014using} adapted a Cox model with an additive term to represent unmeasured confounding, and proposed an estimating equation based on this particular model structure. \cite{kianian2019causal} proposed the weighting scheme to identify the local causal effect that is applicable to instrumental variable estimation of proportional hazards models. \cite{martinez2019adjusting} proposed using two-stage residual inclusion with an individual frailty term to handle the non-linear structure of the Cox model.

Other parametric and semi-parametric survival models have been also considered in the IV setting~(e.g.,~\cite{li2015bayesian, huling2019instrumental, wan2015bias}). However,  these model choices were mostly chosen to avoid theoretical difficulties rather than a priori knowledge of survival distributions.
In this paper, we aim to go beyond specifying a particular survival model. 

In this work, we develop a nonparametric IV-based method for estimating the causal effect of a binary treatment on survival probability in the presence of unmeasured confounding. Our nonparametric estimator is based on semiparametric theory and influence function-based inference, which allows the estimator to be robust and efficient while still enjoying flexibility in estimating survival and censoring distributions. 
In Section~\ref{sec:setting}, we introduce notation, data structures, and assumptions. In Section~\ref{sec:propose} we propose nonparametric estimators under two different censoring assumptions 
and demonstrate their asymptotic properties in Section~\ref{sec:asymptotic}. The proposed method is then compared to other methods through simulation studies in Section~\ref{sec:simulation}. We apply our method to a  colorectal cancer screening study using data from the Prostate, Lung, Colorectal, and Ovarian (PLCO) Cancer Screening Trial in  Section~\ref{sec:application}.

\section{Setting}
\label{sec:setting}
\subsection{Notation and data structure}
\label{ssec:notation}

Consider an instrumental variable $Z$, either binary $Z \in \{0,1\}$ or continuous $Z \in \mathbb{R}$, a binary treatment $A \in \{0,1\}$, and baseline covariates $\bm{X}_{0} \in \mathbb{R}^{q}$. We denote the time to event (e.g., time to death or time to discharge from hospital) $T$ as our primary outcome of interest.
Let $Y_{t} := \mathbb{I}(T > t)$ be an indicator of the event not occurring before time $t \in [0, \tau]$, so $E\left[Y_{t}\right] = E\left[\mathbb{I}(T > t) \right]$ is the survival probability at time $t$.

Time to event data are often subject to right censoring (such as subjects lost to follow-up prior to observing the event or those who have not experienced the event by end-of-study). In this setting, with $C$ denoting the time at censoring, we observe only one of $T$ and $C$ for each subject, whichever comes first. Let $R = \mathbb{I}( T < C )$ indicate whether a subject's event precedes censoring.
In summary we have $n$, independent and identically distributed (i.i.d.) observations of $\mathcal{O} = \left( \bm{X}_{0}, Z, A, R Y_{1},  R Y_{2}, \ldots,  R Y_{\tau} \right)$: $\{ \mathcal{O}_{1}, \mathcal{O}_{2}, \ldots, \mathcal{O}_{n} \} \overset{i.i.d.}{\sim} \mathbb{P}$.
The fact that $Y_{t}$ may not be observable if $R = 0$ restricts the identification of a causal effect. 
Further,  any remaining unmeasured confounding after adjusting for $\bm{X}_{0}$ may bias the causal effect of treatment $A$ on $Y_{t}$. 
We will address unmeasured confounding via an instrument to  estimate a local causal effect.

The representation of the observed failure times through $\{ R Y_{t} \}$, however, may lose some information on the censoring time $C$ that is completely observable if $R=0$ or partially observable (e.g., lower bound of $C$) if $R=1$. In a way that does not miss any information on the observed censoring time, we can also view the observed data as i.i.d. observations of $\mathcal{O}^{\prime}$ = $\left( \bm{X}_{0}, Z, A, \min (T, C), R \right)$. In Section~\ref{ssec:hazards}, we will develop the causal estimator based on the latter representation, too.

\subsection{Causal estimands}

Our target estimand is the causal difference in survival probabilities under two different treatment arms ($A=1$ vs. $A=0$). This, for example, would help to answer the question: what is the difference in the chance of surviving beyond five years after receiving adjuvant chemotherapy versus radiation?  
Even though treatment effects on survival are most commonly evaluated using hazard ratios estimated from a Cox proportional hazards model, a direct contrast of survival probabilities is often of primary relevance to clinicians and patients. 

We use a potential outcomes framework ~\citep{neyman1923application, rubin1974estimating} to formally define our causal estimand. Let us first consider a binary instrument $Z \in \{ 0, 1 \}$. 
A potential outcome denoted by $Y^{A = a}_{t}$ refers to the outcome $Y_{t}$ that would be observed if the subject takes the treatment ($a=1$) or the control ($a=0$); and $A^{Z=z}$  refers to the treatment value $A$ that would be observed when the subject is assigned the instrument $(z=1)$ or not $(z=0)$. Equation~\eqref{eq:psi_b} formally introduces our target estimand, implying a \textit{causal} difference in survival probability at time $t$ between  subjects under  treatment ($A=1$) versus  control ($A=0$) for those who would take a treatment if and only if they were assigned the instrument ($1 = A^{Z=1} > A^{Z=0} = 0$).
\begin{equation}
\label{eq:psi_b}
\psi^{\text{b}}(t)  = \mathbb{E} \left( Y^{A=1}_{t} - Y^{A=0}_{t} | A^{Z=1} > A^{Z=0}  \right)
\end{equation}
This estimand is a \textit{local} causal effect applied to the subpopulation of \textit{``compliers"} in the sense that we are conditioning on the subject who would take a treatment if and only if they receive the instrument. This estimand is also called the complier average causal effect (CACE) or local average treatment effect (LATE)~\citep{angrist1996identification, frolich2007nonparametric, tan2006regression, ogburn2015doubly}. One rationale for focusing on this subpopulation is that the LATE is defined by the effect due to the treatment induced by the instrumental variable, and this effect is free from unmeasured confounding. This enables us to obtain an unbiased causal effect due to the treatment.

Although many of the IV methods for survival outcomes have been developed for a binary instrument~\citep{richardson2017nonparametric, kianian2019causal}, instruments in practice are often continuous, e.g., genetic risk scores in Mendelian randomization~\citep{burgess2013use, nordestgaard2012effect} or differential distance from the institution where the treatments are provided ~\citep{rassen2009instrumental, baiocchi2014instrumental}. 

When $Z$ is continuous, our target estimand is similar to the binary case, but with a slightly different subpopulation. For a positive $\kappa > 0$, $\psi^{\text{c}, \kappa}(t)$ denotes our causal estimand with a continuous instrument:
\begin{equation}
	\label{eq:psi_c}
	\psi^{\text{c}, \kappa}(t) = \mathbb{E}\left( Y^{A=1}_{t} - Y^{A=0}_{t} | A^{Z+\kappa} > A^{Z - \kappa} \right)
\end{equation}
The target estimand in \eqref{eq:psi_c} is defined by the difference in survival probabilities at time $t$ between the treatment and control groups within the subpopulation in which subjects take the treatment $(A=1)$ when they are assigned an instrument value of $Z + \kappa$ and take the control when they are assigned an instrument value $Z-\kappa$. For example, if $Z$ denotes  distance  to the nearest hospital that provides a particular treatment of interest and $A$ is an indicator of receiving that treatment,  a difference of $2 \kappa$ could be defined as  distance to the hospital that would alter the intervention received.
\cite{mauro2018instrumental} elucidate advantages to this type of causal estimand. We provide details of our causal estimand~\eqref{eq:psi_c} when the instrument is continuous in the Supporting Information. In the the remaining sections, we  focus on the estimation of~\eqref{eq:psi_b} with a binary IV.

\subsection{Identification}

Clearly,  for each subject, we are not able to observe the potential outcomes under both treatments nor under different instrument values. Moreover, $Y_{t}$ is not observable if censoring precedes the event at $t$. The following identification assumptions are thus needed to estimate $\psi^{\text{b}}(t)$ from the observables, $\mathcal{O} = \left( \bm{X}_{0}, Z, A, R Y_{1},  R Y_{2}, \ldots,  R Y_{\tau} \right)$. 

\begin{itemize}
    \item[] (A1)  Consistency: If $Z = z$, then $A = A^{z}$; if $A = a$, then $Y_{t} = Y^{a}_{t}$.
    \item[] (A2)  Ignorability: $Z \bigCI (Y^{z}_{t}, A^{z}) | \bm{X}_{0}$.
     \item[] (A3) Exclusion restriction: $Y^{Z=z}_{t} = Y^{Z=z, A^{z}}_{t} = Y^{A^{z}}_{t}$.
     \item[] (A4) Independent censoring indicator: $R \bigCI (Y_{t}, Y_{t+1}, \ldots, Y_{\tau}) | \bm{X}_{0}, Z, A$.
      \item[] (A5) Positivity:  $\mathbb{P}\left( Z = z | \bm{X}_{0} \right) > 0$ and $\mathbb{P}\left( R  = 1 | \bm{X}_{0},  Z = z, A = a\right) > 0$.
     \item[] (A6) Monotonicity: $\mathbb{P}\left( A^{Z=0} > A^{Z=1} \right) = 0$; i.e., $\mathbb{P} \left( A^{Z=1} \leq A^{Z=0} \right) = 1$.
\end{itemize}     

Assumptions (A1)--(A6) are needed to bridge the gap between observations and potential outcomes, and  observables and non-observables. Note that Assumption (A4) is a key assumption needed to identify the causal effect when there is censoring. This assumption implies that whether an event precedes censoring is independent of the survival indicator at \textit{any time point t} ($=0,1,\ldots, \tau$) conditional on baseline covariates, instrument, and treatment. 
Assumption (A4) puts a constraint on the gap time between the censoring and failure times, i.e., $Q := T - C$; if this gap time $Q$, which is not observable, is conditionally independent of failure time $T$, then (A4) holds. Therefore, if the censoring time is a random perturbation around the failure time (e.g., when a patient's exit time from the study is nearly randomly distributed around their failure time), then assumption (A4) is satisfied. In other words, in practice, if censoring time is strongly associated with failure time, (A4) would be a reasonable assumption.
However, this assumption might not be justifiable in some contexts (e.g., administrative censoring or uniformly distributed censoring times), so we will consider ignorable censoring  in Section~\ref{ssec:hazards}.
\section{Proposed Nonparametric Estimator}
\label{sec:propose}
We have discussed a set of identification assumptions to connect our target estimand~\eqref{eq:psi_b}, represented through  counterfactuals, to observable quantities. We now focus on understanding the true data-generating process $\mathbb{P}$ of the observed data $\mathcal{O}$ (under identification assumptions) to estimate our causal estimand in~\eqref{eq:psi_b} with reasonably small bias and variance. This is often accomplished with parametric modeling of $\mathcal{O}$; for example, we often directly fit a Cox proportional hazards model for $T$ conditioned on $A$ and $\bm{X}_{0}$. 
However, it is almost impossible to guarantee that we have modeled the true data generating process correctly, particularly when $\bm{X}_{0}$ is high-dimensional. 
We instead present a causal estimator with more flexible properties using its influence function in which case we do not need to assume a correct parametric model for the true data-generating process.

Influence function-based causal inference has been used less frequently with IVs and survival~\citep{robins1991correcting, vansteelandt2014structural, mauro2018instrumental, kennedy2019robust, yang2020semiparametric, dukes2019doubly, diaz2019statistical}, compared to more standard uncensored settings with no unmeasured confounding.
An influence function for a causal estimand provides an optimal estimator, with knowledge on asymptotic behaviour,  and promises both flexibility and efficiency.
Finding an influence function is often  challenging, but  representating  the target estimand as a series of conditional expectations can provide relatively simple and intuitive influence functions. \cite{van2003unified} and  \cite{kennedy2016semiparametric} provide excellent reviews of influence function theory. Here, we develop an influence function-based estimator for the LATE and demonstrate its estimation procedure and properties.  
\subsection{Nonparametric estimator conditioning on non-censored subjects}
\label{ssec:binary}

An influence function is a function of the observed data, so we must first  convert  counterfactuals into  observable quantities.
The  following Lemma shows how to represent our target estimand~\eqref{eq:psi_b} through identifiable conditional expectations. 

\begin{lemma} (Target estimand $\psi^{\text{b}}(t)$ through conditional expectations).
	\label{lemma:binary_psi}
Under the identification assumptions (A1)-(A6), our target estimand, the LATE of $A$ on $Y_{t}$, can be represented through the following conditional expectations:
\[
		\psi^{\text{b}}(t)   = \frac{ \mathbb{E} \left(  \mathbb{E} \left( Y_{t} | \bm{X}_{0}, Z = 1  \right)  \right) -    \mathbb{E} \left(   \mathbb{E} \left( Y_{t} | \bm{X}_{0},  Z = 0 \right)  \right) }{\mathbb{E} \left(  \mathbb{E} \left( A | \bm{X}_{0}, Z = 1  \right)  \right) -  \mathbb{E} \left(  \mathbb{E} \left( A | \bm{X}_{0}, Z = 0 \right)  \right)}, \\
\]
where $\mathbb{E}\left( Y_{t} | \bm{X}_{0}, Z = z \right) = \sum\limits_{a \in \{0,1 \}} \mathbb{E}\left(  Y_{t} | \bm{X}_{0}, R = 1, Z = z, A = a \right) \mathbb{P}(A = a | \bm{X}_{0},  Z = z)$ for $z \in \{0,1\}$.
\end{lemma}

Note that due to censoring, if $R=0$ (i.e., if censoring precedes event), $Y_{t}$ might not be observed, so we cannot identify $\mathbb{E}(Y_{t} | \bm{X}_{0}, Z=1)$. However, if we expand $\mathbb{E}(Y_{t} | \bm{X}_{0}, Z)$ into two conditionals, we are able to identify it using two, identifiable expectations:
\begin{eqnarray*}
	\mathbb{E}(Y_{t} | \bm{X}_{0}, Z ) &=&  \sum\limits_{a \in \{0,1\}} \mathbb{E}(Y_{t} | \bm{X}_{0}, Z, A = a) \mathbb{P}(a | \bm{X}_{0}, Z) \\
	&=& \sum\limits_{a \in \{ 0, 1\}} \mathbb{E}(Y_{t} | \bm{X}_{0}, Z, A = a, R = 1) \mathbb{P}(a | \bm{X}_{0}, Z).
\end{eqnarray*}	 
The last equality in the above is due to (A4). For ease of notation, we represent the estimand as:
$\psi^{\text{b}, Z=z}_{1}(t) =  \mathbb{E}\left( \mathbb{E}\left( Y_{t} | \bm{X}_{0}, Z = j  \right) \right)$, 
$\psi^{\text{b}, Z = z}_{2} = \mathbb{E}\left(  \mathbb{E} \left( A | \bm{X}_{0}, Z = j \right) \right)$ for $z \in\{0,1\}$, and 
$\psi^{\text{b}} (t) =  \left\{  \psi^{\text{b}, Z= 1}_{1}(t) - \psi^{\text{b}, Z=0}_{1}(t) \right\} \left\{  \psi^{\text{b}, Z=1}_{2} - \psi^{\text{b}, Z=0}_{2} \right\}^{-1} := \psi^{\text{b}}_{1}(t) (\psi^{\text{b}}_{2})^{-1}$.

Before introducing an influence function of $\psi^{\text{b}}(t)$, we
consider the following nuisance functions that are estimable from the data where $\mathbb{I}(\cdot)$ is an indicator function:
\begin{itemize}[label={}]
	\item[(i)] \textbf{Survival indicator}:  $\mu_{t,z, a}(\bm{X}_{0}) = \mathbb{E}(Y_{t} | \bm{X}_{0}, Z = z, A =a) = \mathbb{E}(Y_{t} | \bm{X}_{0}, R = 1, Z = z, A = a)$.  
	\item[(ii)] \textbf{Censoring indicator}: $ \omega_{z,a}(\bm{X}_{0}) = \mathbb{E}(\mathbb{I}(T < C) | \bm{X}_{0}, Z =z, A = a)$. 
	\item[(iii)] \textbf{Treatment propensity score}:  $\pi_{z}(\bm{X}_{0}) = \mathbb{E}(A  | \bm{X}_{0}, Z = z)$.
	\item[(iv)] \textbf{Instrument prevalence}: $\delta_{z}(\bm{X}_{0}) = \mathbb{P}(Z = z | \bm{X}_{0})$.
\end{itemize}
We now introduce an efficient influence function of $\psi^{b}(t)$ and a causal estimator based on this influence function.
Let $\IF(f)$ be an operator that produces an efficient influence function given a function $f$ and let $\mathbb{P}_{n}$ denote the empirical measure. For notational simplicity, we omit $(\bm{X}_{0})$ from each of the nuisance functions. Unless otherwise mentioned, all nuisance functions are conditioned on baseline covariates $\bm{X}_{0}$. Let $\Theta$ denote a set of the nuisance functions and $\hat{\Theta}$ denote the estimated nuisance functions.
\begin{theorem}[Influence function of $\psi^{b}(t)$]
	\label{thm:phib_influence}
Under  identification assumptions (A1)-(A6), the efficient influence function of the causal difference in survival probability among compliers, $\psi^{\text{b}}(t)$, is given by:
\[
	\mbox{\IF}\left( \psi^{\text{b}}(t) \right) =  \left\{ (M_{1} (t) - M_{0} (t)   ) - \psi^{\text{b}}(t)( \Pi_{1} - \Pi_{0} ) \right\} / \psi^{\text{b}}_{2},
\]	
where, for each of the instrumental variable values $j=0,1$:
\begin{eqnarray}
\label{eq:M_j}
M_{j} (t; \Theta) &=&  \mu_{t, j, 1} \pi_{j} + \mu_{t, j, 0}(1-\pi_{j}  ) +  \frac{\mathbb{I}(Z = j)}{\delta_{j}} \left[   \frac{R A}{\omega_{j, 1} } (Y_{t} - \mu_{t,j, 1}) + \mu_{t, j, 1}(A - \pi_{j}) \right] \nonumber \\ &+& \frac{\mathbb{I}(Z = j)}{\delta_{j}} \left[   \frac{R (1-A)}{\omega_{j, 0} } (Y_{t} - \mu_{t, j, 0}) + \mu_{t, j, 0}((1-A) - (1-\pi_{j}) \right]  \nonumber \\
\Pi_{j}(\Theta) &=&  (A - \pi_{j} ) \mathbb{I}(Z = j) / \delta_{j} + \pi_{j}. 	\nonumber
\end{eqnarray}	
\end{theorem}
Note that in the above equations, all of the variables are observable given the estimated nuisance functions -- in particular, $Y_{t}$  is always accompanied by $R=1$. 
Denote the influence function of $\psi^{\text{b}}_{1}(t; \Theta)$ and $\psi^{\text{b}}_{2}(\Theta)$ as $\phi^{\text{b}}_{1}(t;\Theta)$ and $\phi^{\text{b}}_{2}(\Theta)$, respectively. Then we can rephrase $M_{1}(t; \Theta) - M_{0}(t; \Theta) = \phi^{\text{b}}_{1}(t; \Theta) + \psi^{\text{b}}_{1}(t; \Theta)$ and $\Pi_{1}(\Theta) - \Pi_{0}(\Theta) = \phi^{\text{b}}_{2}(\Theta) + \psi^{\text{b}}_{2}(\Theta)$.
The influence function above guides us in constructing an efficient, nonparametric causal estimator for~\eqref{eq:psi_b} that solves the estimating equation  $\mathbb{P}_{n} (  \mbox{\IF}( \psi^{\text{b}}(t ; \hat{\Theta}) )  ) = 0$ given $\hat{\Theta} = \{ \hat{\mu}_{t, z, a},  \hat{\omega}_{z, a}, \hat{\pi}_{z}, \hat{\delta}_{z} \}$. 
\begin{proposition}[Influence function-based estimator for $\psi^{b}$(t)]
	An influence function-based estimator $\hat{\psi}^{\text{b}}(t)$ or $\hat{\psi}^{\text{b}}(t; \hat{\Theta})$ given the estimated nuisance functions $\hat{\Theta}$ is as follows:
\begin{eqnarray}
\label{eq:IFest_psib}
	&& \hat{\psi}^{\text{b}}(t; \hat{\Theta}) =\mathbb{P}_{n}\left( M_{1}(t; \hat{\Theta}) - M_{0}(t; \hat{\Theta}) \right) / \mathbb{P}_{n} \left( \Pi_{1}(\hat{\Theta}) - \Pi_{0}(\hat{\Theta}) \right).
\end{eqnarray}
\end{proposition}
In Section~\ref{sec:asymptotic}, we provide estimation procedures and theoretical properties of this estimator.
In parallel, we present the influence function for~\eqref{eq:psi_c} under a continuous instrument in the Supporting Information S3.

\subsection{Nonparametric estimator conditioning on the risk set}
\label{ssec:hazards}

We provide a causal estimator under the common  ignorable censoring assumption based on $n$ i.i.d. observations of $\mathcal{O}^{\prime}$ = $\left( \bm{X}_{0}, Z, A, \min (T, C), R \right)$. Consider replacing (A4) by (A4*) Random censoring : $T \bigCI C \mid A, \bm{X}_{0}$.
The fundamental difference between (A4) and (A4*) is in the \textit{conditioning set} induced by censoring. Under (A4), the distribution of $Y_{t}$ is identified by conditioning on the non-censoring indicator $R=1$, which is time-invariant; on the other hand, (A4*) identifies the discrete hazards up to time point $t$ by conditioning on the risk set at $k=1,2,\ldots, t$, which is time-varying. Accordingly, nuisance functions that involve censoring, i.e., the conditional distribution of $Y_{t}$ and the censoring distribution, will be estimated differently under (A4*).

Under ignorable censoring assumption (A4*), an influence function-based nonparametric estimator for survival probabilities was proposed by~\cite{diaz2019statistical}.  In that work, survival and censoring functions were constructed through the product of hazards conditional on events  only being observed at a finite number of time points, $t= 1,2,\ldots, \tau$. 
First, let's consider the following nuisance functions:
\begin{enumerate}[label={}]
	\item[(i)] \textbf{Survival function}:  $S_{t,z,a}(\bm{X}_{0}) = \mathbb{E}(Y_{t} | \bm{X}_{0}, Z = z, A =a) =  \prod\limits_{k=1}^{t} \{ 1 - h_{k}(\bm{X}_{0}, Z = z, A = a) \}$, where $h_{k}(\bm{X}_{0}, Z = z, A = a) = \mathbb{E}(\mathbb{I}( \min(T, C) = k,  R = 1   | \bm{X}_{0}, Z = z, A = a)$.
	\item[(ii)] \textbf{Censoring function}: $G_{t, z,a}(\bm{X}_{0}) = \mathbb{E}(\mathbb{I}(C > t) | \bm{X}_{0}, Z = z , A = a) = \prod\limits_{k=1}^{t} \{ 1 - g_{k}(\bm{X}_{0}, Z = z, A = a) \}$, where $g_{k}(\bm{X}_{0}, A = a, Z = z) = \mathbb{E}(\mathbb{I}( \min(T, C) = k,  R = 0  ) | \bm{X}_{0}, Z = z, A = a)$.
	\item[(iii)] \textbf{Treatment propensity scores}:  $\pi_{z}(a ; \bm{X}_{0}) = \mathbb{E}(A  = a | \bm{X}_{0}, Z = z)$.
\end{enumerate}

Then the estimator $\hat{\psi}^{\text{b}}_{\text{Naive-hazard}} (t; \hat{\Theta})$ or $\hat{\psi}^{\text{b}}_{\text{Naive-hazard}} (t)$ based on the efficient influence function for $\mathbb{E}(Y^{A = 1}_{t} - Y^{A = 0}_{t} )$ proposed in~\cite{diaz2019statistical} is given by:
\begin{eqnarray}
\label{eq:naive_hazard}
\hat{\psi}^{\text{b}}_{\text{Naive-hazard}} (t; \hat{\Theta}) &=&  \mathbb{P}_{n} \left[ \Lambda^{A=1}(t; \hat{\Theta}) - \Lambda^{A=0}(t; \hat{\Theta}) \right], \\ 
\Lambda^{A = a} (t; \Theta) &= & - \sum\limits_{k=1}^{t} \frac{\mathbb{I}(A  = a)  \mathbb{I}(Y_{k-1} = 1)}{ \pi(a ; \bm{X}_{0}) G_{k-1, a}(\bm{X}_{0})}   \frac{S_{t, a}(\bm{X}_{0})}{ S_{k, a}(\bm{X}_{0})}  \{ R \mathbb{I}(Y_{k} = 0) - h_{k}(\bm{X}_{0}, A = a)  \} + S_{t, a}(\bm{X}_{0}), \nonumber
\end{eqnarray}
where $S_{k, a}(\bm{X}_{0}) = \prod\limits_{l=1}^{k} \{ 1 - h_{l}(\bm{X}_{0}, A = a) \}$ and 
$G_{k,a}(\bm{X}_{0}) = \prod\limits_{l=1}^{k} \{  1- g_{l}(\bm{X}_{0}, A = a) \}$
for $a  \in \{0, 1\}$.   The above estimator $\hat{\psi}^{\text{b}}_{\text{Naive-hazard}}(t)$ was shown to be a doubly robust and efficient estimator in the absence of unmeasured confounding~\citep{diaz2019statistical}. Note that we omit the subscript $z$ for each function since they do not condition on an instrument $Z$.

However, the estimator $\hat{\psi}^{\text{b}}_{\text{Naive-hazard}}(t)$ is susceptible to unmeasured confounding.
By adapting the above estimator~\eqref{eq:naive_hazard}, we propose a new influence function-based estimator for the LATE~\eqref{eq:psi_b} that retains the same properties as~\eqref{eq:naive_hazard} even in the presence of unmeasured confounding.
To do so, we introduce a binary instrument $Z$ with an additional nuisance function \textbf{(iv)} \textbf{instrumental prevalence:} $\delta_{z}(\bm{X}_{0}) = p( Z =z | \bm{X}_{0})$.  Then the next theorem introduces an influence function-based estimator for $\psi^{\text{b}}(t)$ under ignorable censoring (A4*).
We call this estimator $\hat{\psi}^{\text{b}}_{\text{IF-hazard}}(t)$. 
\begin{proposition}(Influence function-based estimator for $\psi^{\text{b}}(t)$ using discrete hazards)
Under the causal conditions (A1)--(A3), censoring condition (A4*), and the monotonicity assumption ($\mbox{A6}$), the efficient influence function-based estimator for $\psi^{\text{b}}(t)$ is given by:
\label{thm:if_hazard}
\begin{eqnarray}
\label{eq:if_hazard}
\hat{\psi}^{\text{b}}_{\text{IF-hazard}} (t; \hat{\Theta}) &=&  \frac{ \mathbb{P}_{n}(\widetilde{\Lambda}^{Z=1}(\hat{\Theta}) -   \widetilde{\Lambda}^{Z=0}(\hat{\Theta}) ) }{  \mathbb{P}_{n} ( \Pi_{1}(\hat{\Theta}) - \Pi_{0}(\hat{\Theta}) ) }, \mbox{ where }
\widetilde{\Lambda}^{Z=z} =  \widetilde{\Lambda}^{Z=z, A = 1} +   \widetilde{\Lambda}^{Z=z, A = 0}, \\   \widetilde{\Lambda}^{Z=z, A = a}(\Theta) & = & - \frac{\mathbb{I}(Z =  z)}{ \delta_{z}(\bm{X}_{0})} \left[ \sum\limits_{k=1}^{t} \frac{\mathbb{I}(A  = a)  \mathbb{I}(Y_{k-1} = 1)}{   \pi_{z}(a ; \bm{X}_{0}) G_{k-1, z, a}(\bm{X}_{0})}   \frac{S_{t, z, a}(\bm{X}_{0})}{ S_{k, z, a}(\bm{X}_{0})}  \{ R \mathbb{I}(Y_{k} = 0) - h_{k}(\bm{X}_{0}, z, a)  \}   \right]  \nonumber   \\
& + &  \frac{\mathbb{I}(Z =  z)}{ \delta_{z}(\bm{X}_{0})}  S_{t, z, a}(\bm{X}_{0}) (\mathbb{I}(A = a) - \pi_{z}(a ; \bm{X}_{0})) + S_{t, z, a}(\bm{X}_{0}) \pi(a ; \bm{X}_{0}) 
\nonumber,
\end{eqnarray}
assuming a new positivity assumption (A5*) $\delta_{z}(\bm{X}_{0}), \pi_{z}(a; \bm{X}_{0}), G_{k, z, a}(\bm{X}_{0}), S_{k, z, a}(\bm{X}_{0}) > 0$ for all $a,z \in \{0, 1\}$ and $k=1,2,\ldots, t$. 
\end{proposition}
The estimator~\eqref{eq:if_hazard} replaces $\hat{\mathbb{E}}(Y_{t} | \bm{X}_{0}, Z = z)$ in~\eqref{eq:IFest_psib} by a collection of estimated hazard functions until time $t$ (instead of directly estimating the survival function through conditional expectation of a binary $Y_{t}$). The estimator~\eqref{eq:if_hazard} is a reasonable choice under administrative censoring and with failure times observed only at certain time points. When the ignorable censoring assumption (A4*) is violated, however, the performance of the estimator~\eqref{eq:if_hazard} becomes highly biased and unstable, as we demonstrate in our later simulation studies. 
\subsection{Estimation and Properties of an Estimator}
\label{ssec:performance}

A principal advantage of our proposed estimator~\eqref{eq:IFest_psib} is that each of the nuisance functions $\hat{\Theta}$ can be estimated nonparametrically. Similarly, we can also nonparametrically estimate the hazard functions for failure and censoring times for estimator~\eqref{eq:if_hazard} as well as $\{\hat{\delta}, \hat{\pi}\}$. Once all of the nuisances are estimated, we take a sample average with the estimated nuisance functions to obtain the estimates~\eqref{eq:IFest_psib} and~\eqref{eq:if_hazard}, respectively.
However, if these two processes use all $n$ i.i.d. samples $\mathcal{\bm{O}}_{n} = \{ \mathcal{O}_{i} \}_{i=1}^{n}$ (or $\{ \mathcal{O}^{\prime}_{i} \}_{i=1}^{n}$), this may result in overfitting because the same data will be used twice for nuisance function estimation and for causal effect estimation. For this reason, we use sample splitting or  cross-fitting~\citep{robins2008higher, zheng2010asymptotic, chernozhukov2018double, diaz2019statistical}.
Specifically, we partition the data into $K$ mutually exclusive groups, $G_{1} \cup G_{2} \cup \cdots \cup G_{k} = \bm{O}_{n}$. We then use the data excluding one group $\bm{O}_{n} \setminus G_{k}$ to estimate the nuisance functions and use the remaining sample in $G_{k}$ to evaluate the causal effect for each iteration $k~(=1,2,\ldots, K)$. Combining all estimates across the  $k$ iterations, the proposed estimator for a binary instrument using sample splitting is given by: 
\begin{equation}
\label{eq:split_psib}
\hat{\psi}^{\text{b, split}}(t) =  \frac{1}{K} \sum\limits_{k=1}^{K} \hat{\psi}^{\text{b}(k)}(t).
\end{equation}
To nonparametrically estimate each nuisance function for $\hat{\psi}^{\text{b} (k)}(t)$ we can can use, for example, random forests implemented in the \texttt{ranger} package~\citep{wright2019fast} available in \texttt{R} (see Algorithm~1 in the Supporting Information). 
For our proposed estimator~\eqref{eq:IFest_psib}, however, we directly estimate the conditional density of the binary outcome $Y_{t}$ at a given time point, instead of estimating a whole survival function; in fact, this enables us to dramatically reduce computing time. 
For the proposed estimator~\eqref{eq:if_hazard}, we can also estimate the discrete hazards functions $\hat{h}_{t}(\bm{X}_{0}, Z, A)$ and $\hat{g}_{t}(\bm{X}_{0}, Z, A)$ using  conditional binary outcome models instead of estimating the whole survival function nonparametrically. 

Further, sample splitting allows us to partition the total error into (i)  the error resulting from  estimation of nuisance functions (e.g., error due to using $\hat{\mu}_{t, z, a}$ instead of $\mu_{t, z, a}$) and (ii) the error from sample approximation (i.e., error due to using $\mathbb{P}_{n}$ instead of $\mathbb{P}$). This is particularly useful for studying the asymptotic behavior of the proposed estimators, which we  discuss next. 

\section{Asymptotic properties of estimators}
\label{sec:asymptotic}
The next two theorems address the large-sample behavior of our two proposed estimators~\eqref{eq:IFest_psib} and~\eqref{eq:if_hazard}. First consider  the following conditions for~\eqref{eq:IFest_psib}:
\begin{itemize}[label={}]
	\item[(C1)]  The nuisance functions for each estimator are in the Donsker class. 
	\item[(C2)] For some constant $\epsilon > 0$, $\mathbb{P}(\epsilon < \hat{\omega}_{z,a} \hat{\delta}_{z} < \infty ) = 1$ and  $\mathbb{P}(\epsilon < \delta_{z} < \infty ) = 1$ for all $z,a \in \{0,1\}$.
\end{itemize}	
However, belonging to a Donsker class (C1) may restrict the complexity of the estimators;
here we use sample splitting to alleviate some of the Donsker class restrictions \citep{chernozhukov2018double, diaz2020machine}. 

\begin{theorem}(Asymptotic distribution of $\hat{\psi}^{\text{b}}(t)$~\eqref{eq:IFest_psib})
	\label{thm:psi_b_asy}
	Under (C1)--(C2) in addition to identification conditions (A1)--(A6), 
\begin{eqnarray}
	\label{eq:binary_result}
	&&\hat{\psi}^{\text{b}}(t; \hat{\Theta}) - \psi^{\text{b}}(t; \Theta) \\  & = & O_{\mathbb{P}} \left\{ \sum\limits_{z,a \in \{ 0,1 \}} \left( \parallel \omega_{z,a} \delta_{z}-  \hat{\omega}_{z, a}  \hat{\delta}_{z} \parallel \cdot \parallel \mu_{t, z, a} - \hat{\mu}_{t, z, a} \parallel  \right) \nonumber  + \sum\limits_{z \in \{0,1\} }  \left( \parallel \delta_{z} -  \hat{\delta}_{z} \parallel \cdot \parallel  \pi_{z}  -  \hat{\pi}_{z}  \parallel\right) \right\}
	\\ & + & \xi^{-1}_{n}(\mathbb{P}_{n} - \mathbb{P}) ( \phi^{\text{b}}_{1}(t ; \Theta) - \psi^{\text{b}}(t; \Theta) \phi^{\text{b}}_{2}(\Theta) ) +  o_{\mathbb{P}}(n^{-1/2}), \nonumber
\end{eqnarray}
where $\xi_{n} = \mathbb{P}_{n}(\phi^{\text{b}}_{2}(\hat{\Theta})  + \psi^{\text{b}}_{2}(\hat{\Theta}))$.
\end{theorem}
Note that the second term of $(\mathbb{P}_{n} - \mathbb{P})( \phi^{\text{b}}_{1}(t ; \Theta) - \psi^{\text{b}}(t ; \Theta) \phi^{\text{b}}_{2}(\Theta))$ is asymptotically normal by the central limit theorem.
Theorem~\ref{thm:psi_b_asy} implies (i) double robustness between the outcome distribution and the joint distribution of the censoring and instrument indicators and between the instrument indicator and the treatment indicator; and (ii) $n^{1/4}$ convergence rate of each nuisance function to sufficiently guarantee $\sqrt{n}$-convergence of $\hat{\psi}^{\text{b}}(t; \hat{\Theta})$. In fact, faster than $n^{1/4}$ rates for one nuisance function estimator may allow slower than $n^{1/4}$ rates for other nuisance estimators.

For estimator~\eqref{eq:if_hazard}, we also prove double robustness by adapting the proof presented in~\cite{moore2009increasing}. Here we assume  (C1) and replace (C2) by (C2*): $\mathbb{P} ( \epsilon < \pi_{z}(a) < \infty) = 1$, $\mathbb{P} ( \epsilon < \hat{\delta}_{z} \hat{G}_{k-1, z, a} < \infty) = 1$, 
$\mathbb{P} ( \epsilon < \hat{S}_{t,z, a} < \infty) = 1$, and $\mathbb{P} ( \epsilon < \hat{\delta}_{z}< \infty) = 1$.
Let $\phi^{\text{b}}_{1, \text{IF-hazard}}(t ; \Theta)$  denote the influence function of $\mathbb{E}(Y^{Z=1}_{t} - Y^{Z=0}_{t})$ under the censoring assumption (A4).

\begin{theorem}(Asymptotic distribution of $\hat{\psi}^{\text{b}}_{\text{IF-hazard}}(t)$)
	\label{thm:IF_hazard}
	When nuisance functions are in the Donsker class or estimated using  sample splitting, the following result holds under (C2*) and (A1)--(A3), (A4*), (A5*), and (A6).
\begin{eqnarray}
\label{eq:IF_hazard_asy}
	 &&\hat{\psi}^{\text{b}}_{\text{IF-hazard}}(t; \hat{\Theta}) - \psi^{\text{b}}_{\text{IF-hazard}}(t; \Theta) \nonumber  \\ &=& O_{\mathbb{P}} \left\{ \sum\limits_{z,a \in \{ 0, 1\}}	\parallel  S_{t,z,a}- \hat{S}_{t,z,a} \parallel \parallel \delta_{z} \bm{G}_{t-1, z, a} - \hat{\delta}_{z}\hat{\bm{G}}_{t-1, z, a} \parallel_{2}  \right. \nonumber + \left. \sum\limits_{z,a \in \{0,1\}} \parallel  \delta_{z} - \hat{\delta}_{z}  \parallel  \parallel \pi_{z}(a) - \hat{\pi}_{z}(a) \parallel \right\}
	 \\ & + & \xi^{-1}_{n}(\mathbb{P}_{n} - \mathbb{P}) \left\{ \phi^{\text{b}}_{1,\text{IF-hazard}}(t ; \Theta)-\psi^{\text{b}}_{\text{IF-hazard}}(t ; \Theta) \phi^{\text{b}}_{2}(\Theta) \right\} +  o_{\mathbb{P}}(n^{-1/2})  \nonumber,
\end{eqnarray}	
where  $\parallel \delta_{z} \bm{G}_{t-1, z, a} - \hat{\delta}_{z}\hat{\bm{G}}_{t-1, z, a} \parallel_{2} = \left\{ \sum\limits_{k=1}^{t} \left( \delta_{z}G_{k-1, z, a} - \hat{\delta}_{z} \hat{G}_{k-1, z, a}  \right)^{2} \right\}^{1/2}$. 
\end{theorem}
Theorem~\ref{thm:IF_hazard} demonstrates the doubly robust properties of \eqref{eq:if_hazard}. All  proofs are provided in the Supporting Information. These asymptotic results are supported by our numerical results in the next section.

\section{Simulation Studies}
\label{sec:simulation}
We present two numerical studies.
First, we investigate the finite-sample performance of our proposed estimator $\hat{\psi}^{\text{b}}(t)$ in~\eqref{eq:IFest_psib} under various scenarios including model misspecification. 
Second, we explore the performance of the previously described influence function-based estimators~\eqref{eq:IFest_psib},~\eqref{eq:naive_hazard}, and~\eqref{eq:if_hazard} under different censoring and unmeasured confounding scenarios.

\subsection{Simulation settings}

We consider two survival outcome models that are commonly implemented in clinical studies: (i) Cox proportional hazards models and (ii) additive hazards models.
We correctly specify each data generating model for survival outcomes in our parametric estimation, but we also consider nonparametric estimation that does not involve any modeling.  Consider the following data generating models  with baseline covariates $\bm{X}_{i,0} \overset{i.i.d.}{\sim} \mbox{MVN}( \bm{0}, \bm{I}_{5 \times 5} )$, binary instruments generated from $\mbox{logit}\left( p (Z_{i} = 1  | \bm{X}_{i,0}) \right) =   \bm{X}^{\prime}_{i,0}  \boldsymbol{\kappa}  $, and binary treatments from $\mbox{logit}\left( p(A_{i} | \bm{X}_{i,0}, Z_{i}, U_{i}) \right) =     -0.1 + \bm{X}^{\prime}_{i,0} \boldsymbol{\alpha}_{x} + Z_{i}  \alpha_{z}+  + U_{i} \alpha_{u}$. 
We generated a non-censoring indicator from   $\mbox{logit}\left( p(R_{i} | \bm{X}_{i,0}, Z_{i}, A_{i}) \right) =  \bm{X}^{\prime}_{i,0} \boldsymbol{\gamma}_{x} +  Z_{i} \gamma_{z} + A_{i} \gamma_{z}$. The two data generating models for survival outcomes are the additive hazards model, $h(t | \bm{X}_{i,0}, Z_{i}, A_{i}, U_{i})  =  h_{0}(t)  +   \bm{X}^{\prime}_{i,0} \boldsymbol{\beta}_{x} + A_{i}  \beta_{a} + U_{i}  \beta_{u}$, and the  Cox proportional hazards model: $h(t |  \bm{X}_{i,0}, Z_{i}, A_{i}, U_{i})  =  h_{0}(t)\exp \left( \bm{X}^{\prime}_{i,0}  \boldsymbol{\beta}_{x} + A_{i}  \beta_{a} +  U_{i}  \beta_{u}\right)$.

We compare misspecified models  to the correctly specified case. We say a model is \textit{misspecified} when we observe $\bm{W}_{i,0} \in \mathbb{R}^{5}$ instead of $\bm{X}_{i,0} \in \mathbb{R}^{5}$. The first four covariates in $\bm{W}_{i,0}$ are transformed versions of the first four covariates in $\bm{X}_{i,0}$ following~\cite{kang2007demystifying};  the fifth covariate is observed correctly. Details can be found in the Supporting Information.
We illustrate the performance of the proposed estimators under four different scenarios: (i) all nuisance functions correctly specified; (ii) incorrectly specified $\hat{\omega}_{z,a}, \hat{\delta}_{z}$; (iii) incorrectly specified $\hat{\pi}_{z}, \hat{\mu}_{t,z,a}$ ; and (iv) incorrectly specified $\hat{\pi}_{z}, \hat{\omega}_{z,a}$. As discussed in Section~\ref{sec:asymptotic}, the influence function-based estimators~\eqref{eq:IFest_psib} and~\eqref{eq:if_hazard} should maintain their consistency under all four scenarios.

In practice, when baseline covariates $\bm{X}_{i,0}$ are high-dimensional, a regression model is likely to misspecify the true data generating process. Hence, we also consider nonparametric estimation of nuisance functions using sample splitting with $K=2$ partitions. 
Details of this procedure are provided both for a binary instrument and a continuous instrument in the Supporting Information.  

To evaluate the performance of each estimator, we report bias and root-mean-squared error (RMSE): 
\begin{eqnarray*}
\widehat{\mbox{bias}}(\hat{\psi}^\text{b}) & = & \frac{1}{\tau}\sum\limits_{t=1}^{\tau} \left|  \frac{1}{I} \sum\limits_{i=1}^{I}  \hat{\psi}^{\text{b}}_{i}(t) - \psi^{\text{b}}_{i}(t) \right| \\
\widehat{\mbox{RMSE}}(\hat{\psi}^{\text{b}}) & = & \frac{\sqrt{n}}{\tau}\sum\limits_{t=1}^{\tau} \left[ \frac{1}{I} \sum\limits_{i=1}^{I} \left\{   \hat{\psi}^{\text{b}}_{i}(t) - \psi^{\text{b}}_{i}(t) \right\}^{2} \right]^{1/2} 
\end{eqnarray*}
We generated $n=1000$ i.i.d. observations, $I=1000$ independent times. Bias and RMSE were integrated over $\tau =30$ times points.

\subsection{Performance of $\hat{\psi}^{\text{b}}(t)$ }
For our first numerical experiment, we compare the performance of (a) $\hat{\psi}^{\text{b}}(t)$ in~\eqref{eq:IFest_psib} with the following two simple estimators for $\psi^{\text{b}}(t)$ that are not based on the influence function: (b) a simple inverse-probability-weighted estimator (IPW estimator) and (c) a regression-based plug-in estimator (Plug-in estimator). See the Supporting Information for details. All three estimators, (a), (b), and (c),  require all or some of the four estimated nuisance functions $\hat{\Theta}$. 
Under parametric modeling, we fit a survival outcome model for $T$ following the same model that $T$ was generated from, i.e., either a Cox model or an additive hazards model.

\begin{figure}[H]
	\centering
	\begin{subfigure}[b]{0.4\textwidth}
		\includegraphics[width=\textwidth]{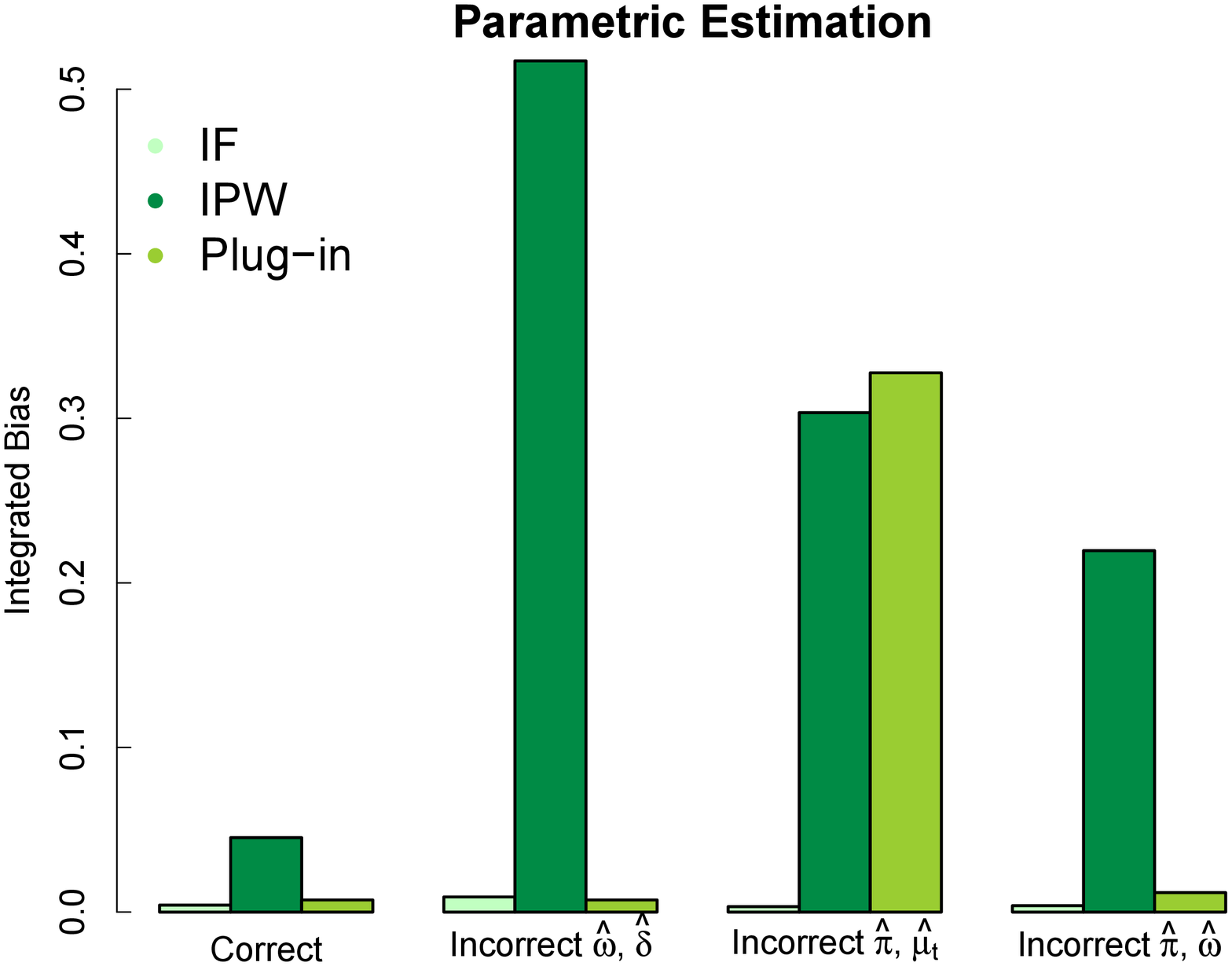}
		\caption{ }
	\end{subfigure}
	\begin{subfigure}[b]{0.4\textwidth}
		\includegraphics[width=\textwidth]{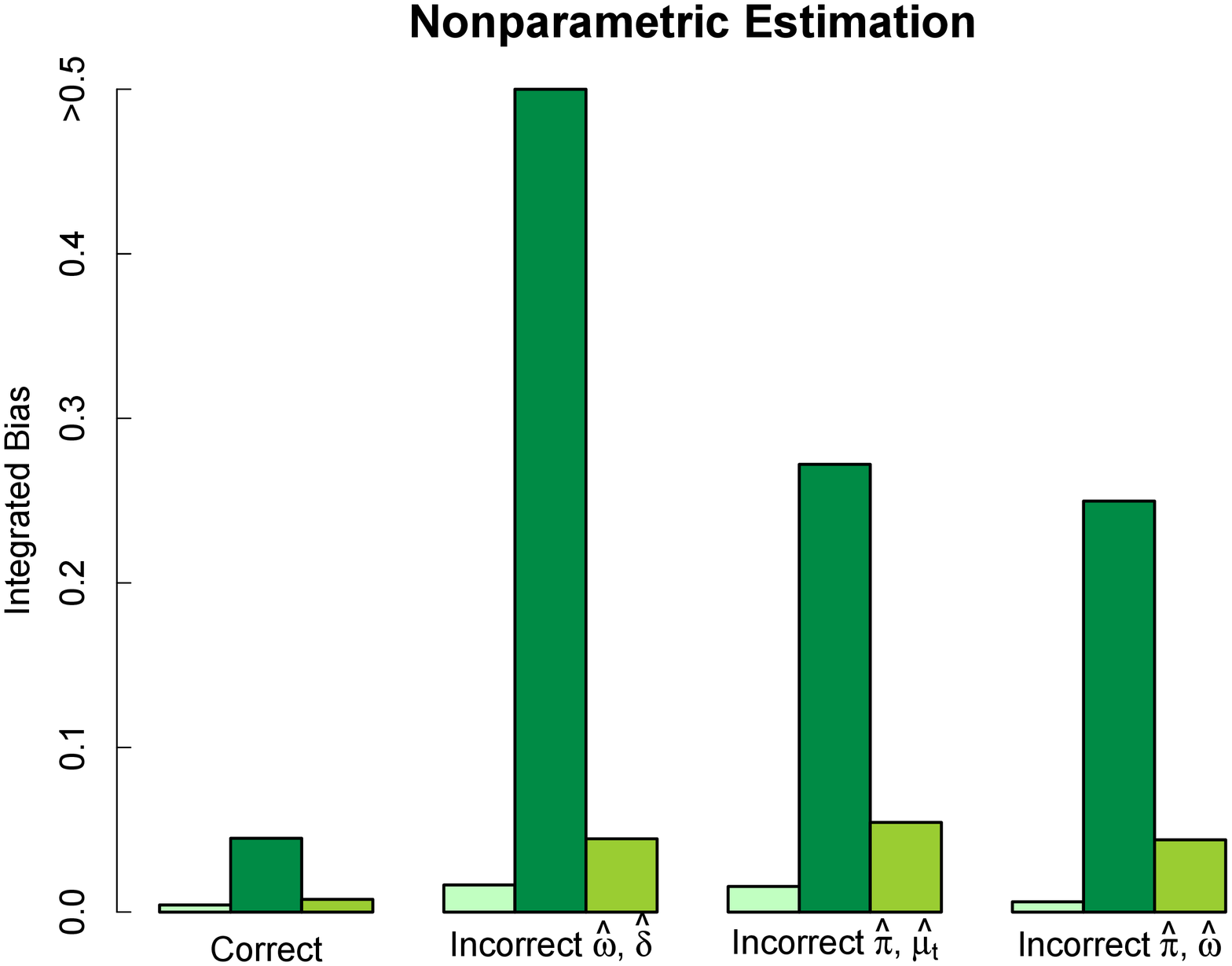}
		\caption{}
	\end{subfigure}
	\begin{subfigure}[b]{0.4\textwidth}
		\includegraphics[width=\textwidth]{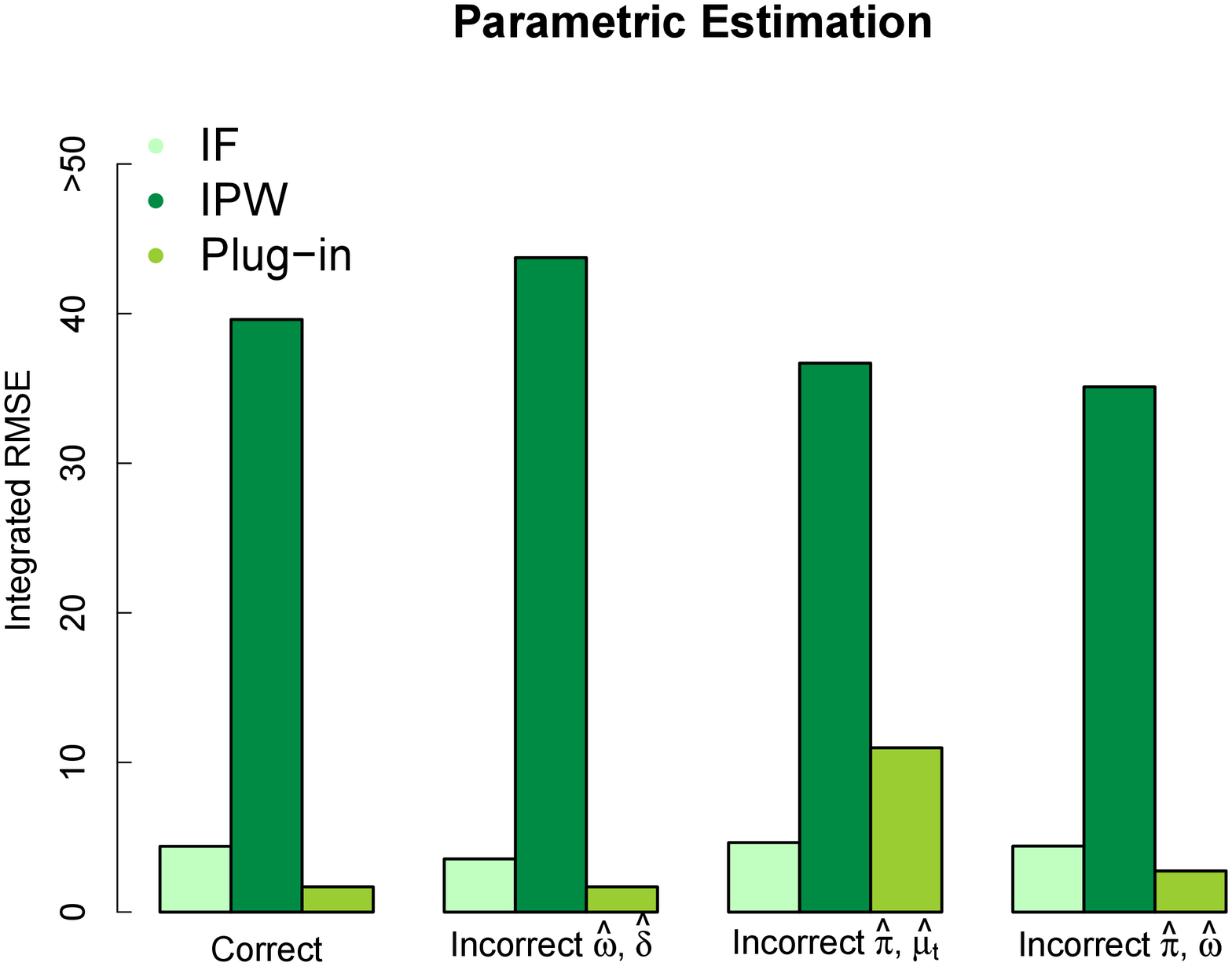}
		\caption{ }
	\end{subfigure}
	\begin{subfigure}[b]{0.4\textwidth}
		\includegraphics[width=\textwidth]{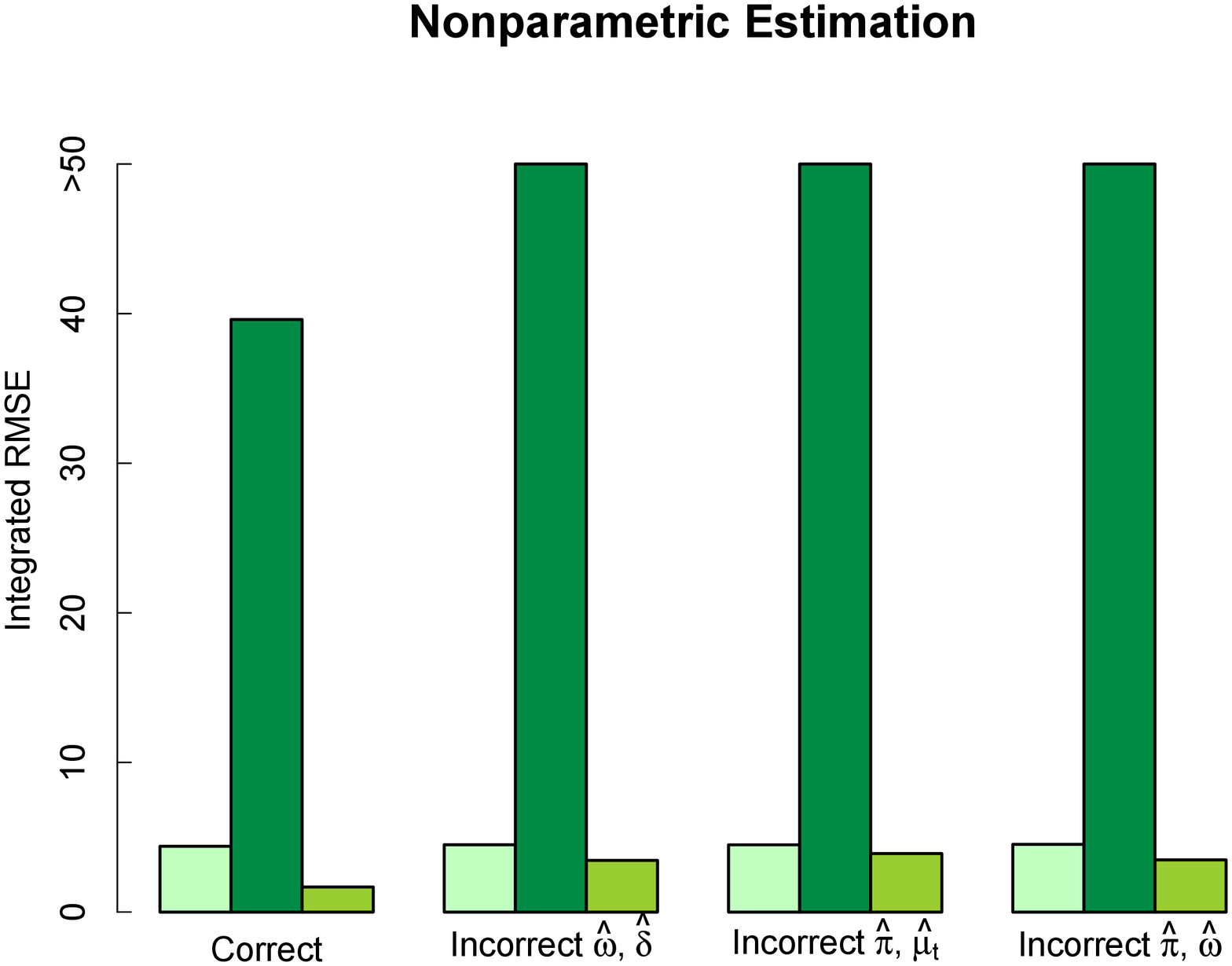}
		\caption{}
	\end{subfigure}
	\caption{\label{fig:cox_binary}Bias (upper panel) and RMSE (lower panel) of $\hat{\psi}^{\text{b}}(t), \hat{\psi}^{\text{b}}_{\text{ipw}}(t)$, and $\hat{\psi}^{\text{b}}_{\text{plugin}}(t)$ with parametric nuisance functions estimation (left panel) and nonparametric estimation (right panel) when survival outcomes are generated from a Cox proportional model and an instrument is binary. Sample size is $n=1000$ and each scenario was replicated $I=1000$ times.}
\end{figure}

Figure~\ref{fig:cox_binary} presents the bias and RMSE for each scenario when survival outcomes are generated from a Cox model. First, our proposed estimator (IF) demonstrates the smallest bias and RMSE across all four scenarios in parametric estimation, and exhibits almost equivalent performance under nonparametric estimation. On the other hand, the IPW estimator is most biased under scenario (ii) and also shows noticeable bias under scenarios (iii) and (iv); in general, the IPW estimator also has the largest RMSE. As expected, the plug-in estimator that only uses $\hat{\pi}_{z}$ and $\hat{\mu}_{t,z, a}$ is most biased under (iii) when the two nuisance functions are incorrect and (iv) when one of them is incorrect with parametric estimation. 
Generally, when the nuisance functions are nonparametrically estimated, all three estimators are less sensitive to model misspecification. This demonstrates the robustness of nonprametric estimation to covariate transformation. 
When the model is correctly specified, the parametric estimator of~\eqref{eq:IFest_psib} performed similarly to the nonparametric estimator where the nuisance functions might not converge at $\sqrt{n}$-rate.   
Similar results when survival outcomes are generated from an additive hazards model are presented in the Supporting Information.

\subsection{Performance of estimators under different censoring and unmeasured confounding conditions}
\label{ssec:sim_discrete}

For the second experiment, we compare the performance of three influence function-based estimators, $\hat{\psi}^{\text{b}}_{\text{Naive-hazard}}(t)$, $\hat{\psi}^{\text{b}}_{\text{IF-hazard}}(t)$, and $\hat{\psi}^{\text{b}}(t)$ under different censoring and unmeasured confounding assumptions. We generated the failure times from a Cox model, and the censoring times under three different scenarios. We also control the amount of unmeasured confounding for each scenario:
(i) $C_{i} \overset{i.i.d.}{\sim} Uniform\left(10, 100 \right)$ and $\beta_{u} = 2.5$;
 (ii) $C_{i} \overset{i.i.d.}{\sim} T_{i} + Uniform\left( -10, 50 \right)$ and $\beta_{u} = 2.5$; and
(iii) $C_{i} = 20$ and $\beta_{u} = 0$.
Scenario (i) satisfies (A4*); (ii) satisfies (A4); and (iii) satisfies (A4*) and assumes no unmeasured confounding.

Similar to the previous experiment, we evaluate the causal effect both using parametric and nonparametric nuisance function estimation and also simulate three different misspecified models as described in the previous simulation study. To make discrete hazard functions valid, censoring and failure times are rounded down to integer values. 
In our setting, we set $\mathbb{E}(Y^{A=1}_{t} - Y^{A=0}_{t}) = \mathbb{E}(Y^{A=1}_{t} - Y^{A=0}_{t} | A^{Z=1} > A^{Z=0} )$, i.e., the LATE is the same as the ATE, making the target estimand  the same across the three estimators. 
\begin{table}[ht]
\centering
\resizebox{0.9\textwidth}{!}{\begin{tabular}{rrrr|rrr}
  \hline
& \multicolumn{3}{c|}{\textbf{Parametric estimation}} &  \multicolumn{3}{c}{\textbf{Nonparametric estimation}}  \\ 
 & $\hat{\psi}^{\text{b}}_{\text{Naive-hazard}}(t)$ & $\hat{\psi}^{\text{b}}_{\text{IF-hazard}}(t)$ & $\hat{\psi}^{\text{b}}(t)$ & $\hat{\psi}^{\text{b}}_{\text{Naive-hazard}}(t)$ & $\hat{\psi}^{\text{b}}_{\text{IF-hazard}}(t)$ & $\hat{\psi}^{\text{b}}(t)$ \\ 
  \hline
    \multicolumn{5}{l}{Scenario (i)} \\
  \hline
\textbf{Correct} & 
17.66 & \cellcolor{GreenYellow}2.04 & 2.57 & 11.97 & 3.03 & 7.48 \\ 
  \textbf{Incorrect $\hat{\omega}, \hat{\delta}$} &
  17.65 & \cellcolor{GreenYellow}3.94 & 9.35 & 11.57 & 4.19 & 8.18 \\ 
\textbf{Incorrect $\hat{\pi}, \hat{\mu}_{t}$} & 
19.49 & \cellcolor{GreenYellow}2.60 & 2.68 & 11.79 & 2.41 & 8.00 \\  
 \textbf{Incorrect $\hat{\pi}, \hat{\omega}$} &
  17.43 & \cellcolor{GreenYellow}2.18 & 2.81 & 11.67 & 3.22 & 8.04 \\ 
   \hline
    \multicolumn{5}{l}{Scenario (ii)} \\
   \hline
\textbf{Correct} & $>10^{10}$ & $>10^{10}$ & \cellcolor{GreenYellow}1.01 & 7.96 & 4.60 & 1.66 \\ 
 \textbf{Incorrect $\hat{\omega}, \hat{\delta}$}  & $>10^{10}$ & $>10^{10}$ &  \cellcolor{GreenYellow}4.81 & 7.70 & 5.92 & 2.67 \\ 
\textbf{Incorrect $\hat{\pi}, \hat{\mu}_{t}$} & $>10^{10}$ & $>10^{10}$ &  \cellcolor{GreenYellow}0.94 & 7.86 & 4.23 & 2.88 \\
  \textbf{Incorrect $\hat{\pi}, \hat{\omega}$} & $>10^{10}$ & $>10^{10}$ & \cellcolor{GreenYellow}0.96 & 7.79 & 4.76 & 2.33 \\ 
  \hline 
  \multicolumn{5}{l}{Scenario (iii)} \\
  \hline
\textbf{Correct}  & \cellcolor{GreenYellow}1.49 &  \cellcolor{GreenYellow}1.58 & 18.05 & 3.73 & 3.20 & 16.49 \\ 
  \textbf{Incorrect $\hat{\omega}, \hat{\delta}$}   & \cellcolor{GreenYellow}1.49 & \cellcolor{GreenYellow}2.10 & 21.85 & 3.94 & 5.96 & 22.21 \\ 
  \textbf{Incorrect $\hat{\pi}, \hat{\mu}_{t}$} &  4.61 & \cellcolor{GreenYellow}1.99 & 17.69 & 2.69 & 3.58 & 17.10 \\ 
 \textbf{Incorrect $\hat{\pi}, \hat{\omega}$}  &  \cellcolor{GreenYellow}1.56 & \cellcolor{GreenYellow}1.61 & 21.88 & 3.02 & 4.07 & 18.07 \\ 
  \hline
\end{tabular}}
\caption{\label{tab:sim2} $100 \times \widehat{\mbox{Bias}}$ of three influence functions under three different scenarios (i)-(iii). In case of parametric estimation, we mark the case when each estimator that satisfies the censoring and unmeasured confounding assumption. Because $\hat{\psi}^{\text{b}}_{\text{Naive-hazard}}(t)$ does not use the instrumental density $\delta$ function, ``Incorrect $\hat{\omega}, \hat{\delta}$" in fact only indicates incorrectly specified $\omega$ function for $\psi^{\text{b}}_{\text{Naive-hazard}}(t)$.
}
\end{table}
Table~\ref{tab:sim2} presents the integrated bias in three causal estimators under three different scenarios. First, with parametric estimation, $\hat{\psi}^{\text{b}}_{\text{IF-hazard}}(t)$ has the smallest bias under (i), but exhibits significant sensitivity to nonignorable censoring in (ii). The naive estimator $\hat{\psi}^{\text{b}}_{\text{Naive-hazard}}(t)$ that does not use an IV retains the smallest bias under scenario (iii) where there is no unmeasured confounding and administrative censoring, 
but it exhibits substantial bias in the presence of unmeasured confounding. 
When nuisance functions are nonparametrically estimated, the estimator with best performance under parametric estimation also shows better performance under nonparametric estimation, but the difference in bias between the three estimators becomes smaller than what we had seen under parametric estimation. 
Results on the RMSE that are presented in the Supporting Information Table S1 demonstrate that our proposed estimator $\hat{\psi}^\text{b}(t)$ has the smallest RMSE under scenario (ii) with parametric estimation but the naive estimator $\hat{\psi}^{\text{b}}_{\text{Naive-hazard}}(t)$ has the smallest RMSE across all three scenarios under nonparametric estimation.

\section{Application to Cancer Screening}
\label{sec:application}

We apply our proposed approach to evaluate the effect of cancer screening on survival using data from the Prostate, Lung, Colorectal, and Ovarian (PLCO) Cancer Screening Trial~\citep{team2000prostate}. Specifically, in the colorectal cancer screening trial, approximately 150,000 participants were \textit{randomly} assigned to either (i) the control arm of usual care or (ii) the intervention arm comprised of two colorectal cancer screening exams -- at baseline and at year 5. Details of the trial can be found in~\cite{prorok2000design} and \cite{kianian2019causal}.
Here, we define our target estimand as the causal effect of colorectal screening on survival probability among those who comply with the randomized assignment.  In this example, we consider the time to all-cause mortality from trial entry as our primary outcome. 

Even though the intervention was randomly assigned, noncompliance was observed among the participants.
Of the $n=142,426$ eligible subjects (those with complete information on the baseline questionnaire, with no history of any cancer including colorectal cancer prior to entry, and age no less than 55 at trial entry),  $70,578~(49.55\%)$ participants were assigned to the control arm and $78,724~(50.45\%)$ participants were assigned to the intervention arm. Of the $78,724$ participants randomized to the intervention arm, $8,146~(11.34\%)$ did not comply with the intervention, meaning that they did not have two colorectal screening exams.

Say $Y_{t}$ is the survival indicator at time $t$ and, $R$ indicates whether death precedes censoring. We define the instrumental variable $Z$ to be a binary indicator of whether the participant was randomized to the intervention (control arm: $Z=0$, intervention arm: $Z=1$), and $A$ is a binary indicator of the actual intervention that the participant received (control arm: $A=0$, intervention arm: $A=1$). Because our instrument is the randomization procedure itself, the underlying assumptions of the IV being associated with treatment, the IV not being associated with unmeasured confounders, and the exclusion restriction all obviously hold.  Note that participants randomized to the control arm would not have the opportunity to have cancer screening, so $\mathbb{E}(A = 1 | Z = 0) = 0$ and hence the monotonicity assumption (A6) also holds. In contrast, $\mathbb{E}(A = 0 | Z = 1)$ is non-zero due to non-compliance. 
We take this into account by, for example, directly estimating $\mathbb{E}(Y_{t} | Z  = 0, \bm{X}_{0})$  via $\mathbb{E}(Y_{t} | Z = 0, A = 0, \bm{X}_{0})$ and not $\sum\limits_{a \in \{0,1\}}  \mathbb{E}(Y_{t} | Z = 0, A = a, \bm{X}_{0}) \mathbb{P}(A = a | Z = 0, \bm{X}_{0})$. 

We consider three different estimators:   $\hat{\psi}^{\text{b}}(t)$,  $\hat{\psi}^{\text{b}}_{\text{IF-hazard}} (t)$
, and $\hat{\psi}^{\text{b}}_{\text{Naive-hazard}}(t)$.
Note that estimators $\hat{\psi}^{\text{b}}(t)$ and $\hat{\psi}^{\text{b}}_{\text{IF-hazard}} (t)$ target the LATE~\eqref{eq:psi_b} while $\hat{\psi}^{\text{b}}_{\text{Naive-hazard}}(t)$  estimates the intent-to-treat effect.
Estimator $\hat{\psi}^{\text{b}}(t)$ is valid under censoring assumption (A4) while $\hat{\psi}^{\text{b}}_{\text{IF-hazard}} (t)$ and $\hat{\psi}^{\text{b}}_{\text{Naive-hazard}}(t)$ are based on (A4*); 
on the other hand, compared to the naive estimator $\hat{\psi}^{\text{b}}_{\text{Naive-hazard}}(t)$, $\hat{\psi}^{\text{b}}(t)$ and $\hat{\psi}^{\text{b}}_{\text{IF-hazard}} (t)$ both provide an unbiased estimator for the causal effect even in the presence of unmeasured confounding.

Before analysis, we  rounded down the observed survival times, e.g., $t=0, 10, 20, \ldots, 8000$, to be able to estimate $\hat{\psi}^{\text{b}}_{\text{IF-hazard}} (t)$
or $\hat{\psi}^{\text{b}}_{\text{Naive-hazard}}(t)$ using the nuisance discrete hazard functions. 
Figure~\ref{fig:realresult} presents our results using these three different estimators with nonparametric nuisance function estimation. We also present the same results with parametric nuisance function estimation in the Supporting Information.
\begin{figure}
  \centerline{
    \includegraphics[width=\textwidth]{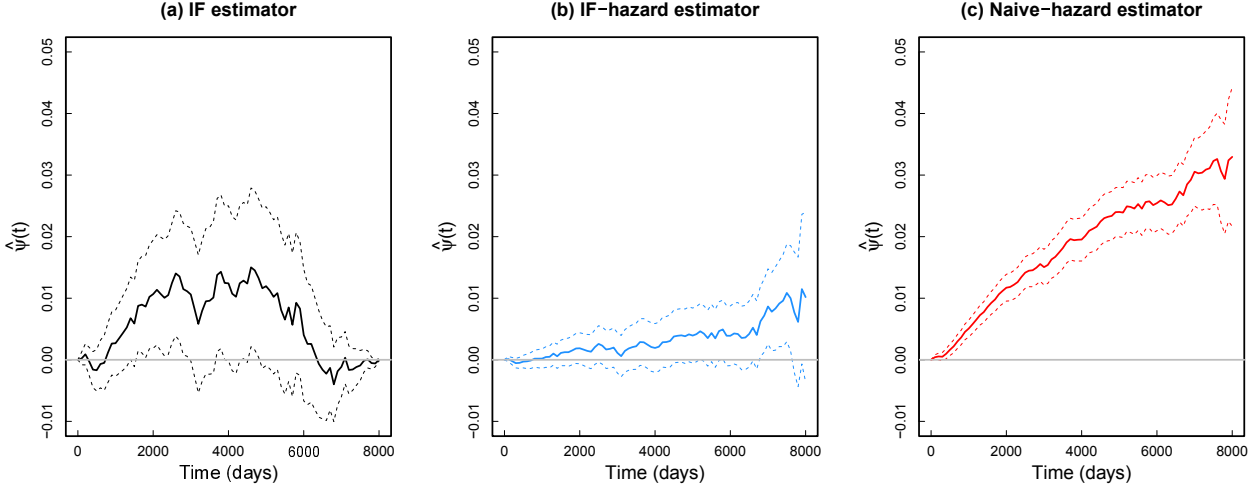}}
  \caption{\label{fig:realresult} Estimated effect of the screening on the survival probability using three different influence function-based estimators when nuisance functions are nonparametrically identified. Point-wise confidence intervals (dotted lines) are estimated through $500$ bootstrap samples.}
\end{figure}

The results of $\hat{\psi}^{\text{b}}_{\text{Naive-hazard}}(t)$ suggests a significantly positive intent-to-treat effect of cancer screening on survival probability, resulting in about a 2\% survival difference at $t=4000$ days assuming unmeasured confounding and ignorable censoring (A4*); however, with the same censoring assumption, $\hat{\psi}^{\text{b}}_{\text{IF-hazard}} (t)$ suggests an attenuated causal effect compared to the results from $\hat{\psi}^{\text{b}}_{\text{Naive-hazard}}(t)$ while the effect is still positive and increasing with time. This suggests that part of the effect observed from $\hat{\psi}^{\text{b}}_{\text{Naive-hazard}}(t)$ might be attributed to unmeasured confounding factors. On the other hand, results from $\hat{\psi}^{\text{b}}(t)$ indicate that ignoring the correlation between censoring and failure time later in follow-up  might overestimate the causal effect on survival probability.


In fact, among the $101,442~(71.2\%)$ subjects who were censored, $83,914~(82.7\%)$ subjects were due to end of study, $16,568~(16.3\%)$ subjects refused to answer, and $960~(0.1\%)$ subjects were censored for other reasons. Therefore, it is more plausible to assume ignorable censoring than nonignorable censoring in this study. Hence, our proposed estiamtor $\hat{\psi}^{\text{b}}_{\text{IF-hazard}} (t)$ seems to be a  reasonable choice for this study since it is both valid uncer ignorable censoring and is   also robust to unmeasured confounding due to noncompliance.

\section{Discussion}

The goal of this work was to provide a flexible approach for the estimation of the causal effect of treatment on survival outcomes that also accounts for unmeasured confounding. To this end, we have introduced novel influence function-based estimators that leverage instruments in the context of survival outcomes. We discuss identification assumptions and provide several estimators that are valid under different censoring assumptions. Our estimators are flexible and exhibit doubly robust properties. Further, they allow slow convergence of $n^{1/4}$ for each nuisance function, and allow both binary and continuous instruments. Our approach encourages the use of machine learning methods instead of less flexible parametric modeling.   



In this work, we have assumed that we have a valid instrument; however, valid instruments are often not available. 
~\cite{tchetgen2017genius} and ~\cite{schooling2019survival} very nicely tackle the issue of invalid instruments for survival outcomes. 
In future work, we plan to assess the sensitivity of our estimators to violations of the IV assumptions. Further, in this work, we only consider right censored survival outcomes. Future work is planned to address interval-censoring, truncation, and competing risks, which often arise in clinical studies.


\section*{Acknowledgement}
The	authors	thank the	National	Cancer	Institute (NCI)	for	access	to
NCI's	data	collected	by	the	Prostate,	Lung,	Colorectal	and	Ovarian	(PLCO)	Cancer	Screening	Trial.	The
statements	contained	herein	are	solely	those	of	the	authors	and	do	not	represent	or	imply	concurrence	or
endorsement	by	NCI. 

\bibliographystyle{biom} \bibliography{reference}

\begin{thebibliography}{}

\bibitem[\protect\citeauthoryear{Andersen, Syriopoulou, and Parner}{Andersen
  et~al.}{2017}]{andersen2017causal}
Andersen, P.~K., Syriopoulou, E., and Parner, E.~T. (2017).
\newblock Causal inference in survival analysis using pseudo-observations.
\newblock {\em Statistics in medicine} {\bf 36,} 2669--2681.

\bibitem[\protect\citeauthoryear{Angrist, Imbens, and Rubin}{Angrist
  et~al.}{1996}]{angrist1996identification}
Angrist, J.~D., Imbens, G.~W., and Rubin, D.~B. (1996).
\newblock Identification of causal effects using instrumental variables.
\newblock {\em Journal of the American Statistical Association} {\bf 91,}
  444--455.

\bibitem[\protect\citeauthoryear{Austin}{Austin}{2014}]{austin2014use}
Austin, P.~C. (2014).
\newblock The use of propensity score methods with survival or time-to-event
  outcomes: reporting measures of effect similar to those used in randomized
  experiments.
\newblock {\em Statistics in medicine} {\bf 33,} 1242--1258.

\bibitem[\protect\citeauthoryear{Baiocchi, Cheng, and Small}{Baiocchi
  et~al.}{2014}]{baiocchi2014instrumental}
Baiocchi, M., Cheng, J., and Small, D.~S. (2014).
\newblock Instrumental variable methods for causal inference.
\newblock {\em Statistics in medicine} {\bf 33,} 2297--2340.

\bibitem[\protect\citeauthoryear{Brueckner, Titman, and Jaki}{Brueckner
  et~al.}{2019}]{brueckner2019instrumental}
Brueckner, M., Titman, A., and Jaki, T. (2019).
\newblock Instrumental variable estimation in semi-parametric additive hazards
  models.
\newblock {\em Biometrics} {\bf 75,} 110--120.

\bibitem[\protect\citeauthoryear{Burgess and Thompson}{Burgess and
  Thompson}{2013}]{burgess2013use}
Burgess, S. and Thompson, S.~G. (2013).
\newblock Use of allele scores as instrumental variables for mendelian
  randomization.
\newblock {\em International journal of epidemiology} {\bf 42,} 1134--1144.

\bibitem[\protect\citeauthoryear{Cheng and Wang}{Cheng and
  Wang}{2012}]{cheng2012estimating}
Cheng, Y.-J. and Wang, M.-C. (2012).
\newblock Estimating propensity scores and causal survival functions using
  prevalent survival data.
\newblock {\em Biometrics} {\bf 68,} 707--716.

\bibitem[\protect\citeauthoryear{Chernozhukov, Chetverikov, Demirer, Duflo,
  Hansen, Newey, and Robins}{Chernozhukov
  et~al.}{2018}]{chernozhukov2018double}
Chernozhukov, V., Chetverikov, D., Demirer, M., Duflo, E., Hansen, C., Newey,
  W., and Robins, J. (2018).
\newblock Double/debiased machine learning for treatment and structural
  parameters.
\newblock {\em The Econometrics Journal} {\bf 21,} C1--C68.

\bibitem[\protect\citeauthoryear{Cole and Hern{\'a}n}{Cole and
  Hern{\'a}n}{2004}]{cole2004adjusted}
Cole, S.~R. and Hern{\'a}n, M.~A. (2004).
\newblock Adjusted survival curves with inverse probability weights.
\newblock {\em Computer methods and programs in biomedicine} {\bf 75,} 45--49.

\bibitem[\protect\citeauthoryear{D{\'\i}az}{D{\'\i}az}{2019}]{diaz2019statistical}
D{\'\i}az, I. (2019).
\newblock Statistical inference for data-adaptive doubly robust estimators with
  survival outcomes.
\newblock {\em Statistics in Medicine} {\bf 38,} 2735--2748.

\bibitem[\protect\citeauthoryear{D{\'\i}az}{D{\'\i}az}{2020}]{diaz2020machine}
D{\'\i}az, I. (2020).
\newblock Machine learning in the estimation of causal effects: targeted
  minimum loss-based estimation and double/debiased machine learning.
\newblock {\em Biostatistics} {\bf 21,} 353--358.

\bibitem[\protect\citeauthoryear{D{\'\i}az, Williams, Hoffman, and
  Schenck}{D{\'\i}az et~al.}{2020}]{diaz2020non}
D{\'\i}az, I., Williams, N., Hoffman, K.~L., and Schenck, E.~J. (2020).
\newblock Non-parametric causal effects based on longitudinal modified
  treatment policies.
\newblock {\em arXiv preprint arXiv:2006.01366} .

\bibitem[\protect\citeauthoryear{Dukes, Martinussen, Tchetgen~Tchetgen, and
  Vansteelandt}{Dukes et~al.}{2019}]{dukes2019doubly}
Dukes, O., Martinussen, T., Tchetgen~Tchetgen, E.~J., and Vansteelandt, S.
  (2019).
\newblock On doubly robust estimation of the hazard difference.
\newblock {\em Biometrics} {\bf 75,} 100--109.

\bibitem[\protect\citeauthoryear{Fr{\"o}lich}{Fr{\"o}lich}{2007}]{frolich2007nonparametric}
Fr{\"o}lich, M. (2007).
\newblock Nonparametric iv estimation of local average treatment effects with
  covariates.
\newblock {\em Journal of Econometrics} {\bf 139,} 35--75.

\bibitem[\protect\citeauthoryear{Hern{\'a}n and Robins}{Hern{\'a}n and
  Robins}{2006}]{hernan2006instruments}
Hern{\'a}n, M.~A. and Robins, J.~M. (2006).
\newblock Instruments for causal inference: an epidemiologist's dream?
\newblock {\em Epidemiology} pages 360--372.

\bibitem[\protect\citeauthoryear{Huling, Yu, and O'Malley}{Huling
  et~al.}{2019}]{huling2019instrumental}
Huling, J.~D., Yu, M., and O'Malley, A.~J. (2019).
\newblock Instrumental variable based estimation under the semiparametric
  accelerated failure time model.
\newblock {\em Biometrics} {\bf 75,} 516--527.

\bibitem[\protect\citeauthoryear{Kang, Schafer, et~al\mbox{.}}{Kang
  et~al.}{2007}]{kang2007demystifying}
Kang, J.~D., Schafer, J.~L., et~al. (2007).
\newblock Demystifying double robustness: A comparison of alternative
  strategies for estimating a population mean from incomplete data.
\newblock {\em Statistical science} {\bf 22,} 523--539.

\bibitem[\protect\citeauthoryear{Kennedy}{Kennedy}{2016}]{kennedy2016semiparametric}
Kennedy, E.~H. (2016).
\newblock Semiparametric theory and empirical processes in causal inference.
\newblock In {\em Statistical causal inferences and their applications in
  public health research}, pages 141--167. Springer.

\bibitem[\protect\citeauthoryear{Kennedy}{Kennedy}{2019}]{kennedy2019nonparametric}
Kennedy, E.~H. (2019).
\newblock Nonparametric causal effects based on incremental propensity score
  interventions.
\newblock {\em Journal of the American Statistical Association} {\bf 114,}
  645--656.

\bibitem[\protect\citeauthoryear{Kennedy, Lorch, and Small}{Kennedy
  et~al.}{2019}]{kennedy2019robust}
Kennedy, E.~H., Lorch, S., and Small, D.~S. (2019).
\newblock Robust causal inference with continuous instruments using the local
  instrumental variable curve.
\newblock {\em Journal of the Royal Statistical Society: Series B (Statistical
  Methodology)} {\bf 81,} 121--143.

\bibitem[\protect\citeauthoryear{Kianian, Kim, Fine, and Peng}{Kianian
  et~al.}{2019}]{kianian2019causal}
Kianian, B., Kim, J.~I., Fine, J.~P., and Peng, L. (2019).
\newblock Causal proportional hazards estimation with a binary instrumental
  variable.
\newblock {\em arXiv preprint arXiv:1901.11050} .

\bibitem[\protect\citeauthoryear{Li and Lu}{Li and Lu}{2015}]{li2015bayesian}
Li, G. and Lu, X. (2015).
\newblock A bayesian approach for instrumental variable analysis with censored
  time-to-event outcome.
\newblock {\em Statistics in medicine} {\bf 34,} 664--684.

\bibitem[\protect\citeauthoryear{Li, Fine, and Brookhart}{Li
  et~al.}{2015}]{li2015instrumental}
Li, J., Fine, J., and Brookhart, A. (2015).
\newblock Instrumental variable additive hazards models.
\newblock {\em Biometrics} {\bf 71,} 122--130.

\bibitem[\protect\citeauthoryear{MacKenzie, Tosteson, Morden, Stukel, and
  O'Malley}{MacKenzie et~al.}{2014}]{mackenzie2014using}
MacKenzie, T.~A., Tosteson, T.~D., Morden, N.~E., Stukel, T.~A., and O'Malley,
  A.~J. (2014).
\newblock Using instrumental variables to estimate a cox’s proportional
  hazards regression subject to additive confounding.
\newblock {\em Health Services and Outcomes Research Methodology} {\bf 14,}
  54--68.

\bibitem[\protect\citeauthoryear{Mart{\'\i}nez-Camblor, Mackenzie, Staiger,
  Goodney, and O’Malley}{Mart{\'\i}nez-Camblor
  et~al.}{2019}]{martinez2019adjusting}
Mart{\'\i}nez-Camblor, P., Mackenzie, T., Staiger, D.~O., Goodney, P.~P., and
  O’Malley, A.~J. (2019).
\newblock Adjusting for bias introduced by instrumental variable estimation in
  the cox proportional hazards model.
\newblock {\em Biostatistics} {\bf 20,} 80--96.

\bibitem[\protect\citeauthoryear{Mauro, Kennedy, and Nagin}{Mauro
  et~al.}{2018}]{mauro2018instrumental}
Mauro, J.~A., Kennedy, E.~H., and Nagin, D. (2018).
\newblock Instrumental variable methods using dynamic interventions.
\newblock {\em arXiv preprint arXiv:1811.01301} .

\bibitem[\protect\citeauthoryear{Moore and van~der Laan}{Moore and van~der
  Laan}{2009}]{moore2009increasing}
Moore, K.~L. and van~der Laan, M.~J. (2009).
\newblock Increasing power in randomized trials with right censored outcomes
  through covariate adjustment.
\newblock {\em Journal of biopharmaceutical statistics} {\bf 19,} 1099--1131.

\bibitem[\protect\citeauthoryear{Neyman}{Neyman}{1923}]{neyman1923application}
Neyman, J. (1923).
\newblock On the application of probability theory to agricultural experiments.
  essay on principles. section 9 (with discussion) translated in statistical
  sciences.
\newblock {\em Statistical Science} pages 465--472.

\bibitem[\protect\citeauthoryear{Nordestgaard, Palmer, Benn, Zacho,
  Tybj{\ae}rg-Hansen, Smith, and Timpson}{Nordestgaard
  et~al.}{2012}]{nordestgaard2012effect}
Nordestgaard, B.~G., Palmer, T.~M., Benn, M., Zacho, J., Tybj{\ae}rg-Hansen,
  A., Smith, G.~D., and Timpson, N.~J. (2012).
\newblock The effect of elevated body mass index on ischemic heart disease
  risk: causal estimates from a mendelian randomisation approach.
\newblock {\em PLoS medicine} {\bf 9,} e1001212.

\bibitem[\protect\citeauthoryear{Ogburn, Rotnitzky, and Robins}{Ogburn
  et~al.}{2015}]{ogburn2015doubly}
Ogburn, E.~L., Rotnitzky, A., and Robins, J.~M. (2015).
\newblock Doubly robust estimation of the local average treatment effect curve.
\newblock {\em Journal of the Royal Statistical Society: Series B (Statistical
  Methodology)} {\bf 77,} 373--396.

\bibitem[\protect\citeauthoryear{Prorok, Andriole, Bresalier, Buys, Chia,
  Crawford, Fogel, Gelmann, Gilbert, Hasson, et~al\mbox{.}}{Prorok
  et~al.}{2000}]{prorok2000design}
Prorok, P.~C., Andriole, G.~L., Bresalier, R.~S., Buys, S.~S., Chia, D.,
  Crawford, E.~D., Fogel, R., Gelmann, E.~P., Gilbert, F., Hasson, M.~A.,
  et~al. (2000).
\newblock Design of the prostate, lung, colorectal and ovarian (plco) cancer
  screening trial.
\newblock {\em Controlled clinical trials} {\bf 21,} 273S--309S.

\bibitem[\protect\citeauthoryear{Rassen, Brookhart, Glynn, Mittleman, and
  Schneeweiss}{Rassen et~al.}{2009}]{rassen2009instrumental}
Rassen, J.~A., Brookhart, M.~A., Glynn, R.~J., Mittleman, M.~A., and
  Schneeweiss, S. (2009).
\newblock Instrumental variables i: instrumental variables exploit natural
  variation in nonexperimental data to estimate causal relationships.
\newblock {\em Journal of clinical epidemiology} {\bf 62,} 1226--1232.

\bibitem[\protect\citeauthoryear{Richardson, Hudgens, Fine, and
  Brookhart}{Richardson et~al.}{2017}]{richardson2017nonparametric}
Richardson, A., Hudgens, M.~G., Fine, J.~P., and Brookhart, M.~A. (2017).
\newblock Nonparametric binary instrumental variable analysis of competing
  risks data.
\newblock {\em Biostatistics} {\bf 18,} 48--61.

\bibitem[\protect\citeauthoryear{Robins, Li, Tchetgen, van~der Vaart,
  et~al\mbox{.}}{Robins et~al.}{2008}]{robins2008higher}
Robins, J., Li, L., Tchetgen, E., van~der Vaart, A., et~al. (2008).
\newblock Higher order influence functions and minimax estimation of nonlinear
  functionals.
\newblock In {\em Probability and statistics: essays in honor of David A.
  Freedman}, pages 335--421. Institute of Mathematical Statistics.

\bibitem[\protect\citeauthoryear{Robins and Tsiatis}{Robins and
  Tsiatis}{1991}]{robins1991correcting}
Robins, J.~M. and Tsiatis, A.~A. (1991).
\newblock Correcting for non-compliance in randomized trials using rank
  preserving structural failure time models.
\newblock {\em Communications in statistics-Theory and Methods} {\bf 20,}
  2609--2631.

\bibitem[\protect\citeauthoryear{Rubin}{Rubin}{1974}]{rubin1974estimating}
Rubin, D.~B. (1974).
\newblock Estimating causal effects of treatments in randomized and
  nonrandomized studies.
\newblock {\em Journal of educational Psychology} {\bf 66,} 688.

\bibitem[\protect\citeauthoryear{Schooling, Lopez, Yeung, and Huang}{Schooling
  et~al.}{2019}]{schooling2019survival}
Schooling, C.~M., Lopez, P., Yeung, S.~A., and Huang, J. (2019).
\newblock Survival bias and competing risk can severely bias mendelian
  randomization studies of specific conditions.
\newblock {\em bioRxiv} page 716621.

\bibitem[\protect\citeauthoryear{Tan}{Tan}{2006}]{tan2006regression}
Tan, Z. (2006).
\newblock Regression and weighting methods for causal inference using
  instrumental variables.
\newblock {\em Journal of the American Statistical Association} {\bf 101,}
  1607--1618.

\bibitem[\protect\citeauthoryear{Tchetgen~Tchetgen, Sun, and
  Walter}{Tchetgen~Tchetgen et~al.}{2017}]{tchetgen2017genius}
Tchetgen~Tchetgen, E.~J., Sun, B., and Walter, S. (2017).
\newblock The genius approach to robust mendelian randomization inference.
\newblock {\em arXiv preprint arXiv:1709.07779} .

\bibitem[\protect\citeauthoryear{Tchetgen~Tchetgen, Walter, Vansteelandt,
  Martinussen, and Glymour}{Tchetgen~Tchetgen
  et~al.}{2015}]{tchetgen2015instrumental}
Tchetgen~Tchetgen, E.~J., Walter, S., Vansteelandt, S., Martinussen, T., and
  Glymour, M. (2015).
\newblock Instrumental variable estimation in a survival context.
\newblock {\em Epidemiology (Cambridge, Mass.)} {\bf 26,} 402.

\bibitem[\protect\citeauthoryear{Team, Gohagan, Prorok, Hayes, and Kramer}{Team
  et~al.}{2000}]{team2000prostate}
Team, P.~P., Gohagan, J.~K., Prorok, P.~C., Hayes, R.~B., and Kramer, B.-S.
  (2000).
\newblock The prostate, lung, colorectal and ovarian (plco) cancer screening
  trial of the national cancer institute: history, organization, and status.
\newblock {\em Controlled clinical trials} {\bf 21,} 251S--272S.

\bibitem[\protect\citeauthoryear{Terza, Basu, and Rathouz}{Terza
  et~al.}{2008}]{terza2008two}
Terza, J.~V., Basu, A., and Rathouz, P.~J. (2008).
\newblock Two-stage residual inclusion estimation: addressing endogeneity in
  health econometric modeling.
\newblock {\em Journal of health economics} {\bf 27,} 531--543.

\bibitem[\protect\citeauthoryear{Van~der Laan, Laan, and Robins}{Van~der Laan
  et~al.}{2003}]{van2003unified}
Van~der Laan, M.~J., Laan, M., and Robins, J.~M. (2003).
\newblock {\em Unified methods for censored longitudinal data and causality}.
\newblock Springer Science \& Business Media.

\bibitem[\protect\citeauthoryear{Vansteelandt, Joffe,
  et~al\mbox{.}}{Vansteelandt et~al.}{2014}]{vansteelandt2014structural}
Vansteelandt, S., Joffe, M., et~al. (2014).
\newblock Structural nested models and g-estimation: the partially realized
  promise.
\newblock {\em Statistical Science} {\bf 29,} 707--731.

\bibitem[\protect\citeauthoryear{Wan, Small, Bekelman, and Mitra}{Wan
  et~al.}{2015}]{wan2015bias}
Wan, F., Small, D., Bekelman, J.~E., and Mitra, N. (2015).
\newblock Bias in estimating the causal hazard ratio when using two-stage
  instrumental variable methods.
\newblock {\em Statistics in medicine} {\bf 34,} 2235--2265.

\bibitem[\protect\citeauthoryear{Wan, Small, and Mitra}{Wan
  et~al.}{2018}]{wan2018general}
Wan, F., Small, D., and Mitra, N. (2018).
\newblock A general approach to evaluating the bias of 2-stage instrumental
  variable estimators.
\newblock {\em Statistics in medicine} {\bf 37,} 1997--2015.

\bibitem[\protect\citeauthoryear{Wright, Wager, and Probst}{Wright
  et~al.}{2019}]{wright2019fast}
Wright, M.~N., Wager, S., and Probst, P. (2019).
\newblock A fast implementation of random forests.
\newblock {\em R package version 0.11} {\bf 2,}.

\bibitem[\protect\citeauthoryear{Yang, Pieper, and Cools}{Yang
  et~al.}{2020}]{yang2020semiparametric}
Yang, S., Pieper, K., and Cools, F. (2020).
\newblock Semiparametric estimation of structural failure time models in
  continuous-time processes.
\newblock {\em Biometrika} {\bf 107,} 123--136.

\bibitem[\protect\citeauthoryear{Zheng and Van Der~Laan}{Zheng and Van
  Der~Laan}{2010}]{zheng2010asymptotic}
Zheng, W. and Van Der~Laan, M.~J. (2010).
\newblock Asymptotic theory for cross-validated targeted maximum likelihood
  estimation.

\end{thebibliography}

\section*{Supporting Information}

Software in the form of \texttt{R} code to implement the parametric and nonparametric estimation for the simulation study is available at one of the authors' Github repository\footnote{\url{https://github.com/youjin1207/survivalIV}}, and the data is available through the National Cancer Institutes upon approval (\url{https://cdas.cancer.gov/plco/}).

\newpage
\def\spacingset#1{\renewcommand{\baselinestretch}%
{#1}\small\normalsize} \spacingset{1}

\setcounter{equation}{0}
\setcounter{figure}{0}
\setcounter{table}{0}
\setcounter{page}{1}
\setcounter{section}{0}
\setcounter{theorem}{0}
\setcounter{lemma}{0}
\renewcommand{\theequation}{S\arabic{equation}}
\renewcommand{\thetheorem}{S\arabic{theorem}}
\renewcommand{\thelemma}{S\arabic{lemma}}
\renewcommand{\thefigure}{S\arabic{figure}}
\renewcommand{\thesection}{S\arabic{section}}
\renewcommand{\thetable}{S\arabic{table}}
\renewcommand{\thefigure}{S\arabic{figure}}

\begin{center}
    {\LARGE\bf Supporting Information for ``Doubly Robust Nonparametric Instrumental Variable Estimators for Survival Outcomes"}
\end{center}

\date{}
\section{Proof from the main manuscript}
\begin{proof}[Proof of Lemma~1]
	\begin{eqnarray*}
		\mathbb{E}\left(  (Y^{A=1}_{t} - Y^{A=0}_{t}) \mathbb{I}(A^{Z=1} > A^{Z=0}) \right) &=& \mathbb{E} \left[  \mathbb{E}\left(  (Y^{A=1}_{t} - Y^{A=0}_{t} ) (A^{Z=1} > A^{Z=0})  | \bm{X}_{0} \right) \right] \\ 
		&=& \mathbb{E} \left[  \mathbb{E} \left( Y^{Z=1}_{t} - Y^{Z=1}_{t} | \bm{X}_{0} \right)  \right] \\
		&=& \mathbb{E} \left[   \mathbb{E}\left( Y_{t} | \bm{X}_{0}, Z= 1 \right) - \mathbb{E}\left( Y_{t} | \bm{X}_{0}, Z = 0 \right)  \right] 
	\end{eqnarray*}
	Then for each $z \in \{0, 1\}$:	
	\begin{eqnarray*}	
		\mathbb{E} \left[   \mathbb{E}\left( Y_{t} | \bm{X}_{0}, Z= z \right) \right] 
		&=& \mathbb{E} \left\{   \sum\limits_{a \in \{0,1 \}} \mathbb{E}\left(  Y_{t} | \bm{X}_{0},  Z = z, A = a \right) p(A = a | \bm{X}_{0},  Z = z) \right\} \\ 
		&=& \mathbb{E} \left\{  \sum\limits_{a \in \{0,1 \}} \mathbb{E}\left(  Y_{t} | \bm{X}_{0}, R = 1, Z = z, A = a \right) p(A = a | \bm{X}_{0},  Z = z) \right\}
	\end{eqnarray*}
	On the other hand, 
	\begin{eqnarray*}
		\mathbb{P}\left(A^{Z=1} > A^{Z=0} \right) &=& \mathbb{P}\left(  A^{Z=1} =1 ~\&~ A^{Z=0} = 0 \right)  \\ 
		&=& \mathbb{E} \left[ \mathbb{E}\left(  A^{Z=1} - A^{Z=0}  | \bm{X}_{0} \right) \right] \\
		&=& \mathbb{E} \left[ \mathbb{E}\left(  A |Z=1, \bm{X}_{0} \right) -  \mathbb{E}\left(  A |Z=0, \bm{X}_{0} \right) \right]. 
	\end{eqnarray*}
	The second line follows from the monotonicity assumption. 	
\end{proof}	

\begin{proof}[Proof of Theorem~1]
Under the identification assumptions (A1)--(A6), our target estimand of a local average treatment effect on $Y_{t}$ can be represented through the following conditional expectations:
\begin{equation}
	\begin{split}
		\psi^{\text{b}}(t)  & =:  \frac{\psi^{\text{b}, Z= 1}_{1}(t) - \psi^{\text{b}, Z=0}_{1}(t) }{ \psi^{\text{b}, Z=1}_{2} - \psi^{\text{b}, Z=0}_{2}} \\ 
		& = \frac{ \mathbb{E} \left(  \mathbb{E} \left( Y_{t} | \bm{X}_{0},  Z = 1  \right)  \right) -  \mathbb{E} \left(  \mathbb{E} \left( Y_{t} | \bm{X}_{0}, Z = 0 \right)  \right) }{\mathbb{E} \left(  \mathbb{E} \left( A | \bm{X}_{0}, Z = 1  \right)  \right) -  \mathbb{E} \left(  \mathbb{E} \left( A | \bm{X}_{0}, Z = 0 \right)  \right)} \\
		& =  \{ \mathbb{E} \left(  \mathbb{E} \left( A | \bm{X}_{0}, Z = 1  \right)  \right) -  \mathbb{E} \left(  \mathbb{E} \left( A | \bm{X}_{0}, Z = 0 \right)  \right) \}^{-1} \\ & \times \left[ \mathbb{E} \left\{  \sum\limits_{a \in \{ 0, 1 \}} \mathbb{E} \left( Y_{t} | \bm{X}_{0},  Z = 1, A = a  \right) \mathbb{P}(A = a | \bm{X}_{0}, Z = 1)  \right\} \right.  \\ &   \quad  -  \left. \mathbb{E} \left\{  \sum\limits_{a \in \{0,1 \}}  \mathbb{E} \left( Y_{t} | \bm{X}_{0},  Z = 0, A = a \right)\mathbb{P}(A = a | \bm{X}_{0}, Z  = 0)  \right\} \right]  \\
	& =  \{  \mathbb{E} \left(  \mathbb{E} \left( A | \bm{X}_{0}, Z = 1  \right)  \right) -  \mathbb{E} \left(  \mathbb{E} \left( A | \bm{X}_{0}, Z = 0 \right)  \right) \}^{-1} \\ & \times \left[ \mathbb{E} \left\{  \sum\limits_{a \in \{ 0, 1 \}} \mathbb{E} \left( Y_{t} | \bm{X}_{0},  R = 1,  Z = 1, A = a  \right) \mathbb{P}(A = a | \bm{X}_{0}, Z = 1)  \right\} \right. \  \\ & \quad  -  \left[ \mathbb{E} \left\{  \sum\limits_{a \in \{0,1 \}}  \mathbb{E} \left( Y_{t} | \bm{X}_{0},  R = 1, Z = 0, A = a \right)\mathbb{P}(A = a | \bm{X}_{0}, Z  = 0)  \right\} \right].
	\end{split}
\end{equation}

Consider the first term of $\psi^{\text{b}, Z = 1}_{1}(t) := \psi^{\text{b}, Z = 1, A = 1}_{1}(t) + \psi^{\text{b}, Z = 1, A = 0}_{1}(t)$.
\begin{eqnarray*}
	\label{eq:phi1_binary}
	\psi^{\text{b},Z=1, A = 1}_{1}(t) &:=&    \mathbb{E} (  \mathbb{E} \left( Y_{t} | \bm{X}_{0},  R = 1,  Z = 1, A = 1  \right) \mathbb{P}(A = 1 | \bm{X}_{0}, Z = 1)  ) \\ 
	&= &  
	\int\limits_{\mathcal{X}_{0}} \mathbb{E} \left( Y_{t} | \bm{X}_{0} = \bm{x}_{0}, R = 1,  Z = 1, A = 1 \right) d\mathbb{P}(A = 1 | \bm{x}_{0}, Z = 1 ) d \mathbb{P}(\bm{x_{0}}).
\end{eqnarray*}
Therefore, 
\begin{eqnarray*}
	\IF\left( \psi^{\text{b}, Z=1, A = 1}_{1}(t) \right) &=& \int\limits_{\mathcal{X}_{0}} \IF\left\{ \mu_{t,1,1}(\bm{x}_{0}) \right\} \pi_{1}(\bm{x}_{0})  d\mathbb{P} \left( \bm{x}_{0} \right)  +
	\int\limits_{\mathcal{X}_{0}} \mu_{t,1,1}(\bm{x}_{0})\IF \left\{ \pi_{1} (\bm{x}_{0}) \right\} d\mathbb{P} \left( \bm{x}_{0} \right)  \\ 
	&&\quad + \int\limits_{\mathcal{X}_{0}}  \mu_{t,1,1}(\bm{x}_{0}) \pi_{1} (\bm{x}_{0})  \IF \left\{  d\mathbb{P} \left( \bm{x}_{0} \right) \right\} \\ 
	&=:& (A) + (B) + (C),
\end{eqnarray*}
where

\begin{eqnarray*}
	(A)&=&\int\limits_{\mathcal{X}_{0}} \IF\left\{ \mu_{t,1,1}(\bm{x}_{0}) \right\}  \pi_{1}(\bm{x}_{0})   d\mathbb{P} \left( \bm{x}_{0} \right) \\
	& =&  \int\limits_{\mathcal{X}_{0}} \frac{\mathbb{I} \left( (\bm{X}_{0}, R, Z, A) = (\bm{x}_{0}, 1, 1, 1)  \right) }{ d\mathbb{P} \left( \bm{X}_{0}, R = 1, Z = 1, A =1 \right) } \left\{ Y_{t} - \mu_{t,1,1}(\bm{x}_{0}) \right\} \pi_{1} (\bm{x}_{0}) d \mathbb{P}\left( \bm{x}_{0} \right) 
	\\ &= & \frac{(Y_{t} - \mu_{t,1,1}(\bm{x}_{0})) \mathbb{I}(R = 1, Z = 1, A = 1) \pi_{1}(\bm{x}_{0})}{ d\mathbb{P}(R = 1 |  \bm{x}_{0}, Z = 1, A = 1) \pi_{1} (\bm{x}_{0}) d \mathbb{P}(Z = 1 | \bm{x}_{0})   } \\ 
	&=& \frac{(Y_{t} - \mu_{t,1,1}(\bm{x}_{0})) \mathbb{I}(R = 1, Z = 1, A = 1)}{  \omega(  \bm{x}_{0}, Z = 1, A = 1) \delta_{1}(\bm{x}_{0}) };
\end{eqnarray*}

\begin{eqnarray*}
	(B)&=&\int\limits_{\mathcal{X}_{0}} \mu_{t,1,1}(\bm{x}_{0})  \IF\left\{ \pi_{1} (\bm{x}_{0})\right\}    d\mathbb{P} \left( \bm{x}_{0} \right) \\
	& =& \int\limits_{\mathcal{X}_{0}}  \mu_{t,1,1}(\bm{x}_{0}) \frac{\mathbb{I} \left( (\bm{X}_{0}, Z ) = (\bm{x}_{0}, 1)  \right)}{ d\mathbb{P}(\bm{x}_{0}, Z = 1) } (A - \pi_{1} (\bm{x}_{0}) ) d\mathbb{P} (\bm{x}_{0})\\
	& = & \mu_{t,1,1}(\bm{x}_{0})\frac{\mathbb{I}(Z = 1)}{ \delta_{1}(\bm{x}_{0} )} \left( A - \pi_{1} (\bm{x}_{0}) \right);
\end{eqnarray*}

\begin{eqnarray*} 
	(C)& = & \int\limits_{\mathcal{X}_{0}}  \mu_{t,1,1}(\bm{x}_{0}) \pi_{1}(\bm{x}_{0})  \IF \left\{  d\mathbb{P} \left( \bm{x}_{0} \right) \right\} \\
	&& = \int\limits_{\mathcal{X}_{0}} \mu_{t,1,1}(\bm{x}_{0}) \pi_{1}(\bm{x}_{0}) \left\{  \mathbb{I}\left( \bm{X}_{0} = \bm{x}_{0} \right) - d\mathbb{P}\left( \bm{x}_{0} \right) \right\} \\
	& = &  \mu_{t,1,1}(\bm{x}_{0}) \pi_{1}(\bm{x}_{0})  - \mathbb{E}\left( \mu_{t,1,1}(\bm{x}_{0}) \pi_{1}(\bm{x}_{0}) \right) \\ 
	&=& \mu_{t,1,1}(\bm{x}_{0}) \pi_{1}(\bm{x}_{0}) - \psi^{\text{b}, Z = 1, A = 1}_{1}(t).
\end{eqnarray*}

Then the above equations finally lead to:
\begin{eqnarray*}
	\IF\left( \psi^{\text{b}, Z=1, A  = 1}_{1}(t)\right) &=&    \frac{\mathbb{I}(R = 1, Z = 1)}{\omega_{1,1}(\bm{X}_{0}) \delta_{1}(\bm{x}_{0})} A (Y_{t} - \mu_{t,1,1}(\bm{X}_{0}) ) + \frac{\mathbb{I}(Z = 1)}{\delta_{1}(\bm{x}_{0})} \mu_{t,1, 1}(\bm{X}_{0}) (A -  \pi_{1} (\bm{x}_{0})) 
	\\ && \quad + \mu_{t,1,1}(\bm{X}_{0}) \pi_{1}(\bm{x}_{0}) - \psi^{\text{b}, Z = 1, A = 1}(t),
\end{eqnarray*}
We can similarly construct \IF($\psi^{\text{b}, Z= z, A=a}_{1}(t)$) for any $a, z \in \{0,1 \}$. Using the same notation defined in Theorem~1 in the main text, the influence function of $\psi^{\text{b}, Z = 1}_{1}(t)$ is given by:
\begin{eqnarray*}
	\IF ( \psi^{\text{b}, Z=1}_{1}(t) ) & =&  \left(  \mu_{t, 1, 1} \pi_{1} +  \mu_{t,1,0} (1-\pi_{1}  \right) \\ &  +&  \frac{\mathbb{I}(Z = 1)}{\delta_{1} } \left[   \frac{R A}{\omega_{1,1} } (Y_{t} - \mu_{t,1, 1}) + \mu_{t,1, 1}(A - \pi_{1}) \right]  \\ 
	& + & \frac{\mathbb{I}(Z = 1)}{\delta_{1}} \left[   \frac{R (1-A)}{\omega_{1,0} } (Y_{t} - \mu_{t,1, 0}) + \mu_{t,1, 0}((1-A) - (1-\pi_{1}) \right] - \psi^{\text{b}, Z = 1}_{1}(t) \\
	&=& M_{1}(t) - \psi^{\text{b}, Z = 1}_{1}(t)
\end{eqnarray*}

On the other hand, the influence function of $\psi^{\text{b}, Z=1}_{2} = \mathbb{E}( \mathbb{E}(A | \bm{X}_{0}, Z=1))$ is given by:
\begin{eqnarray*}
	\label{eq:phi2_if}
	\mbox{\IF}\left(\psi^{\text{b}, Z=1}_{2}\right)  &= & \frac{\mathbb{I}(Z = 1)}{\delta_{1}} \left\{ A - \pi_{1}\right\} +  \pi_{1} - \psi^{\text{b}, Z=1}_{2} \\ 
	&=&  \Pi_{1} - \psi^{\text{b}, Z = 1}_{2} 
\end{eqnarray*}

We can similarly derive $\phi^{\text{b}}_{1}(t)$ := $\IF (\psi^{\text{b}, Z=0}_{1}(t) )$ and $\phi^{\text{b}}_{2}$ $:= \IF\left(\psi^{\text{b}, Z=0}_{2}\right)$. Then $\IF (\psi^{\text{b}}_{1}(t) ) = \IF (\psi^{\text{b}, Z=1}_{1}(t) ) + \IF (\psi^{\text{b}, Z=0}_{1}(t) )$ and $\IF (\psi^{\text{b}}_{2} ) = \IF (\psi^{\text{b}, Z=1}_{2} ) + \IF (\psi^{\text{b}, Z=0}_{2} )$.
To sum up, with the same notations as defined in Theorem~1 we have:
\begin{eqnarray*}
	\IF(\psi^{\text{b}}(t)) &=& \frac{ \IF(\psi^{\text{b}}_{1} (t)) \psi^{\text{b}}_{2} - \psi^{\text{b}}_{1} (t) \IF (\psi^{\text{b}}_{2}) }{(\psi^{\text{b}}_{2})^2} \\ 
	&=&  \frac{ \phi^{\text{b}}_{1}(t) \psi^{\text{b}}_{2} - \psi^{\text{b}}_{1} (t) \phi^{\text{b}}_{2} }{(\psi^{\text{b}}_{2})^2} \\ 
	&=&  \frac{ (M_{1}(t) - M_{0}(t) - \psi^{\text{b}}_{1} (t) ) \psi^{\text{b}}_{2} - \psi^{\text{b}}_{1} (t) (\Pi_{1} - \Pi_{0} - \psi^{\text{b}}_{2} )   }{   (\psi^{\text{b}}_{2})^2 } \\ 
	&=& \frac{M_{1}(t) - M_{0}(t) - \psi^{\text{b}}_{1} (t) }{\psi^{\text{b}}_{2} } - \frac{\psi^{\text{b}}_{1} (t)}{ (\psi^{\text{b}}_{2} )^2}(\Pi_{1} - \Pi_{0}) + \frac{\psi^{\text{b}}_{1}(t)}{\psi^{\text{b}}_{2}} \\ 
	& =& \left\{   (M_{1}(t) -M_{0}(t)) - \psi^{\text{b}}(t) (\Pi_{1} - \Pi_{0} ) \right\} / \psi^{\text{b}}_{2}.
\end{eqnarray*}	
This finally yields an influence function-based estimator for $\psi^{\text{b}}(t)$:
\begin{eqnarray*}
	 \hat{\psi}^{\text{b}}(t) =\mathbb{P}_{n}\left( M_{1} (t; \hat{\Theta}) - M_{0} (t; \hat{\Theta}) \right) / \mathbb{P}_{n} \left( \Pi_{1}(\hat{\Theta}) - \Pi_{0}(\hat{\Theta}) \right).
\end{eqnarray*}
\end{proof}

\begin{proof}[Proof of Proposition~2]
By \cite{diaz2019statistical}, under the idenfication assumptions described in the main text, the efficient influence function of $\mathbb{E}(Y_{t} | \bm{X}_{0}, Z = z,  A = a)$ is given by:
\begin{eqnarray*}
&& -\sum\limits_{k=1}^{t} \frac{\mathbb{I}(A = a) \mathbb{I}(Z = z) \mathbb{I}(Y_{k-1} = 1) }{ \delta_{z}(\bm{X}_{0}) \pi_{z}(a; \bm{X}_{0}) G_{k-1, z, a}(\bm{X}_{0})} \frac{S_{t, z, a}(\bm{X}_{0})}{S_{k, z, a}(\bm{X}_{0})} \{ R \mathbb{I}(Y_{k} = 0) - h_{k, z, a}(\bm{X}_{0}) \} \\ 
&&  + S_{k, z, a}(\bm{X}_{0}) - \mathbb{E}(Y_{t} | \bm{X}_{0}, Z = z, A = a).
\end{eqnarray*}
In Proof of Theorem~1, replace the influence function of $\mathbb{E}(Y_{t} | \bm{X}_{0}, Z = z, A = a)$ by the function above.
\end{proof}

\begin{proof}[Proof of Theorem~2]
	Consider an influence function-based estimator $\hat{\psi}^{\text{b}}(t)$ as follows:
	\begin{eqnarray*}
		\hat{\psi}^{\text{b}}(t) &=& \mathbb{P}_{n}\left( M_{1} (t;\hat{\Theta}) - M_{0} (t;\hat{\Theta}) \right) / \mathbb{P}_{n} \left( \Pi_{1}(\hat{\Theta}) - \Pi_{0}(\hat{\Theta}) \right)
		\\ &=& \mathbb{P}_{n} \left( \hat{\phi}^{\text{b}}_{1}(t) + \hat{\psi}^{\text{b}}_{1}(t) \right) / \mathbb{P}_{n} \left( \hat{\phi}^{\text{b}}_{2} + \hat{\psi}^{\text{b}}_{2} \right)
		\\ &=:& \mathbb{P}_{n} \left( \hat{\phi}^{\text{b}}_{1}(t; \hat{\Theta}) + \hat{\psi}^{\text{b}}_{1}(t; \hat{\Theta}) \right) / \mathbb{P}_{n} \left( \hat{\phi}^{\text{b}}_{2}(\hat{\Theta}) + \hat{\psi}^{\text{b}}_{2} (\hat{\Theta})\right) ,		
	\end{eqnarray*}
where $\Theta$ denotes a set of nuisance functions and $\hat{\Theta}$ is its estimates.
	Then we have:
	\begin{eqnarray}
		\label{eq:substract}	
		&& \hat{\psi}^{\text{b}}(t; \hat{\Theta}) - \psi^{\text{b}}(t; \Theta) =  \frac{ \mathbb{P}_{n} \left( \phi^{\text{b}}_{1}(t; \hat{\Theta}) + \psi^{\text{b}}_{1}(t; \hat{\Theta}) \right) }{\mathbb{P}_{n} \left( \phi^{\text{b}}_{2}(\hat{\Theta}) + \psi^{\text{b}}_{2} (\hat{\Theta})\right)} -   \frac{ \mathbb{P} \left( \psi^{\text{b}}_{1}(t; \Theta)  \right)  }{\mathbb{P} \left( \psi^{\text{b}}_{2}(\Theta)  \right)}  \\ 
		&=& \xi^{-1}_{n}\left[ \mathbb{P}_{n} ( \psi^{\text{b}}_{1}(t; \hat{\Theta}) + \phi^{\text{b}}_{1}(t; \hat{\Theta})  ) -   \mathbb{P} ( \psi^{\text{b}}_{1}(t; \Theta)  ) - \psi^{\text{b}}(t)  \left(  \mathbb{P}_{n} ( \psi^{\text{b}}_{2}( \hat{\Theta}) + \psi^{\text{b}}_{2}(\hat{\Theta}) ) -   \mathbb{P} ( \psi^{\text{b}}_{2}(\Theta)  ) \right)  \right] \nonumber \\ 
		&=&   \xi^{-1}_{n} \left[ (\mathbb{P}_{n} - \mathbb{P}) \left\{ \psi^{\text{b}}_{1}(t; \hat{\Theta}) + \phi^{\text{b}}_{1}(t; \hat{\Theta})  - \psi^{\text{b}}_{1}(t; \Theta)  - \phi^{\text{b}}_{1}(t; \Theta)\right\} \right. \nonumber \\
		&& \quad \quad  \quad -   \left.  \psi^{\text{b}}(t)  (\mathbb{P}_{n} - \mathbb{P}) \left\{ \psi^{\text{b}}_{2}(\hat{\Theta}) + \phi^{\text{b}}_{2}(\hat{\Theta})   -    \psi^{\text{b}}_{2}(\Theta)- \phi^{\text{b}}_{2}(\Theta)    \right\}  \right] \nonumber \\ 
		\quad &+&   \xi^{-1}_{n} \left[ (\mathbb{P}_{n} - \mathbb{P})  \left\{ \psi^{\text{b}}_{1}(t; \Theta)  + \phi^{\text{b}}_{1}(t; \Theta) - \psi^{\text{b}}(t)    (\psi^{\text{b}}_{2}(\Theta) + \phi^{\text{b}}_{2}(\Theta) \right\}    \right] \nonumber \\ 
		\quad & + &  \xi^{-1}_{n}\left[ \mathbb{P} ( \phi^{\text{b}}_{1}(t; \hat{\Theta}) +  \psi^{\text{b}}_{1}(t; \hat{\Theta})  - \psi^{\text{b}}_{1}(t; \Theta)  ) - \psi^{\text{b}}(t)  \left(  \mathbb{P} ( \phi^{\text{b}}_{2}(\hat{\Theta}) +  \psi^{\text{b}}_{2}( \hat{\Theta}) - \psi^{\text{b}}_{2}(\Theta)  ) \right)  \right]  \nonumber,
	\end{eqnarray}
where $\xi_{n} = \mathbb{P}_{n} \left( \phi^{\text{b}}_{2}(\hat{\Theta}) + \psi^{\text{b}}_{2} (\hat{\Theta})\right)$.
In the last equation, the first term is $o_{\mathbb{P}}(n^{-1/2})$ by (C1) and the second term is asymptotically normal by the central limit theorem applied to $\psi^{\text{b}}_{1}(t; \Theta)  + \phi^{\text{b}}_{1}(t; \Theta) - \psi^{\text{b}}(t)    (\psi^{\text{b}}_{2}(\Theta) + \phi^{\text{b}}_{2}(\Theta) = \phi^{\text{b}}_{1}(t;\Theta) - \psi^{\text{b}}(t;\Theta)\phi^{\text{b}}_{2}(\Theta)$. So we will focus on the last term.

Let $\phi^{\text{b}, z, a}_{1}(t; \Theta)$ and $\phi^{\text{b}, z}_{2}(\Theta)$ denote the influence function of $\psi^{\text{b}, z, a}_{1}(t; \Theta)$ and $\psi^{\text{b}, z}_{2}(\Theta)$, respectively. Then by  iterative double conditional expectations under (A4), we have the following factorizations of the nuisance functions:
\begin{eqnarray*} 
		&& \mathbb{P}( \psi^{\text{b}, Z=1, A = 1}_{1}(t; \hat{\Theta}) + \phi^{\text{b}, Z =1, A = 1}_{1}(t; \hat{\Theta}) ) - \mathbb{P} ( \psi^{\text{b}, Z=1, A = 1}_{1}(t; \Theta) )  \\ && = \mathbb{E} \left[ \hat{\mu}_{t, 1, 1} \hat{\pi}_{1} +  \frac{\mathbb{I}(R = 1, Z = 1)}{ \hat{\omega}_{1,1} \hat{\delta}_{1} } A(Y_{t} - \hat{\mu}_{t, 1, 1} )  + \frac{\mathbb{I}(Z=1)}{\hat{\delta}_{1}}  \hat{\mu}_{t,1,1}(A - \hat{\pi}_{1}) -   \mu_{t,1,1} \pi_{1}  \right]  \\ 
		&& = \mathbb{E}\left[  \frac{\mathbb{I}(R= 1, Z = 1) }{ \hat{\omega}_{1,1} \hat{\delta}_{1} } A  \left( \mu_{t, 1, 1}- \hat{\mu}_{t,1,1} \right) + \frac{\mathbb{I}(Z=1)}{\hat{\delta}_{1}}  \hat{\mu}_{t,1,1} (A - \hat{\pi}_{1})  \right]  + \mathbb{E}\left[  \hat{\mu}_{t,1,1} \hat{\pi}_{1} - \mu_{t,1,1} \pi_{1} \right] \\ 
		&& = \mathbb{E}\left[  \frac{  \omega_{1,1} \mathbb{I}(Z = 1) }{ \hat{\omega}_{1,1} \hat{\delta}_{1} } A  \left( \mu_{t, 1, 1}- \hat{\mu}_{t,1,1} \right) + \frac{\mathbb{I}(Z=1)}{\hat{\delta}_{1}}  \hat{\mu}_{t,1,1} (A - \hat{\pi}_{1})  \right]  + \mathbb{E}\left[  \hat{\mu}_{t,1,1} \hat{\pi}_{1} - \mu_{t,1,1} \pi_{1} \right] \\ 
		&& = \mathbb{E}\left[  \frac{  \omega_{1,1} \mathbb{I}(Z = 1) }{ \hat{\omega}_{1,1} \hat{\delta}_{1} } \pi_{1}  \left( \mu_{t, 1, 1}- \hat{\mu}_{t,1,1} \right) + \frac{\mathbb{I}(Z=1)}{\hat{\delta}_{1}}  \hat{\mu}_{t,1,1} (\pi_{1} - \hat{\pi}_{1})  \right] + \mathbb{E}\left[  \hat{\mu}_{t,1,1} \hat{\pi}_{1} - \mu_{t,1,1} \pi_{1} \right] \\ 
		&& = \mathbb{E}\left[  \frac{\omega_{1,1} \delta_{1} }{ \hat{\omega}_{1,1} \hat{\delta}_{1} } \pi_{1} \left( \mu_{t,1,1} - \hat{\mu}_{t,1,1}  \right) + \frac{\delta_{1}}{\hat{\delta}_{1}}  \hat{\mu}_{t,1,1}(\pi_{1} - \hat{\pi}_{1} ) \right]  + \mathbb{E}\left[  \hat{\mu}_{t,1,1} \hat{\pi}_{1} - \mu_{t,1,1} \pi_{1} \right] \\ 
		&& = \mathbb{E} \left[ \frac{ \omega_{1,1} \delta_{1} - \hat{\omega}_{1,1} \hat{\delta}_{1} }{ \hat{\omega}_{1,1} \hat{\delta}_{1} } \pi_{1} (\mu_{t,1,1} - \hat{\mu}_{t,1,1})  \right]  +  \mathbb{E} \left[ \frac{  \delta_{1} - \delta_{1} }{\hat{\delta}_{1} }  \hat{\mu}_{t,1,1}(\pi_{1} - \hat{\pi}_{1})   \right]
		\\ && \lesssim \parallel \omega_{1,1}\delta_{1} -  \hat{\omega}_{1,1} \hat{\delta}_{1}  \parallel \cdot \parallel \mu_{t,1,1} - \hat{\mu}_{t,1,1} \parallel + \parallel \delta_{1} -   \hat{\delta}_{1}  \parallel \cdot \parallel \pi_{1} - \hat{\pi}_{1} \parallel .
	\end{eqnarray*}	
The last inequality holds when $\mathbb{P}( \epsilon < \hat{\omega}_{1,1} \hat{\delta}_{1} < \infty ) = 1$, $\mathbb{P}(\epsilon < \hat{\delta}_{1} < \infty ) = 1$, $\mathbb{P}(\epsilon < \pi_{1} < \infty ) = 1$, and $\mathbb{P}( \epsilon < \hat{\mu}_{t,1,1} < \infty)= 1$ for some $\epsilon > 0$. Note that all of these nuisance functions have a maximum of one, so roughly speaking these conditions only restrict too small values of each nuisance function. 
We have a similar result for $\mathbb{P}( \psi^{\text{b}, Z=z, A = a}_{1}(t; \hat{\Theta}) + \phi^{\text{b}, Z=z, A = a}_{1}(t; \hat{\Theta}) ) - \mathbb{P} ( \psi^{\text{b}, Z=z, A = a}_{1}(t; \Theta) )$ for any $z, a \in \{0,1\}$. 

Moreover, we can show a similar result for $\mathbb{P}( \psi^{\text{b}, Z=1}_{2}( \hat{\Theta}) +  \phi^{\text{b}, Z=1}_{2}( \hat{\Theta})) - \mathbb{P} ( \psi^{\text{b}, Z=1}_{2}(\Theta) )$ using double conditional expectations as follows:
\begin{eqnarray*} 
		&& \mathbb{P}( \psi^{\text{b}, Z=1}_{2}(\hat{\Theta})  + \phi^{\text{b}, Z = 1}_{2}(\hat{\Theta})) - \mathbb{P} ( \psi^{\text{b}, Z=1}_{2}(\Theta) ) \\ 
		&& =  \mathbb{E} \left[  \hat{\pi}_{1} + \frac{\mathbb{I}(Z = 1)}{\hat{\delta}_{1}} \left\{ A - \hat{\pi}_{1} \right\} - \pi_{1} \right]  \\ 
		&& = \mathbb{E} \left[  \frac{\mathbb{I}(Z = 1) - \hat{\delta}_{1} }{\hat{\delta}_{1}}  \left( \pi_{1}  -  \hat{\pi}_{1}  \right)  \right] \\ 
		&& \lesssim  \parallel \delta_{1} -  \hat{\delta}_{1} \parallel  \cdot  \parallel \pi_{1}  -  \hat{\pi}_{1}  \parallel.
\end{eqnarray*}

Combining the results from the nominator and denominator together, we can finally show that Equation~\ref{eq:substract} exhibits the following double robustness structure:
\begin{eqnarray*}
 && O_{\mathbb{P}} \left\{ \sum\limits_{z,a \in \{ 0, 1\}}	\parallel \omega_{z,a} \delta - \hat{\omega}_{z,a} \hat{\delta}  \parallel \parallel  \mu_{t, z,a}  - \hat{\mu}_{t, z, a} \parallel + \parallel  \delta - \hat{\delta}  \parallel \parallel  \pi_{z} - \hat{\pi}_{z} \parallel \right\}
 \\ && \quad +   \xi^{-1}_{n}(\mathbb{P}_{n} - \mathbb{P})(\phi^{\text{b}}_{1}(t ; \Theta) - \psi^{\text{b}}(t; \Theta) \phi^{\text{b}}_{2}(\Theta)) + o_{\mathbb{P}}(n^{-1/2})   
\end{eqnarray*}	
That is, if an outcome model is misspecified, a censoring model and an intervention density should be correctly specified; if a treatment model is misspecified, an intervention density should be correctly specified. 
\end{proof}

\begin{proof}[Proof of Theorem~3]
	The structure of this proof is similar to that of Theorem~2 except that we now have a different estimator for the failure and censoring outcomes. Since the denominator of the estimator~(3) and~(5) are the same, we only need to focus on the two different numerators.
	Let $\psi^{\text{b}, Z =z, A = a}_{1, \text{IF-hazard}}(t; \Theta) = \mathbb{E}(Y_{t} | Z = z, A = a)$ under the censoring assumption (A4*), and denote the influence function of $\psi^{\text{b}, Z = z, A = a}_{1, \text{IF-hazard}} (t; \Theta)$ as $\phi^{\text{b}, Z = z, A = a}_{1, \text{IF-hazard}} (t; \Theta)$.
	Then the next equations examine the convergence of $\mathbb{P}\left(\psi^{\text{b}, Z = z, A = a}_{1, \text{IF-hazard}} (t; \hat{\Theta}) +\phi_{1, \text{IF-hazard}}^{\text{b}, Z = z, A =a}(t; \hat{\Theta}) \right)$ to $\mathbb{P}\left({\psi}^{\text{b}, Z = z, A = a}_{1, \text{IF-hazard}} (t; \Theta) \right)$ for $z, a \in \{0, 1\}$.
	By iterative double expectations under (A4*), we have the following factorization of the nuisance functions:
	\begin{eqnarray*} 
		\allowdisplaybreaks		
		&& \mathbb{P}( \psi^{\text{b}, Z=z, A = a}_{1, \text{IF-hazard}}(t; \hat{\Theta}) + \phi^{\text{b}, Z=z, A = a}_{1, \text{IF-hazard}}(t; \hat{\Theta})) - \mathbb{P} ( \psi^{\text{b}, Z=z, A = a}_{1, \text{IF-hazard}}(t; \Theta) )  \\ &=& 
		\mathbb{E} \left[  - \frac{\mathbb{I}(Z =  z)}{ \hat{\delta}_{z}(\bm{X}_{0})} \left[ \sum\limits_{k=1}^{t} \frac{\mathbb{I}(A  = a)  \mathbb{I}(Y_{k-1} = 1)}{   \hat{G}_{k-1, z, a}(\bm{X}_{0})}   \frac{\hat{S}_{t, z, a}(\bm{X}_{0})}{ \hat{S}_{k, z, a}(\bm{X}_{0})}  \{ R \mathbb{I}(Y_{k} = 0) - \hat{h}_{k}(\bm{X}_{0}, z, a)  \}   \right] \right.   \\
		& + &  \mathbb{E} \left[ \frac{\mathbb{I}(Z =  z)}{ \hat{\delta}_{z}(\bm{X}_{0})}  \hat{S}_{t, z, a}(\bm{X}_{0}) (\mathbb{I}(A = a) - \hat{\pi}_{z}(a ; \bm{X}_{0}))  + \hat{S}_{t, z, a}(\bm{X}_{0}) \hat{\pi}(a ; \bm{X}_{0})  - 
		S_{t, z, a}(\bm{X}_{0}) \pi(a ; \bm{X}_{0}) \right] \\
		&=& 
		\mathbb{E} \left[  - \frac{\mathbb{I}(Z =  z)}{ \hat{\delta}_{z}(\bm{X}_{0})} \left[ \sum\limits_{k=1}^{t} \frac{\mathbb{I}(A  = a)  \mathbb{I}( C > k-1)}{   \hat{G}_{k-1, z, a}(\bm{X}_{0})}   \frac{\hat{S}_{t, z, a}(\bm{X}_{0})}{ \hat{S}_{k, z, a}(\bm{X}_{0})}  \{ \mathbb{I}(T = k)  - \mathbb{I}(T > k-1)\hat{h}_{k}(\bm{X}_{0}, z, a)  \}   \right] \right.   \\
		& + &  \mathbb{E} \left[ \frac{\mathbb{I}(Z =  z)}{ \hat{\delta}_{z}(\bm{X}_{0})}  \hat{S}_{t, z, a}(\bm{X}_{0}) (\mathbb{I}(A = a) - \hat{\pi}_{z}(a ; \bm{X}_{0}))  + \hat{S}_{t, z, a}(\bm{X}_{0}) \hat{\pi}(a ; \bm{X}_{0})  - 
		S_{t, z, a}(\bm{X}_{0}) \pi(a ; \bm{X}_{0}) \right] \\
			&=& 
		\mathbb{E} \left[  - \frac{\mathbb{I}(Z =  z)}{ \hat{\delta}_{z}(\bm{X}_{0})} \left[ \sum\limits_{k=1}^{t} \frac{\mathbb{I}(A  = a)  \mathbb{I}( C > k-1)}{   \hat{G}_{k-1, z, a}(\bm{X}_{0})}   \frac{\hat{S}_{t, z, a}(\bm{X}_{0})}{ \hat{S}_{k, z, a}(\bm{X}_{0})}  S_{k-1,z,a}(\bm{X}_{0})\{ h_{k}(\bm{X}_{0}, z, a)  -  \hat{h}_{k}(\bm{X}_{0}, z, a)  \}    \right] \right.   \\
		& + &  \mathbb{E} \left[ \frac{\mathbb{I}(Z =  z)}{ \hat{\delta}_{z}(\bm{X}_{0})}  \hat{S}_{t, z, a}(\bm{X}_{0}) (\mathbb{I}(A = a) - \hat{\pi}_{z}(a ; \bm{X}_{0}))  + \hat{S}_{t, z, a}(\bm{X}_{0}) \hat{\pi}(a ; \bm{X}_{0})  - 
		S_{t, z, a}(\bm{X}_{0}) \pi(a ; \bm{X}_{0}) \right] \\
			&=& 
		\mathbb{E} \left[  - \frac{\delta_{z}(\bm{X}_{0})}{ \hat{\delta}_{z}(\bm{X}_{0})} \left[ \sum\limits_{k=1}^{t} \frac{ \pi_{z}(a; \bm{X}_{0})  G_{k-1, z, a}(\bm{X}_{0})}{   \hat{G}_{k-1, z, a}(\bm{X}_{0})}   \frac{\hat{S}_{t, z, a}(\bm{X}_{0})}{ \hat{S}_{k, z, a}(\bm{X}_{0})}  S_{k-1,z,a}(\bm{X}_{0})\{ h_{k}(\bm{X}_{0}, z, a)  -  \hat{h}_{k}(\bm{X}_{0}, z, a)  \}   \right] \right.   \\
		& + &  \mathbb{E} \left[ \frac{\delta_{z}(\bm{X}_{0})}{ \hat{\delta}_{z}(\bm{X}_{0})}  \hat{S}_{t, z, a}(\bm{X}_{0}) (\pi_{z}(a; \bm{X}_{0}) - \hat{\pi}_{z}(a ; \bm{X}_{0}))  + \hat{S}_{t, z, a}(\bm{X}_{0}) \hat{\pi}(a ; \bm{X}_{0})  - 
		S_{t, z, a}(\bm{X}_{0}) \pi(a ; \bm{X}_{0}) \right] \\
		&& \mbox{(Hereafter, Let us omit the conditioning $\bm{X}_{0}$ in notations for simplicity)} \\ 
			&=& 
		\mathbb{E} \left[  - \pi_{z}(a)   \left[ \sum\limits_{k=1}^{t} \frac{\delta_{z} G_{k-1, z, a}}{ \hat{\delta}_{z}\hat{G}_{k-1, z, a}}   \frac{\hat{S}_{t, z, a}}{ \hat{S}_{k, z, a}}   S_{k-1,z,a}  \{ h_{k,z,a} - \hat{h}_{k, z, a}  \}   \right] \right.   \\
		& + &  \mathbb{E} \left[ \frac{\delta_{z} - \hat{\delta}_{z}}{ \hat{\delta}_{z}}  \hat{S}_{t, z, a} (\pi_{z}(a) - \hat{\pi}_{z}(a)) +  \pi_{z}(a) (\hat{S}_{t,z,a}- S_{t,z,a} ) \right].
	\end{eqnarray*}	

Following the proof in Appendix A from~\cite{moore2009increasing}, we may rewrite 
\begin{eqnarray}
	\label{eq:moore}
&&  -  \pi_{z}(a) \left[ \sum\limits_{k=1}^{t} \frac{\delta_{z} G_{k-1, z, a}}{ \hat{\delta}_{z}\hat{G}_{k-1, z, a}}   \frac{\hat{S}_{t, z, a}}{ \hat{S}_{k, z, a}}   S_{k-1,z,a}  \{ h_{k,z,a} - \hat{h}_{k, z, a}  \}    - (\hat{S}_{t,z,a}- S_{t,z,a} ) \right] \\
& =& \pi_{z}(a)  \hat{S}_{t,z,a}  \sum\limits_{k=1}^{t} \frac{\delta_{z} G_{k-1, z, a}}{ \hat{\delta}_{z}\hat{G}_{k-1, z, a}}\left\{ - \frac{\hat{S}_{k, z, a} - S_{k, z, a}}{\hat{S}_{k, z, a}} + \frac{\hat{S}_{k-1, z, a} - S_{k-1, z, a}}{\hat{S}_{k-1, z, a}}  \right\} + \pi_{z}(a)(\hat{S}_{t,z,a}- S_{t,z,a} ).	\nonumber 
\end{eqnarray}
Then applying the Cauchy-Schwarz inequality and the triangle inequality, we have the following inequality.
\begin{eqnarray*}
&&\eqref{eq:moore}= \pi_{z}(a)   \hat{S}_{t,z,a} \sum\limits_{k=1}^{t}  \frac{\delta_{z} G_{k-1, z, a} - \hat{\delta}_{z}\hat{G}_{k-1, z, a}}{ \hat{\delta}_{z}\hat{G}_{k-1, z, a}}  \left\{  \frac{\hat{S}_{k-1, z, a} - S_{k-1, z, a}}{\hat{S}_{k-1, z, a}} - \frac{\hat{S}_{k, z, a} - S_{k, z, a}}{\hat{S}_{k, z, a}} \right\}    \\
& \leq & \pi_{z}(a) \hat{S}_{t,z,a} \left\{ \sum\limits_{k=1}^{t} \left(  \frac{\delta_{z} G_{k-1, z, a} - \hat{\delta}_{z}\hat{G}_{k-1, z, a}}{ \hat{\delta}_{z}\hat{G}_{k-1, z, a}}  \right)^2 \right\}^{1/2} \left\{ \sum\limits_{k=1}^{t}  \left(  \frac{\hat{S}_{k-1, z, a} - S_{k-1, z, a}}{\hat{S}_{k-1, z, a}} - \frac{\hat{S}_{k, z, a} - S_{k, z, a}}{\hat{S}_{k, z, a}} \right)^{2} \right\}^{1/2} \\
&\leq& \pi_{z}(a)\left\{ \sum\limits_{k=1}^{t} \left(  \frac{\delta_{z} G_{k-1, z, a} - \hat{\delta}_{z}\hat{G}_{k-1, z, a}}{ \hat{\delta}_{z}\hat{G}_{k-1, z, a}}  \right)^2 \right\}^{1/2} \left\{ \sum\limits_{k=1}^{t}  \left(  \frac{\hat{S}_{k-1, z, a} - S_{k-1, z, a}}{\hat{S}_{k-1, z, a}}  \right)^{2} - \left( \frac{\hat{S}_{k, z, a} - S_{k, z, a}}{\hat{S}_{k, z, a}} \right)^{2} \right\}^{1/2} \\ &=& \pi_{z}(a) \hat{S}_{t,z,a} \left\{ \sum\limits_{k=1}^{t} \left(  \frac{\delta_{z} G_{k-1, z, a} - \hat{\delta}_{z}\hat{G}_{k-1, z, a}}{ \hat{\delta}_{z}\hat{G}_{k-1, z, a}}  \right)^2 \right\}^{1/2} \left\lVert  \hat{S}_{t, z, a} - S_{t, z, a}\right\rVert.
\end{eqnarray*}	

Therefore,
\begin{eqnarray*}
		&& \mathbb{P}( \psi^{\text{b}, Z=z, A = a}_{1, \text{IF-hazard}}(t; \hat{\Theta}) + \phi^{\text{b}, Z=z, A = a}_{1, \text{IF-hazard}}(t; \hat{\Theta})) - \mathbb{P} ( \psi^{\text{b}, Z=z, A = a}_{1, \text{IF-hazard}}(t; \Theta) )  \\ &\leq& 
		 \pi_{z}(a)  \left\lVert S_{t,z,a}- \hat{S}_{t,z,a} \right\rVert \left\lVert \sum\limits_{k=1}^{t} \frac{\delta_{z} G_{k-1, z, a} - \hat{\delta}_{z}\hat{G}_{k-1, z, a}}{ \hat{\delta}_{z}\hat{G}_{k-1, z, a}} \right\rVert+ \frac{\delta_{z} - \hat{\delta}_{z}}{\hat{\delta}_{z}}\hat{S}_{t,z,a}(\pi_{z}(a) - \hat{\pi}_{z}(a)) \\ 
		&= & O_{\mathbb{P}} \left\{ \parallel  S_{t,z,a}- \hat{S}_{t,z,a} \parallel \parallel \sum\limits_{k=1}^{t} (\delta_{z} G_{k-1, z, a} - \hat{\delta}_{z}\hat{G}_{k-1, z, a})   \parallel + \parallel  \delta_{z} - \hat{\delta}_{z}  \parallel  \parallel \pi_{z}(a) - \hat{\pi}_{z}(a) \parallel\right\},
\end{eqnarray*}	
when $\mathbb{P} ( \epsilon < \pi_{z}(a) < \infty) = 1$, $\mathbb{P} ( \epsilon < \hat{\delta}_{z} \hat{G}_{k-1, z, a} < \infty) = 1$, 
$\mathbb{P} ( \epsilon < \hat{S}_{t,z, a} < \infty) = 1$, and $\mathbb{P} ( \epsilon < \hat{\delta}_{z}< \infty) = 1$.

We can apply the above result to \emph{any} $z, a \in \{0,1\}$ of $\mathbb{P}( \psi^{\text{b}, Z=z, A = a}_{1, \text{IF-hazard}}(t; \hat{\Theta}) + \phi^{\text{b}, Z=z, A = a}_{1, \text{IF-hazard}}(t; \hat{\Theta})) - \mathbb{P} ( \psi^{\text{b}, Z=z, A = a}_{1, \text{IF-hazard}}(t; \Theta) )$.
If we use the convergence result for the denominator $\psi^{\text{b}}_{2}(\Theta)$ from Theorem~2, we can conclude the following double robustness of $\hat{\psi}^{\text{b}}_{\text{IF-hazard}}(t)$:
	\begin{eqnarray*}
		&& O_{\mathbb{P}} \left\{ \sum\limits_{z,a \in \{ 0, 1\}}	\parallel  S_{t,z,a}- \hat{S}_{t,z,a} \parallel \parallel \sum\limits_{k=1}^{t} (\delta_{z} G_{k-1, z, a} - \hat{\delta}_{z}\hat{G}_{k-1, z, a}) \parallel   \parallel + \parallel  \delta_{z} - \hat{\delta}_{z}  \parallel  \parallel \pi_{z}(a) - \hat{\pi}_{z}(a) \parallel \right\}
		\\ && \quad  + \xi^{-1}_{n}(\mathbb{P}_{n} - \mathbb{P}) \left\{ \phi^{\text{b}}_{1,\text{IF-hazard}}(t ; \Theta)-\psi^{\text{b}}_{\text{IF-hazard}}(t ; \Theta) \phi^{\text{b}}_{2}(\Theta) \right\} +  o_{\mathbb{P}}(n^{-1/2}).
	\end{eqnarray*}	
\end{proof}

\section{Simulation models}
\label{sec:appendix_sim}

In this supplementary section, we describe the detailed data generating models used for simulation.

\subsection{Cox proportional model}
\label{ssec:appendix_cox}

\begin{eqnarray*}
	\bm{X}_{i,0} &\sim& MVN (\bm{0}, \bm{I}_{5 \times 5})\\
	\mbox{logit}( p(Z_{i} = 1 | \bm{X}_{i, 0})) &= &   \bm{X}^{\prime}_{i,0} \boldsymbol{\kappa}   \mbox{ or } Z_{i} \mid \bm{X}_{i,0} \sim N( \bm{X}^{\prime}_{i,0} \boldsymbol{\kappa}  , \bm{I}_{5 \times 5}) \\
	\mbox{logit}\left(p(A_{i} | \bm{X}_{i,0}, Z_{i}, U_{i})  \right) & =  &   -0.1 + \bm{X}^{\prime}_{i,0} \boldsymbol{\alpha}_{x} + Z_{i}  \alpha_{z}  + U_{i} \alpha_{u}     \\ 
	h(t_{i} | \bm{X}_{i,0}, Z_{i}, A_{i}, U_{i}) & = & h_{0}(t_{i})\exp \left(\bm{X}^{\prime}_{i,0}  \boldsymbol{\beta}_{x} + A_{i}  \beta_{a} + U_{i}  \beta_{u} \right) \\  \mbox{logit}\left(p(R_{i} | \bm{X}_{i,0}, Z_{i}, A_{i})\right) &=&   \bm{X}^{\prime}_{i,0} \boldsymbol{\gamma}_{x} + Z_{i} \gamma_{z} + A_{i} \gamma_{a}  
\end{eqnarray*}

We generated instruments with $\boldsymbol{\kappa} = (-0.5, -0.5, 1, -1, -0.7)^{\prime}$ for a binary case (simulation with a continuous IV will be formally introduced in Section~\ref{sec:continuous}).
For a treatment density, we set $\alpha_{z} = 2$; $\boldsymbol{\alpha}_{x} = (0.5, 0.5, -1.0, -1.0, 1.5)^{\prime}$; and $\alpha_{u} = -0.3$.
Baseline hazard function is $h_{0}(t) = 0.0005t + 0.0003 t^{2}$ and we set $\beta_{a} = 1.5$, $\boldsymbol{\beta}_{x} = (0.5, 0.5, -0.5, -0.5, -0.5)^{\prime}$, and $\beta_{u} = 0.5$.
For a censoring indicator $R$,  we set $\gamma_{z} = -0.5$, $\gamma_{a} = +0.5$; and $\boldsymbol{\gamma}_{x} = (0.3, 0.3, -0.3, -0.3, -0.3)^{\prime}$.

\subsection{Additive hazards model}

\begin{eqnarray*}
		\bm{X}_{i,0} &\sim& MVN (\bm{0}, \bm{I}_{5 \times 5})\\
	\mbox{logit}( p(Z_{i} = 1 | \bm{X}_{i,0}))  &= &  \bm{X}^{\prime}_{i,0} \boldsymbol{\kappa}  \mbox{ or } Z_{i} \mid \bm{X}_{i,0} \sim N( \bm{X}^{\prime}_{i,0} \boldsymbol{\kappa}  , 1) \\
	\mbox{logit}\left( p(A_{i} | \bm{X}_{i,0}, Z_{i}, U_{i}) \right)& =  &   0.1 + \bm{X}^{\prime}_{i,0} \boldsymbol{\alpha}_{x} + Z_{i}  \alpha_{z} + U_{i} \alpha_{u}     \\ 
	h(t_{i} | \bm{X}_{i,0}, Z_{i}, A_{i}, U_{i}) & = & h_{0}(t_{i}) + \exp \left( \bm{X}^{\prime}_{i,0}  \boldsymbol{\beta}_{x} + A_{i}  \beta_{a}   + U_{i}  \beta_{u} \right) \\  \mbox{logit}\left(  p(R_{i} | \bm{X}_{i,0}, Z_{i}, A_{i})\right) &=& \bm{X}^{\prime}_{i,0} \boldsymbol{\gamma}_{x} + Z_{i} \gamma_{z} + A_{i} \gamma_{a}   
\end{eqnarray*}

We generated instruments with $\boldsymbol{\kappa} = (-0.5, -0.5, 1, -1, -0.7)^{\prime}$ both for a binary case.
For a treatment density, we set $\alpha_{z} = 2$; $\boldsymbol{\alpha}_{x} = (0.1, 0.1, -0.2, -0.2, 0.3)^{\prime}$; and $\alpha_{u} = 0.3$.
Baseline hazard is $h_{0}(t) = 0.0005t + 0.0003 t^{2}$ while we set $\beta_{a} = 0.03$, $\boldsymbol{\beta}_{x} = (0.01,  0.01, -0.01, -0.01, -0.01)^{\prime}$, and $\beta_{u} = -0.01$.
We set $\gamma_{z} = 0.5$, $\gamma_{a} = -0.5$, and $\boldsymbol{\gamma}_{x} = (-0.3, -0.3, 0.3, 0.3, 0.3)^{\prime}$ for generating a censoring indicator $R$.

\subsection{Simple estimators}

For the first simulation, we consider two simple estimators for comparison:
\begin{eqnarray*}
	&&\hat{\psi}^{\text{b}}_{\text{ipw}}(t)  =  \left[  \mathbb{P}_{n} \left(\frac{Y R  A Z}{\hat{\omega}_{1,1}  \hat{\pi}_{1} \hat{\delta}_{1} } \right) \mathbb{P}_{n}  \left(  \frac{A Z}{\hat{\delta}_{1}} \right) + 
	\mathbb{P}_{n}  \left( \frac{Y R  (1-A) Z}{ \hat{\omega}_{1,0} (1-\hat{\pi}_{1})  \hat{\delta}_{1} }  \right) \mathbb{P}_{n} \left( \frac{(1-A) Z }{ \hat{\delta}_{1}} \right)  \right. 
	\\  
	&- & \left. \left\{     \mathbb{P}_{n}  \left(\frac{Y R  A (1-Z) }{\hat{\omega}_{0,1}  \hat{\pi}_{0} \hat{\delta}_{0} } \right)  \mathbb{P}_{n} \left(  \frac{A (1-Z)}{\hat{\delta}_{0}} \right) + 
	\mathbb{P}_{n} \left( \frac{Y R  (1-A) (1-Z)}{ \hat{\omega}_{0,0} (1-\hat{\pi}_{0})  \hat{\delta}_{0} }  \right)  \mathbb{P}_{n} \left( \frac{(1-A) (1-Z) }{ \hat{\delta}_{0}} \right)  \right\}       \right]
	\\
	& & \times   \mathbb{P}_{n}  \left\{ \frac{AZ}{\hat{\delta}_{1}}  - \frac{A(1-Z)}{\hat{\delta}_{0}} \right\}^{-1}  \\ 
	&& \hat{\psi}^{\text{b}}_{\text{plug-in}}(t)  = \frac{  \mathbb{P}_{n} \left( \hat{\mu}_{t, 1, 1} \hat{\pi}_{1} + \hat{\mu}_{t,1,0} (1 - \hat{\pi}_{1})  \right) -  \mathbb{P}_{n} \left(  \hat{\mu}_{t, 0, 1} \hat{\pi}_{0} + \hat{\mu}_{t,0,0} (1 - \hat{\pi}_{0})  \right)    }{   \mathbb{P}_{n} (\hat{\pi}_{1}) -   \mathbb{P}_{n}  (\hat{\pi}_{0}) }.
\end{eqnarray*}

\subsection{Simulation under different censoring distributions}

Section~5 in the main text illustrates the performance of three, influence function-based estimators under three different censoring distributions and unmeasured confounding assumptions. For this simulation, we generated simulated data from the following models:
\begin{itemize}
	\item (Baseline covariates)  $\bm{X}_{i,0} \overset{i.i.d.}{\sim} \mbox{MVN}( \bm{0}, \bm{I}_{5 \times 5} )$.
	\item (IV assignment)  Binary instrument: $\mbox{logit}\left(p (Z_{i} = 1  | \bm{X}_{i,0}) \right) =   \bm{X}^{\prime}_{i,0}  \boldsymbol{\kappa}  $.
	\item (Treatment assignment) $\mbox{logit}\left( p(A_{i} | \bm{X}_{i,0}, Z_{i}, U_{i})\right)  =    -0.1 + \bm{X}^{\prime}_{i,0} \boldsymbol{\alpha}_{x} +Z_{i}  \alpha_{z} + U_{i} \alpha_{u}    $.
	\item (Time-to-event)  Cox model: $h(t_{i} | \bm{X}_{i,0}, Z_{i}, A_{i}, U_{i})  =  h_{0}(t_{i})\exp \left( \bm{X}^{\prime}_{i,0}  \boldsymbol{\beta}_{x} + A_{i}  \beta_{a}  + U_{i}  \beta_{u}\right)$.
\end{itemize}

We generated the data with $\boldsymbol{\kappa} = (-0.5, -0.5, 1, -1, -0.7)^{\prime}$, $\alpha_{z} = 2$, $\boldsymbol{\alpha}_{x} = (0.5, 0.5, -1.0, -1.0, 1.5)^{\prime}$. To control the unmeasured confounding, we set $\alpha_{u} = 0.3$ under (i) and (ii) and we set $\alpha_{u} = 0.0$ under no unmeasured confounding scenario (iii). 
Baseline hazard is $h_{0}(t) = 0.0005t + 0.0003 t^{2}$, and we have $\beta_{a} = 1.5$, $\boldsymbol{\beta}_{x} = (0.5, 0.5, -0.5, -0.5, -0.5)^{\prime}$, and we set $\beta_{u} = 2.5$ for (i)--(ii) and $\beta_{u} = 0.5$ for (iii).

\subsection{Misspecified Covariates}

We use the same transformation on the four (out of five) baseline covariates of $\bm{X}_{0} = (X_{0,1}, X_{0,2}, X_{0,3}, X_{0,4})$ in~\cite{kang2007demystifying} and keep using $X_{0,5}$ as it is to construct the misspecified covariates set $\bm{W}_{0}  = (W_{0,1} , W_{0, 2}, W_{0, 3}, W_{0,4},W_{0,5} ) \in \mathbb{R}^{5}$. The following equations relate $\bm{X}_{0}$ into $\bm{W}_{0}$.
\begin{eqnarray*}
	W_{0,1} &=& \exp(X_{0, 1}/2) \\ 
	W_{0,2} &=& X_{0,2} / (1 + \exp(X_{0,1}))+ 10\\
	W_{0,3} &=& (X_{0,1} X_{0,3}/25 + 0.6)^{3} \\
	W_{0,4} &=& (X_{0,2} + X_{0,4} + 20 )^2 \\ 
	W_{0,5} &= & X_{0,5}.
\end{eqnarray*}	


\section{Continuous instrument}
\label{sec:continuous}

In this section, we propose a causal estimator for~(2) in the main text with a continuous IV under censoring assumption (A4) and demonstrate its performance through simulation.

\subsection{Nonparametric estimator}

While assumptions (A1)--(A4) are commonly defined for $\psi^{\text{b}}(t)$ and $\psi^{\text{c}, \kappa}(t)$, our proposed estimator for $\psi^{\text{c}, \kappa}(t)$ needs another set of positivity and monotonicity assumption mainly because the subpopulation on which the local average treatment effect is measured is differently defined when $Z$ is continuous.
Positivity assumption for a continuous instrument often requires non-zero probabilities of observing the instrument over \textit{all} possible values of $Z$, which might be too stringent and sometimes infeasible. Instead, we consider a modified causal estimator with a relaxed version of positivity assumption. This is earned by conditioning on the subpopulation where perturbing instrumental values by $\pm \kappa$ could potentially reverse a treatment value for some positive $\kappa > 0$.
When a perturbed instrumental value by $\pm \kappa$ is still within the support of the variable, i.e., $Z \pm \kappa \in [z_{\min}, z_{\max}]$, then the identification requires the following positivity and monotonicity assumption~\citep{mauro2018instrumental, kennedy2019nonparametric}.	
\begin{enumerate}
	\item[ ($\mbox{A5}^{\text{c}}$)] Positivity: $\mathbb{P}\left(  \epsilon < \mathbb{P}( Z \pm \kappa  | \bm{X}_{0} ) /  \mathbb{P}( Z | \bm{X}_{0} ) <  \infty \right) =  1$ and $\mathbb{P}(R=1 | \bm{X}_{0}, Z, A = a) > 0$ a.e. for some $\epsilon > 0$ across all $Z \pm \kappa \in [z_{\min}, z_{\max}]$ and $a \in \{0, 1\}$.
	\item[ ($\mbox{A6}^{\text{c}}$)] Monotonicity: $\mathbb{P}\left( A^{Z+\kappa} > A^{Z-\kappa}  \right)  > 0$
\end{enumerate}
With those assumptions, the following Lemma~\ref{lemma:continuous_psi} represents the causal estimator~(2) in the main text through the identifiable conditional expectations.

\begin{lemma}(Target estimand $\psi^{\text{c}, \kappa}(t)$ through conditional expectations)
	\label{lemma:continuous_psi}
Define $Z_{+\kappa} = Z + \kappa \mathbb{I}( Z + \kappa < z_{\text{max}} )$ and $Z_{-\kappa} =  Z - \kappa \mathbb{I}(Z - \kappa > z_{\text{min}})$.	
Under the identification assumptions (A1)--(A4) and ($\mbox{A5}^{\text{c}}$)--($\mbox{A6}^{\text{c}}$), our target estimand of a local average treatment effect on $Y_{t}$ can be represented through the following conditional expectations:

\begin{eqnarray}
\label{eq:psi_c_lemma}
\psi^{\text{c}, \kappa}(t) = \frac{\mathbb{E}\left( \mathbb{E}\left( Y_{t} | \bm{X}_{0},  Z_{+\kappa} \right) - \mathbb{E}\left( Y_{t} | \bm{X}_{0},  Z_{-\kappa}   \right) \right)}{\mathbb{E}\left(  \mathbb{E}\left(  A | \bm{X}_{0}, Z_{+\kappa}   \right) - \mathbb{E} \left( A | \bm{X}_{0}, Z_{-\kappa}   \right) \right)}.
\end{eqnarray}
Similarly to the binary case, $\mathbb{E}\left( Y_{t} | \bm{X}_{0}, Z_{\pm \kappa} \right) = \sum\limits_{a \in \{0,1 \}} \mathbb{E}\left(  Y_{t} | \bm{X}_{0}, R = 1, Z_{\pm \kappa}, A = a \right) \mathbb{P}(A = a | \bm{X}_{0},  Z_{\pm \kappa})$  by (A4).
\end{lemma}
Based on Lemma~\ref{lemma:continuous_psi}, we derive an efficient influence function for $\psi^{\text{c}, \kappa}(t)$. Different from the binary instrument case, here we should pay more attention on the boundary of continuous $Z$, $z_{\text{min}}$ and $z_{\text{max}}$. Also, we generalize the previous notation of a $\delta$-function for a continuous $Z$: $\delta_{z} =  \delta(Z ; \bm{X}_{0}) := d\mathbb{P}(Z | \bm{X}_{0})$ while other nuisance functions $\mu_{t,z,a}, \pi_{z}$ and $\omega_{z,a}$ are defined in the same way. 
\begin{theorem}[Influence function of $\psi^{\text{c}, \kappa}(t)$]
\label{thm:psi_c}
An efficient influence function of $\psi^{\text{c}, \kappa}(t)$ is given by:
\begin{eqnarray}
\label{eq:psi_c_if}
\IF\left( \psi^{\text{c}, \kappa}(t) \right) = \left\{ M^{\text{c}}(Y_{t}; Z_{+\kappa} ) - M^{\text{c}}(Y_{t}; Z_{-\kappa} ) - \psi^{\text{c}, \kappa}(t) \left(  \Pi^{\text{c}}(A; Z_{+\kappa} ) - \Pi^{\text{c}}(A; Z_{-\kappa}) \right) \right\} / \psi^{\text{c}}_{2}.
\end{eqnarray}	
When $Z \pm \kappa \in [z_{\min}, z_{\max}]$, 
\begin{eqnarray*}
	M^{c}(Y_{t}; Z_{\pm \kappa}, \Theta) &=& \left( \mu_{t, Z \pm \kappa,1} \pi_{Z \pm \kappa} + \mu_{t, Z \pm \kappa, 0} (1-\pi_{Z \pm \kappa}  \right) \\ &  +&  \frac{\delta_{Z \mp \kappa} }{\delta_{Z} } \left[   \frac{R A}{\omega_{Z, 1} } (Y_{t} - \mu_{t, Z, 1}) + \mu_{t,Z, 1}(A - \pi_{Z}) \right]  \\ 
	& + & \frac{\delta_{Z \mp \kappa} }{\delta_{Z} }  \left[   \frac{R (1-A)}{\omega_{Z, 0} } (Y_{t} - \mu_{t, Z, 0}) + \mu_{t,Z, 0}((1-A) - (1-\pi_{Z}) \right] \\
	\Pi^{c}(A; Z_{\pm \kappa}, \Theta) &=&  (A - \pi_{Z} ) \delta_{Z \mp \kappa} / \delta_{Z} + \pi_{Z \pm \kappa}. 	
\end{eqnarray*}	
Based on the influence function as a function of estimable nuisance functions~\eqref{eq:psi_c_if}, we propose an influence function-based estimator for $\psi^{\text{c}, \kappa}(t)$:
\begin{eqnarray}
\label{eq:psi_c_hat}
\hat{\psi}^{\text{c}, \kappa}(t;\hat{\Theta})  = \frac{ \mathbb{P}_{n}\left( M^{\text{c}}(Y_{t}; Z_{+\kappa}, \hat{\Theta}) - M^{\text{c}}(Y_{t}; Z_{-\kappa},\hat{\Theta}) \right) }{ \mathbb{P}_{n} \left( \Pi^{\text{c}}(A; Z_{+ \kappa},\hat{\Theta} ) - \Pi^{\text{c}}(A; Z_{- \kappa},\hat{\Theta} )   \right) }.
\end{eqnarray}	
\end{theorem}		

An efficient influence function when $Z \pm \kappa \notin [z_{\text{min}}, z_{\text{max}}]$ is presented in Section~\ref{ssec:proof}. Similar to the binary case, the estimator~\eqref{eq:psi_c_hat} requires nuisance functions estimates $\hat{\Theta} = \{ \hat{\mu}_{t, z, a}, \hat{\omega}_{z, a}, \hat{\pi}_{z}, \hat{\delta}_{z} \}$ for continuous $z \in \mathbb{R}$
and then evaluates the causal effect through the sample average $\mathbb{P}_{n}$. In the next section, we will elaborate the estimation procedures and the large-sample properties of the proposed estimator~\eqref{eq:psi_c_hat}.

\subsection{Asymptotic properties of an estimator}

An influence function-based estimator~\eqref{eq:psi_c_hat} for a continuous IV also has a desirable properties as we had in the binary case under certain assumptions.
Let us assume the following conditions hold in addition to the identification assumptions (A1)--(A4) and ($\mbox{A5}^{\text{c}}$)--($\mbox{A6}^{\text{c}}$):
\begin{itemize}
	\item[($\mbox{C1}^{\text{c}}$)]  The nuisance functions are in the Donsker class. 
	\item[($\mbox{C2}^{\text{c}}$)] For some constant $\epsilon > 0$, $\mathbb{P}(\epsilon < \delta_{z \pm \kappa}/\delta_{z} < \infty) = 1$ and $\mathbb{P}( \epsilon < \hat{\omega}_{z,a}  < \infty) = 1$ for any $a \in \{0,1\}$ and $z \pm \kappa \in [z_{\min}, z_{\max}]$.
\end{itemize}	

Then we have the following doubly robust properties for $\hat{\psi}^{\text{c}, \kappa}(t)$.
\begin{theorem}[Asymptotic distribution of $\hat{\psi}^{\text{c}, \kappa}(t)$]
	\label{thm:psi_c_asy}	
Under (A1)--(A4), ($\text{A5}^{\text{c}}$)--($\text{A6}^{\text{c}}$), and ($\text{C1}^{\text{c}}$)--($\text{C2}^{\text{c}}$), 
	\begin{eqnarray}
	\label{eq:psi_c_asy}
	&&	\hat{\psi}^{\text{c}, \kappa}(t) - \psi^{\text{c}, \kappa}(t)  \nonumber
	\\ &= &  O_{\mathbb{P}} \left( \sum\limits_{a \in \{0,1\}} \left\lVert    \frac{ \hat{\delta}_{Z - \kappa}  }{\hat{\delta}_{Z} } - \frac{\delta_{Z - \kappa}}{\delta_{Z}} \right\rVert
	\parallel   \omega_{Z, a} - \hat{\omega}_{Z, a}  \parallel \cdot \parallel  \mu_{t, Z, a} - \hat{\mu}_{t, Z, a} \parallel + \left\lVert    \frac{ \hat{\delta}_{Z - \kappa}  }{\hat{\delta}_{Z} } - \frac{\delta_{Z - \kappa} }{\delta_{Z}} \right\rVert \parallel \pi_{Z} - \hat{\pi}_{Z} \parallel \right. \nonumber \\
	&+&  \left. \sum\limits_{a \in \{0,1\}} \left\lVert  \frac{ \hat{\delta}_{Z + \kappa} }{\hat{\delta}_{Z}} - \frac{\delta_{Z + \kappa}}{\delta_{Z}} \right\rVert 
	\parallel   \omega_{Z, a} - \hat{\omega}_{Z, a}  \parallel \cdot \parallel  \mu_{t, Z, a} - \hat{\mu}_{t, Z, a} \parallel   + \left\lVert    \frac{ \hat{\delta}_{Z + \kappa}}{\hat{\delta}_{Z}} - \frac{\delta_{Z + \kappa}}{\delta_{Z}} \right\rVert \parallel \pi_{Z} - \hat{\pi}_{Z} \parallel \right)   \nonumber \\ 
	& + & \xi^{-1}_{n}(\mathbb{P}_{n} - \mathbb{P}) \left\{ \phi^{c, \kappa}_{1}(t; \Theta) - \psi^{c, \kappa}_{1}(t; \Theta) \phi^{c, \kappa}_{2}(\Theta)\right\} + o_{\mathbb{P}}(n^{-1/2}),  
	\end{eqnarray}	
where $\psi^{c, \kappa}_{1}(t; \Theta) = \mathbb{E}\left( \mathbb{E}(Y_{t} | \bm{X}_{0}, Z_{+\kappa}) - \mathbb{E}(Y_{t} | \bm{X}_{0}, Z_{-\kappa})\right)$ and $\psi^{c, \kappa}_{2}(\Theta)  = \mathbb{E}\left( \mathbb{E}(A | \bm{X}_{0}, Z_{+\kappa}) - \mathbb{E}(A | \bm{X}_{0}, Z_{-\kappa})\right)$; and $\phi^{c, \kappa}_{1}(t; \Theta)$ and $\phi^{c, \kappa}_{2}(\Theta)$ are the influence function of $\psi^{c, \kappa}_{1}(t; \Theta)$ and $\psi^{c, \kappa}_{2}(\Theta)$, respectively.
\end{theorem}
In addition to the double robustness properties, Theorem~\ref{thm:psi_c_asy} suggests that as long as $\hat{\delta}_{z}$ is correctly specified, $\hat{\psi}^{\text{c}, \kappa}(t)$ is consistent. We will examine this point through simulation study.

\subsection{Simulation}

In this simulation, we generate the full data of $\{ (\bm{X}_{i,0}, Z_{i}, A_{i}, R_{i}, T_{i}): i=1,2,\ldots, n\}$ similarly to the binary instrument case, but now we have a continuous $Z_{i} | \bm{\tilde{X}}_{i,0} \overset{ind}{\sim} N( \bm{\tilde{X}}^{\prime}_{i,0} \tilde{\boldsymbol{\kappa}})$ instead of binary $Z_{i}$'s here. For computational simplicity, we consider two-dimensional covariates $\bm{\tilde{X}}_{0} = (X_{0,1}, X_{0,2})$ instead of the five-dimensional $\bm{X}_{i,0}$ only when generating the instrumental variable. Set $\boldsymbol{\tilde{\kappa}} = (-1.0, -1.0)^{\prime}$ and $\beta_{u} = -0.02$ in the additive hazards model.
See Section~\ref{sec:appendix_sim} for the details.
We compare three different estimators: (a) a simple inverse-weighted estimator (IPW estimator); (b) a regression-based plug-in estimator; and (c) the proposed estimator~\eqref{eq:psi_c_hat}. Same as the binary case, we will illustrate the performance of these three estimators under four misspecification scenarios but replace (iv) by the case that misspecifies $\hat{\mu}_{t,z,a}$ in addition to $\hat{\pi}_{z}$ and $\hat{\omega}_{z,a}$: in summary, we simulate (i) correctly specified nuisance functions; (ii) incorrectly specified $\hat{\omega}_{z,a}, \hat{\delta}_{z}$; (iii) incorrectly specified $\hat{\pi}_{z}, \hat{\mu}_{t,z,a}$ functions; and (iv) incorrectly specified $\hat{\pi}_{z}, \hat{\omega}_{z,a}, \hat{\mu}_{t,z,a}$.
\begin{figure}[!ht]
	\centering
	\begin{subfigure}[b]{0.4\textwidth}
		\includegraphics[width=\textwidth]{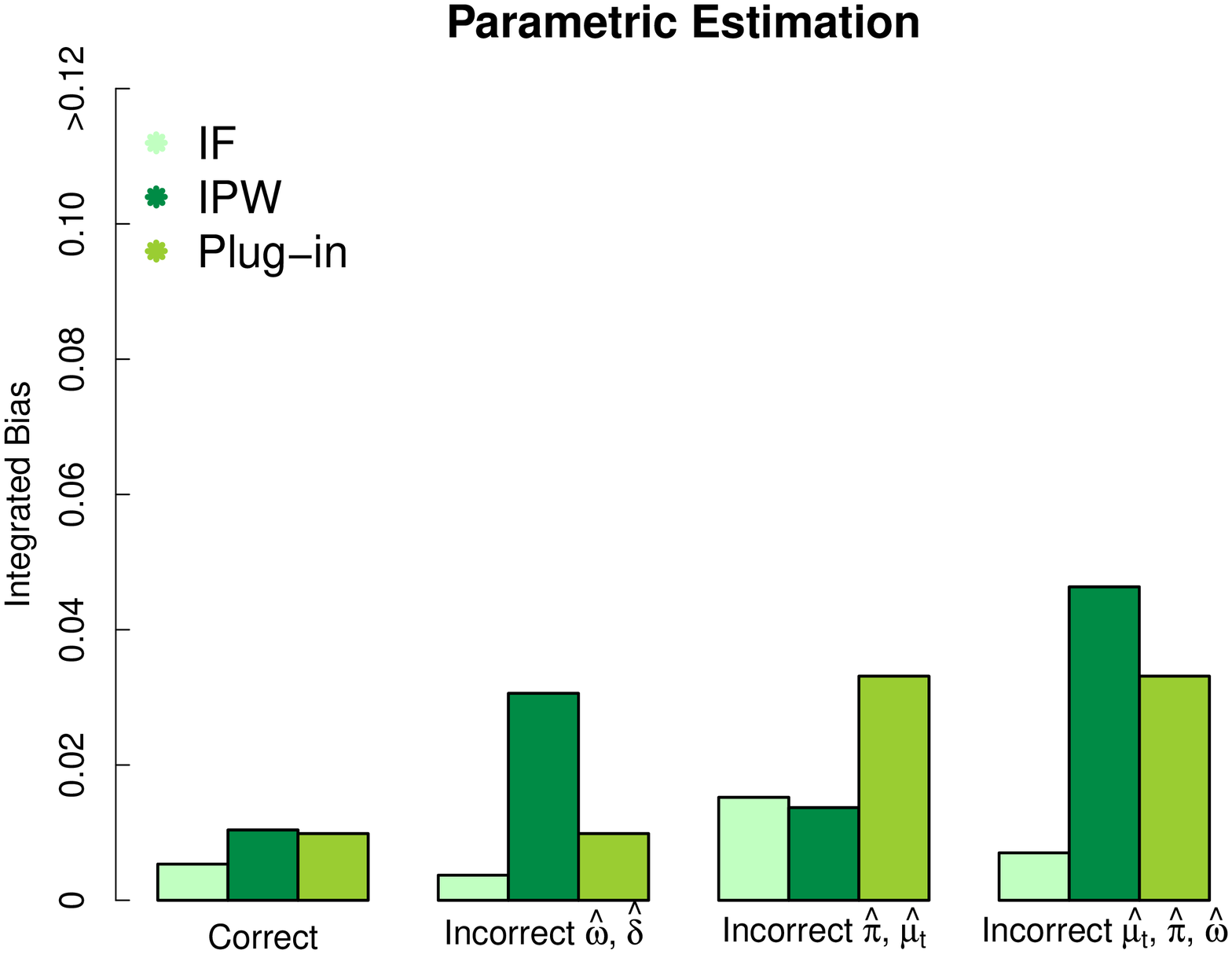}
		\caption{ }
	\end{subfigure}
	\begin{subfigure}[b]{0.4\textwidth}
		\includegraphics[width=\textwidth]{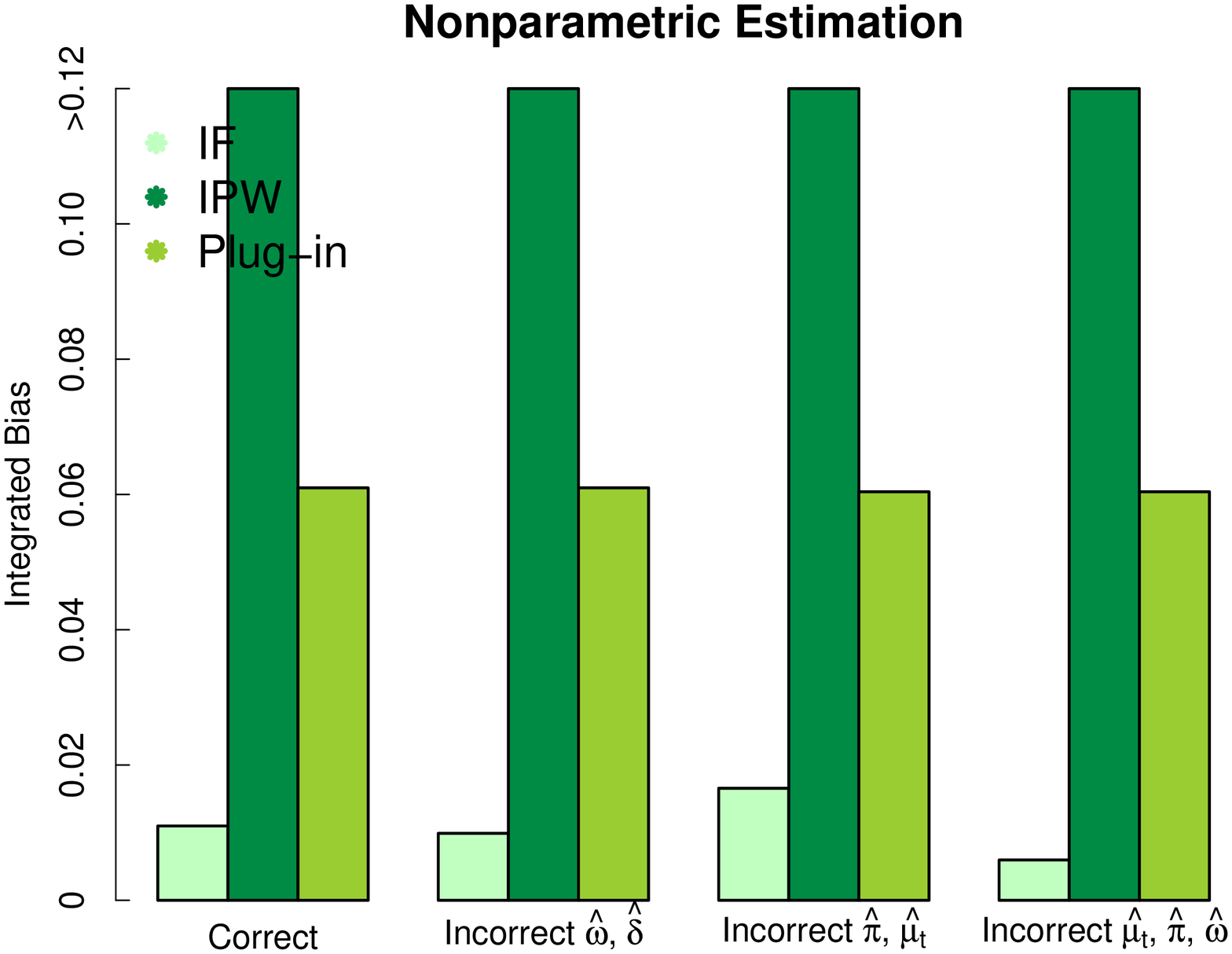}
		\caption{}
	\end{subfigure}
	\begin{subfigure}[b]{0.4\textwidth}
		\includegraphics[width=\textwidth]{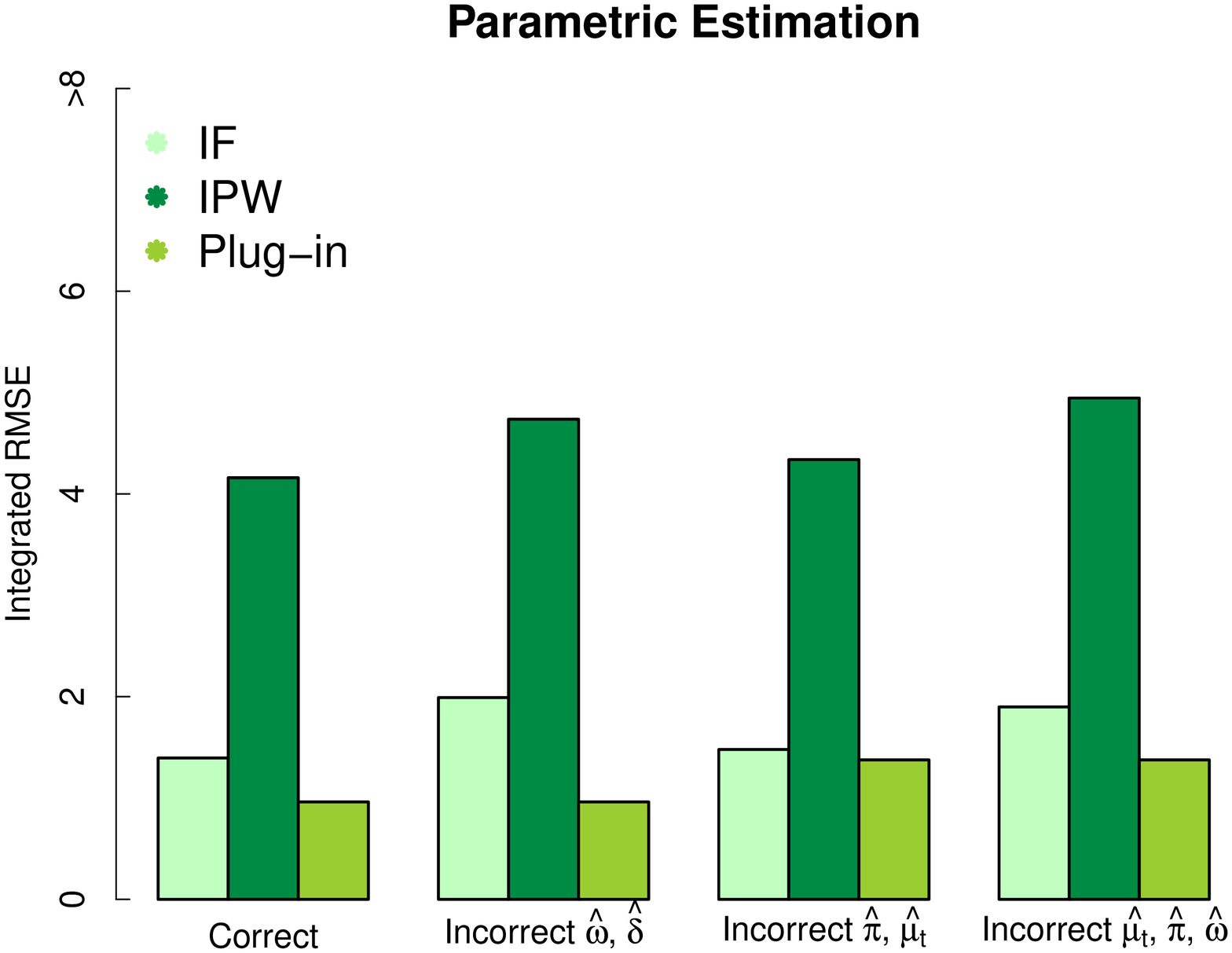}
		\caption{ }
	\end{subfigure}
	\begin{subfigure}[b]{0.4\textwidth}
		\includegraphics[width=\textwidth]{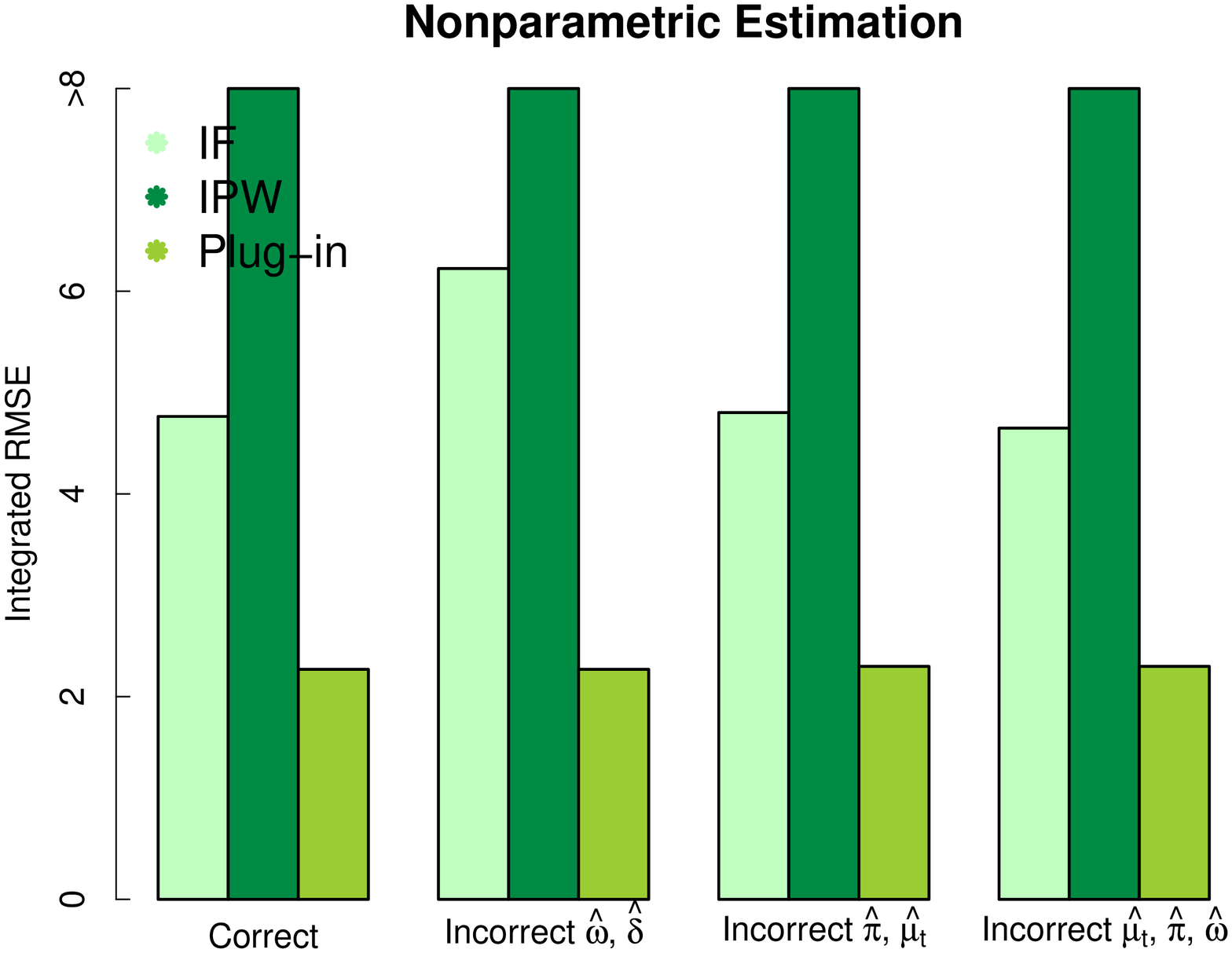}
		\caption{}
	\end{subfigure}
	\caption{\label{fig:additive_continuous} Integrated bias (upper panel) and RMSE (lower panel) of $\hat{\psi}^{\text{c}, \kappa}(t), \hat{\psi}^{\text{c}, \kappa}_{\text{ipw}}(t)$, and $\hat{\psi}^{\text{c}, \kappa}_{\text{plugin}}(t)$ with parametric nuisance functions estimation (left panel) and nonparametric estimation (right panel) when survival outcomes are generated from an additive hazards model and an instrument is continuous. Sample size is $n=1000$ and each scenario was replicated $I=1000$ times.}
\end{figure}

Figure~\ref{fig:additive_continuous} presents the results when the instrument is continuous and survival outcomes are generated from an additive hazards model. The results are generally similar to the binary case except for a few following points: under parametric estimation the proposed estimator (IF) shows relatively inconsistent performance compared to the binary case;  this might be due to the increased variability in estimating a continuous density of $\delta$ function instead of a binary density. When (iii) $\hat{\mu}_{t,z,a}$ and $\hat{\pi}_{a}$ are both incorrectly specified, the performance of the estimator is dependent heavily on the bias in $\hat{\delta}_{z}$, so the estimator shows higher bias than in the correctly specified case. However, the influence function-base estimator is still least sensitive to the model misspecification. 
With nonparametric estimation, the proposed estimator maintains the smallest bias over all four scenarios. The plug-in estimator maintains the smallest RMSE as it does not involve $\delta$-function estimation, which often induces substantial variability.

\subsection{Proof}
\label{ssec:proof}

\begin{proof}[Proof of Theorem~\ref{thm:psi_c}]
	We followed the proof of Theorem 1 of \cite{mauro2018instrumental} and adapted the proof of Theorem~1 with a continuous density for $Z$, $\delta(Z; \bm{X}_{0}) := d\mathbb{P}(Z | \bm{X}_{0})$.
	
	Define a continuous version of $M_{j}$ ($j=0,1$) from Theorem~1 and denote it as $M^{c}$, which is defined across $Z \in [z_{\min}, z_{\max}]$.	
	\begin{eqnarray*}	
		M^{c}(Y_{t}; Z \pm \kappa) &=& \left( \mu_{t, Z \pm \kappa, 1} \pi_{Z \pm \kappa } + \mu_{t, Z \pm \kappa, 1} (1-\pi_{Z \pm \kappa}  \right) \\ &  +&  \frac{\delta(Z \mp \kappa ; \bm{X}_{0})}{\delta(Z; \bm{X}_{0}) } \left[   \frac{R A}{\omega_{Z, 1} } (Y_{t} - \mu_{t, Z, 1}) + \mu_{t,Z, 1}(A - \pi_{Z}) \right]  \\ 
		& + & \frac{\delta(Z \mp \kappa ; \bm{X}_{0})}{\delta(Z; \bm{X}_{0})}  \left[   \frac{R (1-A)}{\omega_{Z, 0} } (Y_{t} - \mu_{t, Z, 0}) + \mu_{t,Z, 0}((1-A) - (1-\pi_{Z}) \right] \\
		\Pi^{c}(A; Z \pm \kappa) &=&  (A - \pi_{Z} ) \delta(Z \mp \kappa; \bm{X}_{0}) / \delta(Z; \bm{X}_{0}) + \pi_{Z \pm \kappa}.
	\end{eqnarray*}
	
Then when $0 < 2 \kappa  < z_{\text{max}} - z_{\text{min}}$,
	
	\begin{eqnarray*}
		&& \IF\left(\psi^{\text{c}}_{1} (t; Z_{+\kappa}) \right) \\ & = &  \left\{  \begin{array}{cc}  M^{\text{c}}(Y_{t}; Z+\kappa) -   \psi^{\text{c}}_{1} (t; Z_{+\kappa}) & Z \in [z_{\text{min}}, z_{\text{max}} - \kappa) \\
			Y_{t} - \psi^{c}_{1}(t; Z_{+\kappa})	& Z \in [z_{\text{max}} - \kappa , z_{\text{max}} ]   \end{array} \right. 
	\end{eqnarray*}

	\begin{eqnarray*}
		&& \IF\left(\psi^{\text{c}}_{1} (t; Z_{-\kappa}) \right) \\ & = &  \left\{  \begin{array}{cc} Y_{t} - \psi^{c}_{1}(t; Z_{-\kappa}) & Z \in [z_{\min} , z_{\text{min}} + \kappa) \\  
			M^{\text{c}}(Y_{t}; Z-\kappa ) -  \psi^{\text{c}}_{1} (t; Z_{-\kappa})  & Z \in [z_{\text{min}} + \kappa, z_{\max}) \end{array} \right. 
	\end{eqnarray*}			
	
	Similarly, 
	\begin{eqnarray*}
		&& \IF\left(\psi^{\text{c}}_{2} (t; Z_{+\kappa}) \right) \\ & = &  \left\{  \begin{array}{cc}  \Pi^{\text{c}}(A; Z+\kappa ) -   \psi^{\text{c}}_{2} (t; Z_{+\kappa}) & Z \in [z_{\text{min}}, z_{\text{max}} - \kappa) \\
			A - \psi^{c}_{2}(t; Z_{+\kappa})	& Z \in [z_{\text{max}} - \kappa , z_{\text{max}} ]   \end{array} \right. 
	\end{eqnarray*}			
	
	\begin{eqnarray*}
		&& \IF\left(\psi^{\text{c}}_{2} (t; Z_{-\kappa}) \right) \\ & = &  \left\{  \begin{array}{cc} A - \psi^{c}_{2}(t; Z_{-\kappa}) & Z \in [z_{\min} , z_{\text{min}} + \kappa) \\  
			\Pi^{\text{c}}(Y_{t}; Z-\kappa ) -  \psi^{\text{c}}_{2} (t; Z_{-\kappa})  & Z \in [z_{\text{min}} + \kappa, z_{\max}) \end{array} \right. 
	\end{eqnarray*}			
	
	Combining the influence functions above together, we have an efficient influence function of $\psi^{\text{c}, \kappa}(t)$ as in Equation~\ref{eq:psi_c_if} 	when $Z \pm \kappa \in [z_{\min}, z_{\max}]$.
	
\end{proof}

\begin{proof}[Proof of Theorem~\ref{thm:psi_c_asy}]
	Let us assume $Z \pm \kappa \in [z_{\min}, z_{\max}]$ for simplicity. 
	For a continuous instrument, we will use Equation~\ref{eq:substract} by replacing $\psi^{\text{b}}_{1}(t)$ and $\psi^{\text{b}}_{2}$ for $\psi^{\text{c}, \kappa}_{1}(t)$ and $\psi^{\text{c}, \kappa}_{2}$; similarly replace $\phi^{\text{b}}_{1}(t)$ and $\phi^{\text{b}}_{2}$ by $\phi^{\text{c}, \kappa}_{1}(t)$ and $\phi^{\text{c}}_{2}$, respectively. 
	\begin{eqnarray*}
		&& \mathbb{P}(\psi^{c, Z_{+\kappa}, A = 1}_{1} (t; \hat{\Theta}) + \phi^{c, Z_{+\kappa}, A = 1}_{1} (t; \hat{\Theta}) ) - \mathbb{P}(\psi^{c, Z_{+\kappa}, A = 1}_{1} (t; \Theta) ) 
		\\ &=& \mathbb{E} \left[  \hat{\mu}_{t}(\bm{X}_{0}, 1, Z + \kappa, A = 1) \hat{\pi}_{Z + \kappa }  
		- \mu_{t}(\bm{X}_{0}, 1, Z + \kappa, A = 1) \pi_{Z + \kappa }  \right]  \\ & + &	  \mathbb{E} \left[    \frac{\hat{\delta}(Z - \kappa ; \bm{X}_{0})}{\hat{\delta}(Z; \bm{X}_{0}) }  \left\{ \frac{R A}{\hat{\omega}_{Z, 1} } (Y_{t} - \hat{\mu}_{t, Z, 1})  + \hat{\mu}_{t,Z, 1}(A - \hat{\pi}_{Z})   \right\}  \right] \\
		&=&  \mathbb{E} \left[  \hat{\mu}_{t, Z + \kappa, 1} \hat{\pi}_{Z + \kappa }  
		- \mu_{t, Z + \kappa, 1} \pi_{Z + \kappa }  \right]  \\ & + &	  \mathbb{E} \left[    \frac{\hat{\delta}(Z - \kappa ; \bm{X}_{0})}{\hat{\delta}(Z; \bm{X}_{0}) }   \left\{  \frac{\omega_{Z, 1} \pi_{Z}}{\hat{\omega}_{Z, 1} } (\mu_{t, Z, 1} - \hat{\mu}_{t, Z, 1})  + \hat{\mu}_{t,Z, 1}(\pi_{Z} - \hat{\pi}_{Z})  \right\} \right] \\
		&=& \mathbb{E}\left[  \frac{\delta(Z; \bm{X}_{0})}{ \delta(Z - \kappa; \bm{X}_{0})} \left\{ \frac{ \hat{\delta}(Z - \kappa; \bm{X}_{0})}{\hat{\delta}(Z ; \bm{X}_{0})} - \frac{\delta(Z - \kappa; \bm{X}_{0})}{\delta(Z ;\bm{X}_{0}) }  \right\}  \frac{ \omega_{Z, 1}  -  \hat{\omega}_{Z,1}   }{  \hat{\omega}_{Z,1}  }  \pi_{Z}  \left\{  \mu_{t, Z, 1}  - \hat{\mu}_{t,Z, 1}     \right\} \right] \\ 
		&+& \mathbb{E}\left[ \frac{\delta(Z; \bm{X}_{0})}{ \delta(Z - \kappa; \bm{X}_{0})} \left\{ \frac{ \hat{\delta}(Z - \kappa; \bm{X}_{0})}{\hat{\delta}(Z ; \bm{X}_{0})} - \frac{\delta(Z - \kappa; \bm{X}_{0})}{\delta(Z ;\bm{X}_{0}) }  \right\} \hat{\mu}_{t,Z,1} \left\{   \pi_{Z} - \hat{\pi}_{Z}  \right\}    \right].
	\end{eqnarray*}		
	The last equation is from the change of variables that was used in \cite{mauro2018instrumental}. 
	Under the assumption that $\mathbb{P}(\epsilon < \delta(Z - \kappa; \bm{X}_{0}) / \delta(Z; \bm{X}_{0})  < \infty  ) = 1$ and $\mathbb{P}( \epsilon <  \hat{\omega}(\bm{X}_{0}, Z, 1)< \infty  ) = 1 $ for some $\epsilon > 0$, we have the following inequality, which holds for any $\{ (Z, A) : Z \in [z_{\min}, z_{\max}],~A \in \{ 0, 1\}   \}$.
	
	\begin{eqnarray*} 
		&&  \mathbb{P}(\psi^{c, Z_{+\kappa}, A = 1}_{1} (t; \hat{\Theta}) + \phi^{c, Z_{+\kappa}, A = 1}_{1} (t; \hat{\Theta})) - \mathbb{P}(\psi^{c, Z_{+\kappa}, A = 1}_{1} (t; \Theta) ) 	\\
		& \lesssim & \left\lVert    \frac{ \hat{\delta}(Z - \kappa; \bm{X}_{0})}{\hat{\delta}(Z;\bm{X}_{0})} - \frac{\delta(Z - \kappa; \bm{X}_{0})}{\delta(Z;\bm{X}_{0})} \right\rVert
		\parallel   \omega_{Z, 1} - \hat{\omega}_{Z, 1} \parallel \cdot \parallel  \mu_{t, Z, 1} - \hat{\mu}_{t, Z, 1} \parallel \\ && \quad + \left\lVert   \frac{ \hat{\delta}(Z - \kappa; \bm{X}_{0})}{\hat{\delta}(Z;\bm{X}_{0})} - \frac{\delta(Z - \kappa; \bm{X}_{0})}{\delta(Z;\bm{X}_{0})} \right\rVert \parallel \pi_{Z} - \hat{\pi}_{Z} \parallel,
	\end{eqnarray*}

	For the denominator of $\psi^{\text{c}, \kappa}(t)$, we can similarly derive:
	\begin{eqnarray*}
		&& \mathbb{P}( \psi^{\text{b}, Z_{+\kappa}}_{2}(\hat{\Theta}) + \phi^{\text{b}, Z_{+\kappa}}_{2}(\hat{\Theta})) - \mathbb{P} ( \psi^{\text{b}, Z_{+\kappa}}_{2}(\Theta) ) \\ 
		&& =  \mathbb{E} \left[  \frac{ \hat{\delta}(Z-\kappa; \bm{X}_{0}) }{ \hat{\delta}(Z;\bm{X}_{0})} ( \pi_{Z} - \hat{\pi}_{Z} ) + \hat{\pi}_{Z+\kappa} - \pi_{Z + \kappa}
		\right] \\ 
		&=& \mathbb{E} \left[   \frac{ \hat{\delta}(Z - \kappa; \bm{X}_{0})/\hat{\delta}(Z ; \bm{X}_{0}) - \delta(Z - \kappa; \bm{X}_{0}) / \delta(Z ;\bm{X}_{0}) }{ \delta(Z - \kappa; \bm{X}_{0}) / \delta(Z; \bm{X}_{0})    } \left\{  \pi_{Z} - \pi_{Z}  \right\} \right] \\
		&& \lesssim  \left\lVert \frac{\hat{\delta}(Z - \kappa; \bm{X}_{0})}{\hat{\delta}(Z; \bm{X}_{0})} - \frac{\delta(Z - \kappa; \bm{X}_{0})}{\delta(Z ;\bm{X}_{0})} \right\rVert  \cdot  \parallel \pi_{Z}  -  \hat{\pi}_{Z}  \parallel.
	\end{eqnarray*}
	
	Hence, combined together we have~\eqref{eq:psi_c_asy}:
	\begin{eqnarray*}
		&& \hat{\psi}^{c, \kappa}(t) -\psi^{c, \kappa}(t) \\ &=& O_{\mathbb{P}} \left(   \sum\limits_{a \in \{0,1\}} \left\lVert   \frac{ \hat{\delta}_{Z - \kappa}  }{\hat{\delta}_{Z} } - \frac{\delta_{Z - \kappa}}{\delta_{Z}} 
		\right\rVert \parallel  \omega_{Z, a} - \hat{\omega}_{Z, a}  \parallel \cdot \parallel  \mu_{t, Z, a} - \hat{\mu}_{t, Z, a} \parallel + \left\lVert   \frac{ \hat{\delta}_{Z - \kappa}  }{\hat{\delta}_{Z} } - \frac{\delta_{Z - \kappa} }{\delta_{Z}} \right\rVert \parallel \pi_{Z} - \hat{\pi}_{Z} \parallel \right. \\
		&+&  \left. \sum\limits_{a \in \{0,1\}} \left\lVert    \frac{ \hat{\delta}_{Z + \kappa} }{\hat{\delta}_{Z}} - \frac{\delta_{Z + \kappa}}{\delta_{Z}} \right\rVert
		\parallel   \omega_{Z, a} - \hat{\omega}_{Z, a}  \parallel \cdot \parallel  \mu_{t, Z, a} - \hat{\mu}_{t, Z, a} \parallel   + \left\lVert    \frac{ \hat{\delta}_{Z + \kappa}}{\hat{\delta}_{Z}} - \frac{\delta_{Z + \kappa}}{\delta_{Z}} \right\rVert \parallel \pi_{Z} - \hat{\pi}_{Z} \parallel \right) \nonumber
		\\&+&(\mathbb{P}_{n} - \mathbb{P}) \left\{ \phi^{c, \kappa}_{1}(t; \Theta) - \psi^{c, \kappa}(t; \Theta)\phi^{c, \kappa}_{2}(\Theta)\right\} + o_{\mathbb{P}}(n^{-1/2})
	\end{eqnarray*}		
\end{proof}

\section{Sample splitting method for nuisance parameter estimation}

\begin{algorithm}[H]
	\caption{Estimating the proposed estimator~(3) in the main text}\label{algo:psib}
	\begin{algorithmic}[1]
		\State Define the number of sample splits $K$. Defaults to $K = 10$. Then for $k \in \{1,2,\ldots, K \}$, 
		let $\mathcal{D} = \{ \mathcal{O}_{i} : S_{i} \neq k  \}$ be the training data and $\mathcal{B} = \{ \mathcal{O}_{i} : S_{i} = k   \}$ be the test data set. 
		\State \textbf{1. Outcome regression}
		\State Implement machine learning methods, e.g., random forest or classification, to learn $\mathbb{E}(Y_{t} | \bm{X}_{0}, Z, A, R = 1)$ only using $\mathcal{D}$ for $t=1,2,\ldots, \tau$.
		\State Obtain predicted $\{ \hat{\mu}_{t, z, a}: z, a \in \{ 0, 1 \}, t= 1,2,\ldots, \tau \}$ for subjects in $\mathcal{B}$.
		\bigskip
		\State \textbf{2. Censoring indicator}
		\State Use machine learning techniques to learn conditional distribution of $\mathbb{E}(R | \bm{X}_{0}, Z, A)$ only using observations in $\mathcal{D}$.
		\State Obtain predicted $\{ \hat{\omega}(\bm{X}_{0}, z, a) : z,a \in \{ 0, 1\} \}$ for subjects in $\mathcal{B}$.
		\bigskip
		\State \textbf{3. Treatment propensity score}
		\State Use machine learning techniques to learn conditional distribution of $\mathbb{E}(A | \bm{X}_{0}, Z)$ using $\mathcal{D}$.
		\State Obtain predicted $\{  \hat{\pi}_{z}  : z \in \{0, 1\}  \}$ for subjects in $\mathcal{B}$.
		\bigskip
		\State \textbf{4. Instrument prevalence}
		\State Use machine learning techniques to learn conditional distribution of $\mathbb{E}(Z | \bm{X}_{0})$ using $\mathcal{D}$.
		\State Obtain predicted $\{  \hat{\delta}_{z}  : z \in \{0, 1\}  \}$ for subjects in $\mathcal{B}$.
		\bigskip
		\State Construct an influence function-based estimator of $\hat{\psi}^{\text{b}(k)}(t)$ using the predicted values \textbf{1}-\textbf{4} of $\mathcal{B}$.
		\State Repeat the procedures \textbf{1}-\textbf{4} for each $k=1,2,\ldots,K$.
		\State \textbf{Output}: $\hat{\psi}^{\text{b, split}}(t) = \frac{1}{K} \sum\limits_{k=1}^{K} \hat{\psi}^{\text{b}(k)}(t)$.
	\end{algorithmic}
\end{algorithm}

\begin{algorithm}[H]
	\caption{Estimating the proposed estimator~\eqref{eq:psi_c_hat}}\label{algo:psic}
	\begin{algorithmic}[1]
	\State Define the number of sample splits $K$. Defaults to $K = 10$. Then for $k \in \{1,2,\ldots, K \}$, 
	let $\mathcal{D} = \{ \mathcal{O}_{i} : S_{i} \neq k  \}$ be the training data and $\mathcal{B} = \{ \mathcal{O}_{i} : S_{i} = k   \}$ be the test data set. 
	\State \textbf{1. Outcome regression}
	\State  Implement machine learning methods, e.g., random forest or classification, to learn $\mathbb{E}(Y_{t} | \bm{X}_{0}, Z, A, R = 1)$ only using $\mathcal{D}$ for $t=1,2,\ldots, \tau$.
	\State Obtain predicted $\{ \hat{\mu}_{t, z, a}: z, a \in \{ 0, 1 \}, t= 1,2,\ldots, \tau \}$ for subjects in $\mathcal{B}$.
	\bigskip
	\State \textbf{2. Censoring indicator}
	\State Use machine learning techniques to learn conditional distribution of $\mathbb{E}(R | \bm{X}_{0}, Z, A)$ only using observations in $\mathcal{D}$.
	\State Obtain predicted $\{ \hat{\omega}(\bm{X}_{0}, z, a) : z,a \in \{ 0, 1\} \}$ for subjects in $\mathcal{B}$.
	\bigskip
	\State \textbf{3. Treatment propensity score}
	\State Use machine learning techniques to learn conditional distribution of $\mathbb{E}(A | \bm{X}_{0}, Z)$ using $\mathcal{D}$.
	\State Obtain predicted $\{  \hat{\pi}_{z}  : z \in \{0, 1\}  \}$ for subjects in $\mathcal{B}$.
	\bigskip
	\State \textbf{4. Instrument prevalence}
	\State Use machine learning techniques to learn the conditional density of $\delta(Z ; \bm{X}_{0})$ using $\mathcal{D}$, and then estimates a kernel density over $z \in [z_{\min}, z_{\max}]$.
	\State Obtain predicted values $\{  \hat{\delta}_{z}   \}$ for subjects in $\mathcal{B}$.
	\bigskip
	\State Construct an influence function-based estimator of $\hat{\psi}^{\text{c}(k)}(t)$ using the predicted values \textbf{1}-\textbf{4} of $\mathcal{B}$.
	\State Repeat the procedure \textbf{1}-\textbf{4} for each $k=1,2,\ldots,K$.
	\State \textbf{Output}: $\hat{\psi}^{\text{c, split}, \kappa}(t) = \frac{1}{K} \sum\limits_{k=1}^{K} \hat{\psi}^{\text{c}(k), \kappa}(t)$.
\end{algorithmic}
\end{algorithm}

When an instrument is continuous, the only difference from the above Algorithm~\ref{algo:psic} is in predicting the conditional density of a continuous instrument. We use the \texttt{R} package of \texttt{ks} to estimate and evaluate the probability density of continuous variables. 
One of the other alternatives to estimate $\hat{\delta}_{z+\kappa}/\hat{\delta}_{z}$ or $\hat{\delta}_{z-\kappa}/\hat{\delta}_{z}$ is to use the strategy proposed by~\cite{diaz2020non}, which, based on the Bayes rule, estimates the density ratio through the conditional probability estimation of a binary indicator.

\section{Additional simulation results}
\label{sec:additional_1}

Figure~\ref{fig:additive_binary} presents the simulation results with three estimators $\hat{\psi}^{\text{b}}(t), \hat{\psi}^{\text{b}}_{\text{ipw}}(t)$, and $\hat{\psi}^{\text{b}}_{\text{plugin}}(t)$ under the censoring assumption (A4) when survival outcomes are generated from an additive hazards model.
\begin{figure}[H]
	\centering
	\begin{subfigure}[b]{0.4\textwidth}
		\includegraphics[width=\textwidth]{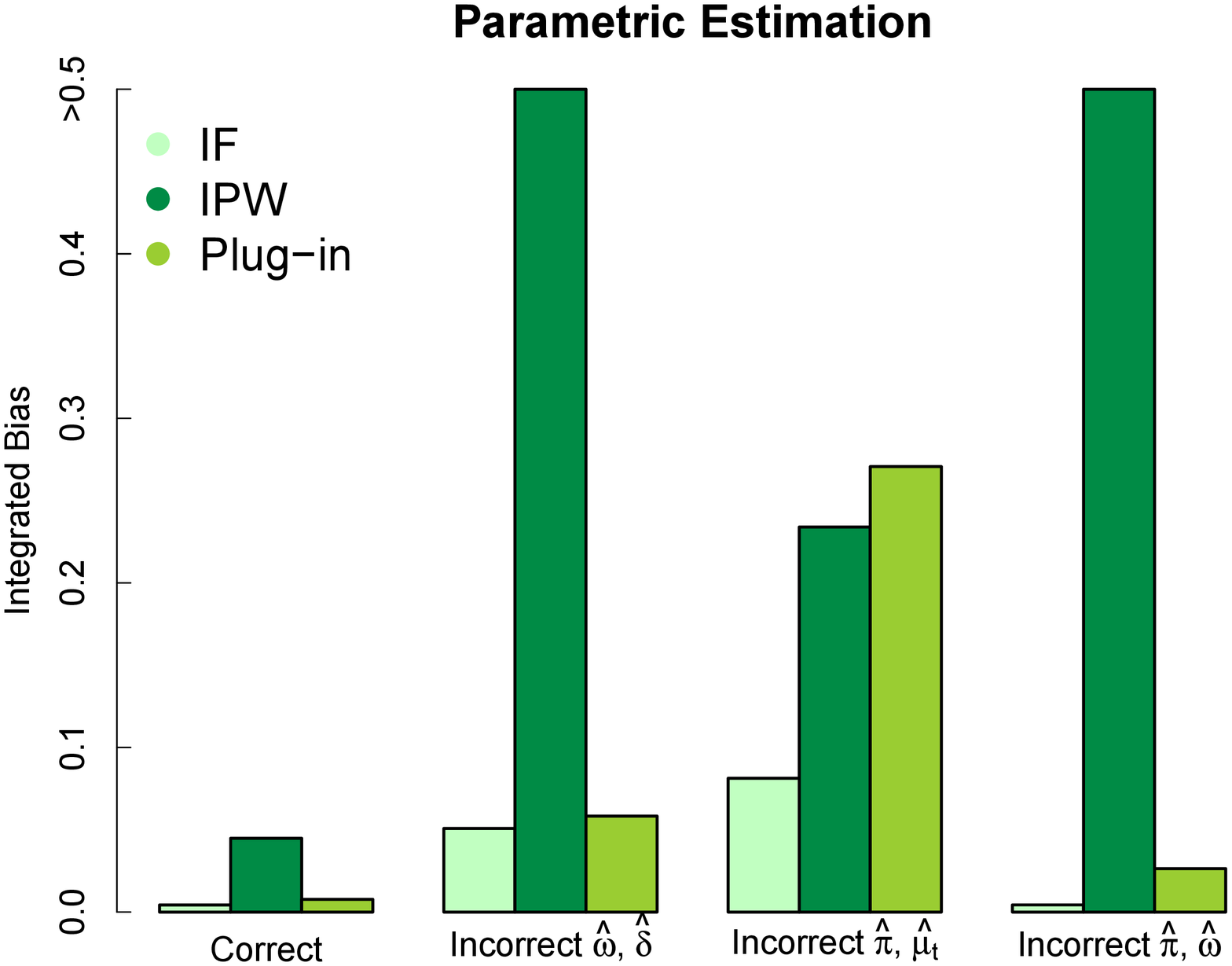}
		\caption{ }
	\end{subfigure}
	\begin{subfigure}[b]{0.4\textwidth}
		\includegraphics[width=\textwidth]{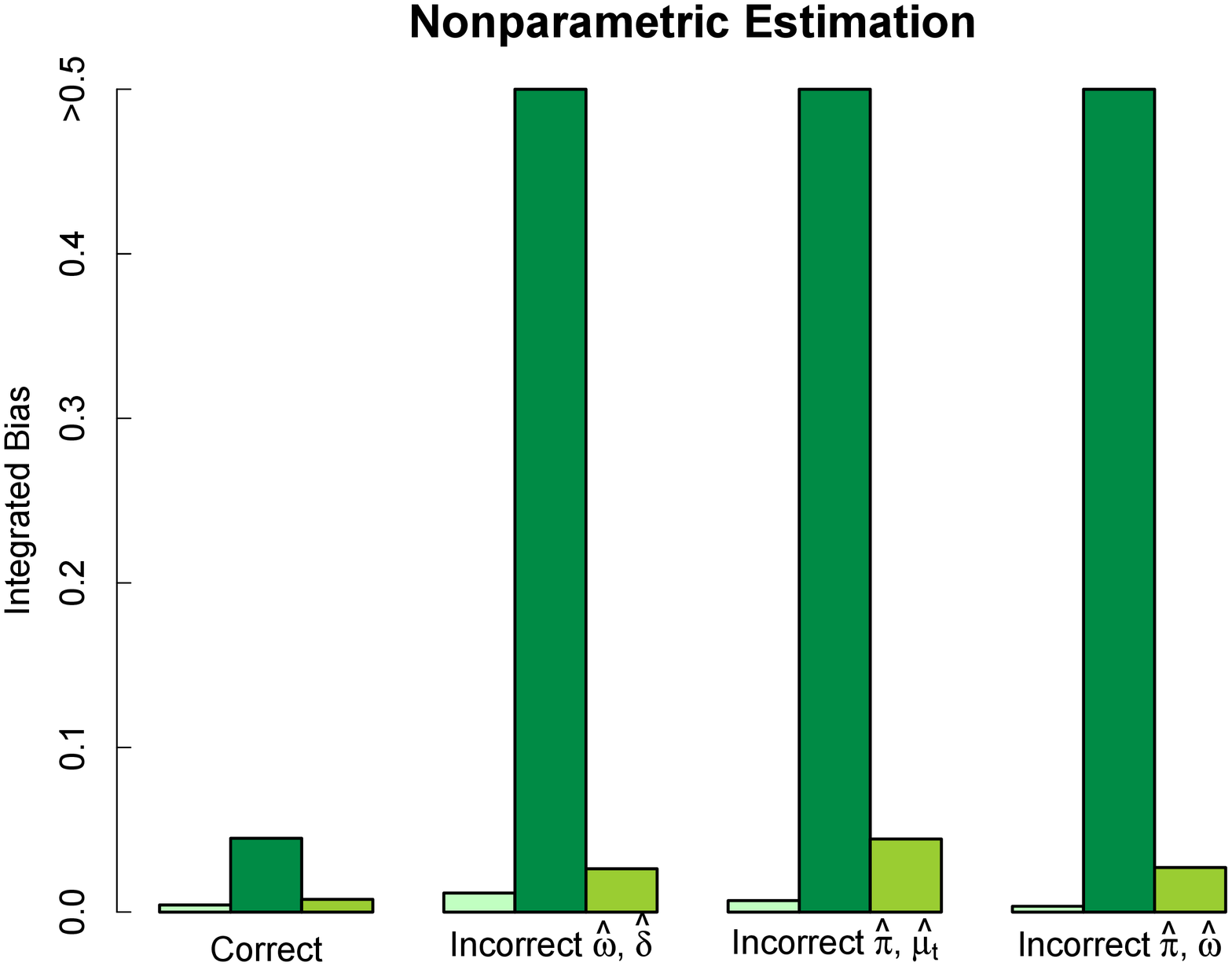}
		\caption{}
	\end{subfigure}
	\begin{subfigure}[b]{0.4\textwidth}
		\includegraphics[width=\textwidth]{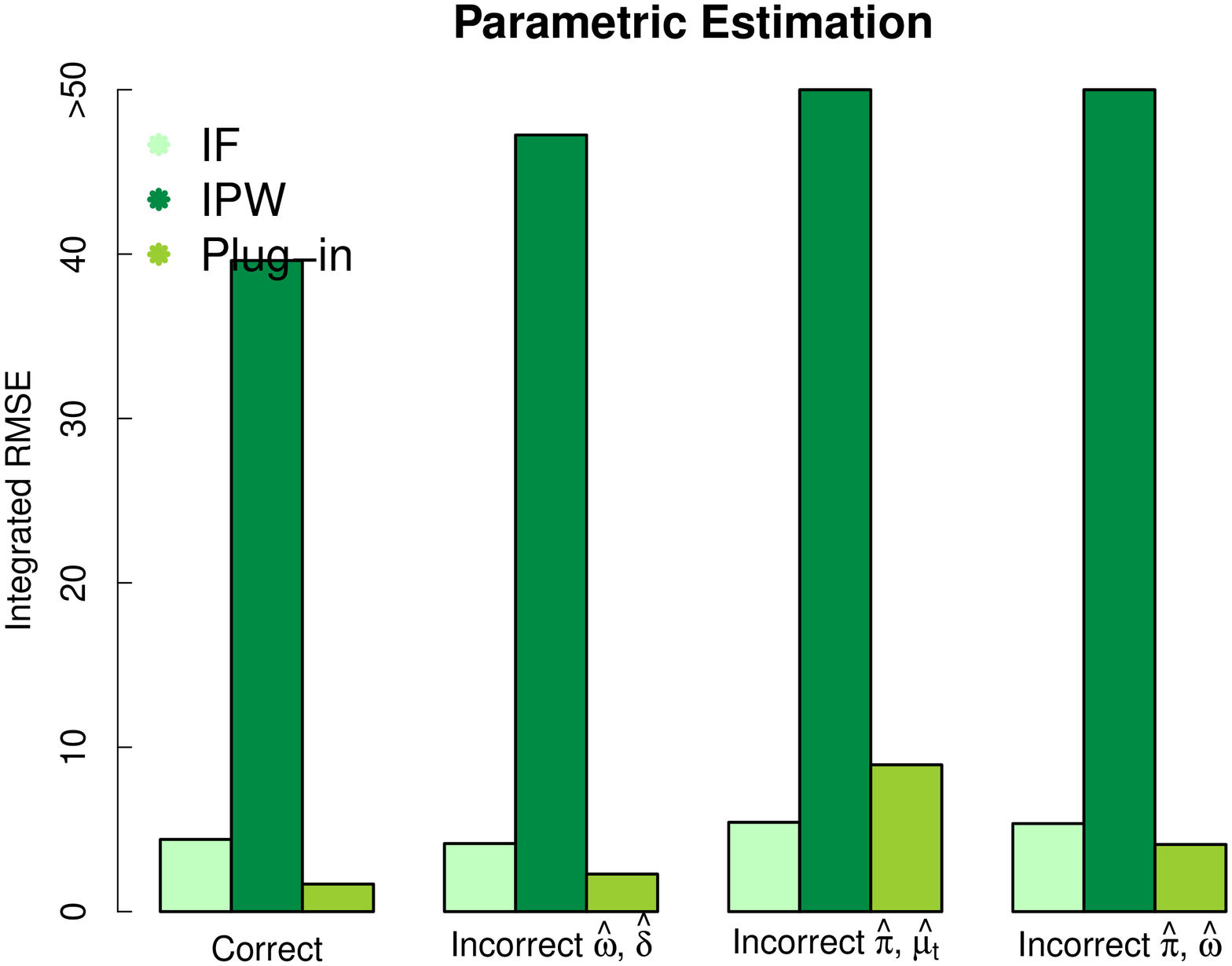}
		\caption{ }
	\end{subfigure}
	\begin{subfigure}[b]{0.4\textwidth}
		\includegraphics[width=\textwidth]{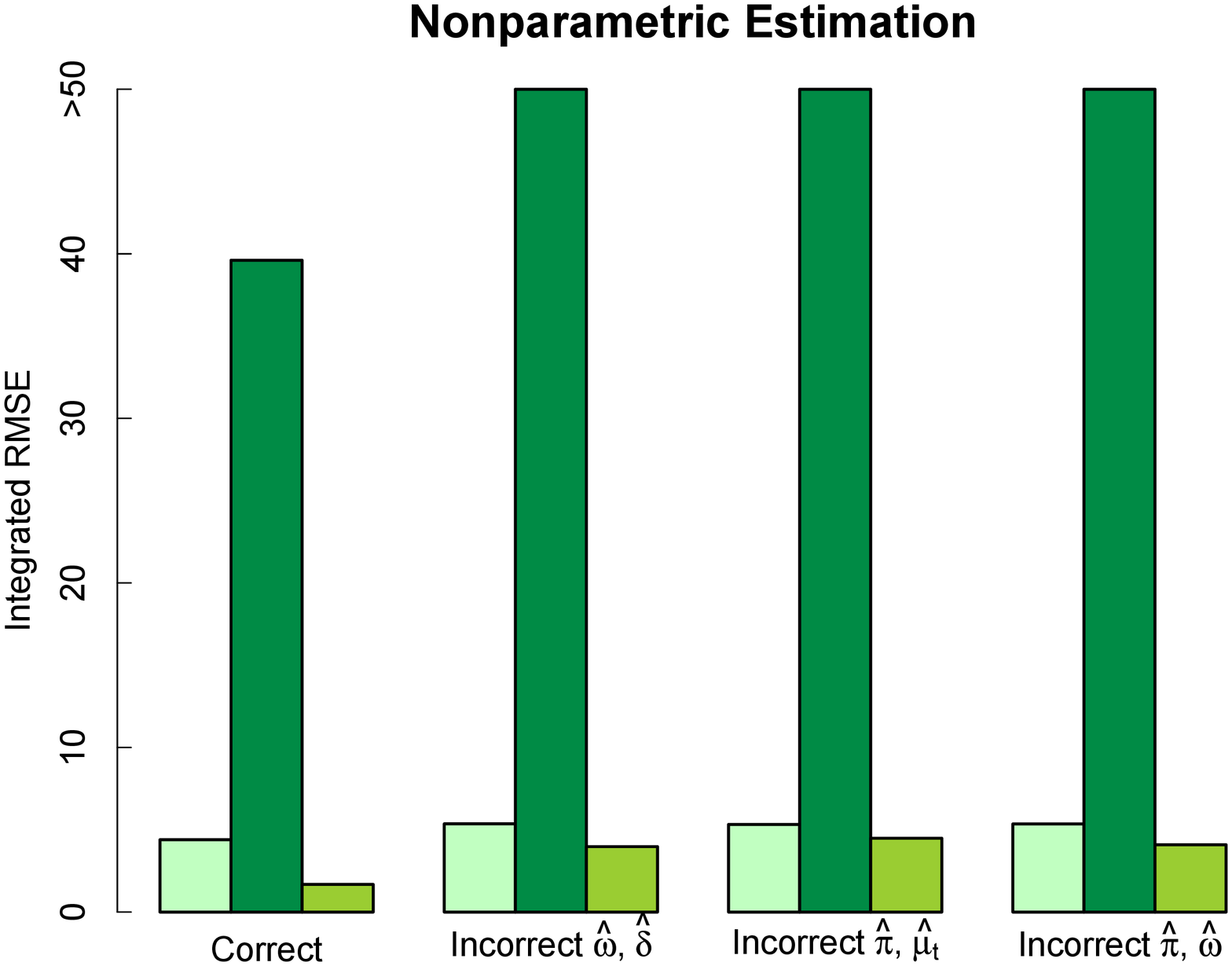}
		\caption{}
	\end{subfigure}
	\caption{\label{fig:additive_binary}Bias (upper panel) and RMSE (lower panel) of $\hat{\psi}^{\text{b}}(t), \hat{\psi}^{\text{b}}_{\text{ipw}}(t)$, and $\hat{\psi}^{\text{b}}_{\text{plugin}}(t)$ with parametric nuisance functions estimation (left panel) and nonparametric estimation (right panel) when survival outcomes are generated from an additive hazards model and an instrument is binary. Sample size is $n=1000$ and each scenario was replicated $I=1000$ times.}
\end{figure}
Figure~\ref{fig:additive_binary} shows that, similar to Figure~1 in the main text, the influence function-based estimator (IF) results in the smallest bias across three model misspecifications; while the plug-in estimator has the smallest RMSE in overall both under parametric and nonparametric estimation. Under nonparametric estimation, both the influence function-based and plug-in estimators are robust against three model specifications while the former estimator still beats the latter slightly.  
The bias of the IPW estimator is over five and its RMSE is exceeding 6000 under nonparametric estimation. 

Table~\ref{tab:rmse} presents the RMSE that is corresponding to the bias results shown in Table~1 from the main text. As expected, the smallest RMSE is observed among the expected winners (colored cells) for each scenario under parametric estimation. However, under nonparametric estimation where the variabilities largely come from nuisance function estimations, the naive estimator $\hat{\psi}^{\text{b}}_{\text{Naive-hazard}}(t)$, which requires the smaller number of nuisance functions than the other two, shows the smallest RMSE even though it has the largest bias under (i) and (ii) (see Table~1 in the main text).
\begin{table}[H]
	\centering
	\resizebox{0.9\textwidth}{!}{\begin{tabular}{rrrr|rrr}
			\hline
			& \multicolumn{3}{c|}{\textbf{Parametric estimation}} &  \multicolumn{3}{c}{\textbf{Nonparametric estimation}}  \\ 
			& $\hat{\psi}^{\text{b}}_{\text{Naive-hazard}}(t)$ & $\hat{\psi}^{\text{b}}_{\text{IF-hazard}}(t)$ & $\hat{\psi}^{\text{b}}(t)$ & $\hat{\psi}^{\text{b}}_{\text{Naive-hazard}}(t)$ & $\hat{\psi}^{\text{b}}_{\text{IF-hazard}}(t)$ & $\hat{\psi}^{\text{b}}(t)$ \\ 
			\hline
			\multicolumn{5}{l}{Scenario (i)} \\
			\hline
			\textbf{Correct} & 
		5.71 & \cellcolor{GreenYellow}5.80 & 7.89 & 3.98 & 6.67 & 8.37 \\ 
			\textbf{Incorrect $\hat{\omega}, \hat{\delta}$} &  
			5.71 &  \cellcolor{GreenYellow}4.70 & 6.76 & 3.84 & 6.68 & 8.55 \\ 
			\textbf{Incorrect $\hat{\pi}, \hat{\mu}_{t}$} & 
			6.25 & \cellcolor{GreenYellow}5.82 & 7.75 & 3.90 & 6.88 & 8.52 \\ 
			\textbf{Incorrect $\hat{\pi}, \hat{\omega}$} & 
			5.61 & \cellcolor{GreenYellow}5.63 & 7.78 & 3.85 & 6.71 & 8.64 \\ 
			\hline
			\multicolumn{5}{l}{Scenario (ii)} \\
			\hline
			\textbf{Correct} &$>10^{10}$ & $>10^{10}$ & \cellcolor{GreenYellow}5.04 & 1.16 & 5.57 & 5.00 \\ 
			\textbf{Incorrect $\hat{\omega}, \hat{\delta}$}  & $>10^{10}$ & $>10^{10}$& \cellcolor{GreenYellow}5.06 & 2.66 & 6.30 & 6.11 \\
			\textbf{Incorrect $\hat{\pi}, \hat{\mu}_{t}$} & $>10^{10}$ & $>10^{10}$&   \cellcolor{GreenYellow}6.27 & 2.70 & 6.35 & 6.38  \\ 
			\textbf{Incorrect $\hat{\pi}, \hat{\omega}$} & $>10^{10}$ & $>10^{10}$ & \cellcolor{GreenYellow}6.17 & 2.67 & 6.24 & 6.20 \\  
			\hline 
			\multicolumn{5}{l}{Scenario (iii)} \\
			\hline
			\textbf{Correct}  & \cellcolor{GreenYellow}1.34 &  \cellcolor{GreenYellow}4.14 & 16.14 & 1.55 & 5.17 & 14.12 \\  
			\textbf{Incorrect $\hat{\omega}, \hat{\delta}$}   &  \cellcolor{GreenYellow}1.34 & \cellcolor{GreenYellow}3.44 & 11.51 & 1.58 & 5.48 & 13.04 \\ 
			\textbf{Incorrect $\hat{\pi}, \hat{\mu}_{t}$} &  2.68 & \cellcolor{GreenYellow}4.41 & 16.58 & 1.28 & 5.39 & 14.53 \\ 
			\textbf{Incorrect $\hat{\pi}, \hat{\omega}$}  &  \cellcolor{GreenYellow}1.23 & \cellcolor{GreenYellow}4.03 & 13.00 & 1.34 & 5.08 & 12.36 \\ 
			\hline
	\end{tabular}}
	\caption{\label{tab:rmse} RMSE of three influence functions under three different scenarios (i)-(iii). In case of parametric estimation, we mark the case when each estimator should be valid and robust against model misspecification.}
\end{table}

\section{Additional application results}
\label{sec:additional_2}

\begin{table}[ht]
\centering
\resizebox{\textwidth}{!}{\begin{tabular}{rlll|lll}
  \hline
 & Control ($Z=0$) & Intervention ($Z=1$) &  & Not screened ($A=0$) & Screened ($A=0$) &  \\ 
Characteristics  & $n=71,578$ & $n=71,848$ & p-value & $n=78,724$ & $n=63,702$ & p-value  \\ 
  \hline
\textbf{Age} (continuous) & 62.61  ( 5.36 ) & 62.56  ( 5.35 ) & 0.0888 & 62.66  ( 5.38 ) & 62.49  ( 5.32 ) & 0.0000 \\ 
 \multicolumn{4}{l|}{\textbf{Age Level}} \\ 
  55-59 yr & 23403  ( 33.16 ) & 24042  ( 33.46 ) &  & 25948  ( 32.96 ) & 21497  ( 33.75 ) &  \\ 
  60-64 yr & 21787  ( 30.87 ) & 22197  ( 30.89 ) &  & 24130  ( 30.65 ) & 19854  ( 31.17 ) &  \\ 
 65-69 yr & 15953  ( 22.60 ) & 16134  ( 22.46 ) &  & 17851  ( 22.68 ) & 14236  ( 22.35 ) &  \\ 
 70-74 yr & 9435  ( 13.37 ) & 9475  ( 13.19 ) & 0.5301 & 10795  ( 13.71 ) & 8115  ( 12.74 ) & 0.0000 \\ 
  \multicolumn{4}{l|}{\textbf{Sex}} \\
  Male & 35326  ( 50.05 ) & 36478  ( 50.77 ) &  & 38472  ( 48.87 ) & 33332  ( 52.32 ) &  \\ 
  Female & 35252  ( 49.95 ) & 35370  ( 49.23 ) & 0.0067 & 40252  ( 51.13 ) & 30370  ( 47.68 ) & 0.0000 \\ 
  \multicolumn{4}{l|}{\textbf{Family History of Any Cancer}} \\
 No & 31292  ( 44.34 ) & 31966  ( 44.49 ) &  & 35109  ( 44.60 ) & 28149  ( 44.19 ) &  \\ 
  Yes & 38991  ( 55.25 ) & 39701  ( 55.26 ) & 0.7599 & 43274  ( 54.97 ) & 35418  ( 55.60 ) & 0.0549 \\ 
  Unknown & 295  ( 0.42 ) & 181  ( 0.25 ) & 0.0000 & 341  ( 0.43 ) & 135  ( 0.21 ) & 0.0000 \\ 
 \multicolumn{4}{l|}{\textbf{Family History of Colorectral Cancer}} \\  
No & 31292  ( 44.34 ) & 31966  ( 44.49 ) &  & 35109  ( 44.60 ) & 28149  ( 44.19 ) &  \\ 
  Yes & 36871  ( 52.24 ) & 37323  ( 51.95 ) & 0.0036 & 40866  ( 51.91 ) & 33328  ( 52.32 ) & 0.0485 \\ 
  Possibly / Unknown & 2415  ( 3.42 ) & 2559  ( 3.56 ) & 0.2491 & 2749  ( 3.49 ) & 2225  ( 3.49 ) & 0.2934 \\ 
  \multicolumn{4}{l|}{\textbf{Colorectal Polyps}} \\
 No & 65291  ( 92.51 ) & 66574  ( 92.66 ) &  & 72769  ( 92.44 ) & 59096  ( 92.77 ) &  \\ 
 Yes & 4637  ( 6.57 ) & 4917  ( 6.84 ) & 0.0646 & 5199  ( 6.60 ) & 4355  ( 6.84 ) & 0.1458 \\ 
  Unknown & 650  ( 0.92 ) & 357  ( 0.50 ) & 0.0000 & 756  ( 0.96 ) & 251  ( 0.39 ) & 0.0000 \\ 
   \multicolumn{4}{l|}{\textbf{Diabetes}}  \\
   No & 64614  ( 91.55 ) & 66015  ( 91.88 ) &  & 71797  ( 91.20 ) & 58832  ( 92.36 ) &  \\ 
  Yes & 5396  ( 7.65 ) & 5554  ( 7.73 ) & 0.7096 & 6268  ( 7.96 ) & 4682  ( 7.35 ) & 0.0000 \\ 
  Unknown & 568  ( 0.80 ) & 279  ( 0.39 ) & 0.0000 & 659  ( 0.84 ) & 188  ( 0.30 ) & 0.0000 \\ 
  \hline
\end{tabular}}
\caption{\label{tab:realtab} Characteristics of the study population. }
\end{table}
Table~\ref{tab:realtab} describes the distribution of the baseline covariates $\bm{X}_{0}$ for the participants, which is corresponding to Table 2 in~\cite{kianian2019causal}. Compared to Table 2 in~\cite{kianian2019causal}, we have a smaller number of participants mainly because we excluded the participants with no history of \textit{all} types of cancer, not only of colorectal cancers. P-values in Table~\ref{tab:realtab} that compare the distribution of the observed covariates $\bm{X}_{0}$ with respect to $Z$ show less evidence of systematic differences in covariates distribution than those for the comparison with respect to $A$. However, we still observe some discrepancy between two arms in sex, family history of colorectal cancer, and colorectal polyps, so we further adjust the distribution of $Z$ by the all baseline covariates in Table~\ref{tab:realtab}. On the other hand, comparison with respect to the actual intervention received ($A$) indicates a significant discrepancy in almost all $\bm{X}_{0}$ between two intervention groups, and we further suspect that there will be unmeasured confounders other than $\bm{X}_{0}$. 

\begin{figure}[ht]
	\centering
	\includegraphics[width=\textwidth]{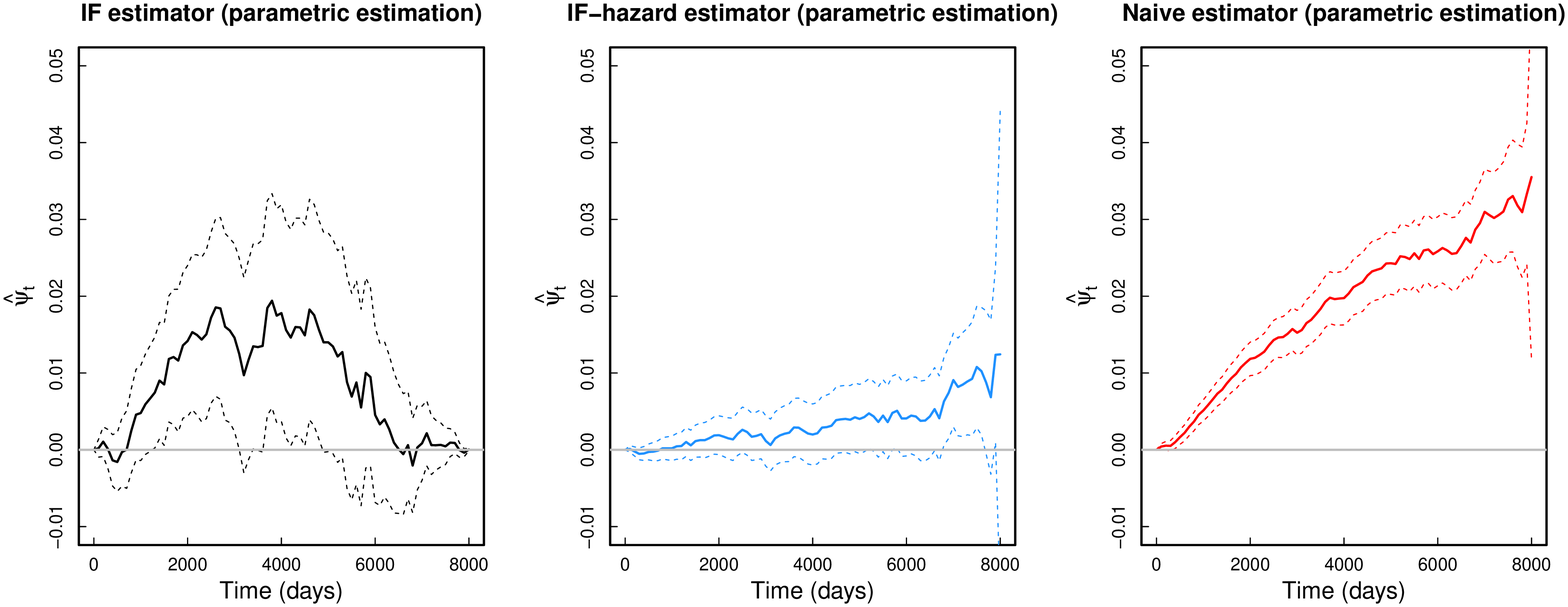}
	\caption{\label{fig:realresult_para} Estimated effect of the screening on the survival probability using three different influence function-based estimators when nuisance functions are parametrically identified. Point-wise confidence intervals (dotted lines) are estimated through $500$ bootstrap samples.}
\end{figure}

Figure~\ref{fig:realresult_para} presents the estimated effect of the screening on the survival probability from the PLCO data when all of the nuisance functions are parametrically identified. For parametric estimation, we used a logistic regression with a linear and additive predictor of the observed covariates for $Y_{t}, R, A$, and $Z$. Results are similar to Figure 2 in the main text when we estimate the nuisance functions nonparametrically instead; but here we have a larger variance of $\hat{\psi}^{\text{b}}_{\text{IF-hazard}} (t)$
, and $\hat{\psi}^{\text{b}}_{\text{Naive-hazard}}(t)$ at the very later time points possibly due to very few observations in the risk set. 


\label{lastpage}

\end{document}